\documentclass[a4paper,onecolumn, 11pt,accepted=2025-02-20]{quantumarticle}
 
\pdfoutput=1
\usepackage{amsmath,amsthm,amsfonts,amssymb}
\usepackage[a4paper,left=2cm,right=2cm,top=2cm,bottom=2cm]{geometry}
\usepackage[english]{babel}
\usepackage[utf8]{inputenc}
\usepackage[T1]{fontenc}
\usepackage[pdftex, pdftitle={Article}, pdfauthor={Author}]{hyperref} 

\pdfoutput=1

\usepackage{graphicx,tikz}
\usetikzlibrary{shapes.geometric}
\usepackage{ifthen}
\usepackage{authblk}
\usepackage{makecell}
\usepackage{tcolorbox}
\usepackage{lipsum}
\usepackage{scalerel}
\tcbuselibrary{skins,breakable}
\usepackage{bm}
\usepackage{algorithm}
\usepackage{algorithmic}
\usepackage[utf8]{inputenc}
\usepackage[T1]{fontenc}

\newboolean{ElectronicVersion}
\setboolean{ElectronicVersion}{true}

\makeatletter
\newtheorem*{rep@theorem}{\rep@title}
\newcommand{\newreptheorem}[2]{%
	\newenvironment{rep#1}[1]{%
		\def\rep@title{#2 \ref*{##1}}%
		\begin{rep@theorem}}%
		{\end{rep@theorem}}}
\makeatother


\usepackage{hyperref}
\hypersetup{
	bookmarksnumbered=true, 
	unicode=false, 
	pdfstartview={FitH}, 
	pdftitle={Exponential speed up for path finding using elfs on multidimensional quantum walks}, 
	pdfauthor={Jianqiang Li and Sebastian Zur}, 
	pdfsubject={}, 
	pdfcreator={}, 
	pdfproducer={}, 
	pdfkeywords={}, 
	pdfnewwindow=true, 
	colorlinks=true, 
	linkcolor=blue, 
	citecolor=blue, 
	filecolor=blue, 
	urlcolor=blue 
}

\newcommand{\eq}[1]{\hyperref[eq:#1]{(\ref*{eq:#1})}}
\renewcommand{\sec}[1]{\hyperref[sec:#1]{Section~\ref*{sec:#1}}}
\newcommand{\thm}[1]{\hyperref[thm:#1]{Theorem~\ref*{thm:#1}}}
\newcommand{\lem}[1]{\hyperref[lem:#1]{Lemma~\ref*{lem:#1}}}
\newcommand{\cor}[1]{\hyperref[cor:#1]{Corollary~\ref*{cor:#1}}}
\newcommand{\app}[1]{\hyperref[app:#1]{Appendix~\ref*{app:#1}}}
\newcommand{\tabl}[1]{\hyperref[tab:#1]{Table~\ref*{tab:#1}}}
\newcommand{\defin}[1]{\hyperref[def:#1]{Definition~\ref*{def:#1}}}
\newcommand{\fig}[1]{\hyperref[fig:#1]{Figure~\ref*{fig:#1}}}
\newcommand{\clm}[1]{\hyperref[clm:#1]{Claim~\ref*{clm:#1}}}
\newcommand{\conj}[1]{\hyperref[conj:#1]{Conjecture~\ref*{conj:#1}}}
\newcommand{\rem}[1]{\hyperref[rem:#1]{Remark~\ref*{rem:#1}}}
\newcommand{\para}[1]{\hyperref[para:#1]{Paragraph~\ref*{para:#1}}}
\newcommand{\exmp}[1]{\hyperref[exmp:#1]{Example~\ref*{exmp:#1}}}
\newcommand{\appx}[1]{\hyperref[appx:#1]{Appendix~\ref*{appx:#1}}}
\newcommand{\fct}[1]{\hyperref[fct:#1]{Fact~\ref*{fct:#1}}}
\newcommand{\alg}[1]{\hyperref[alg:#1]{Algorithm~\ref*{alg:#1}}}
\newcommand{\tab}[1]{\hyperref[tab:#1]{Table~\ref*{tab:#1}}}
\newcommand{\probl}[1]{\hyperref[probl:#1]{Problem~\ref*{probl:#1}}}
\newcommand{\assump}[1]{\hyperref[assump:#1]{Assumption~\ref*{assump:#1}}}
\def\proj#1{\ket{#1}\bra{#1}}

\newcommand{\thmthm}[2]{\hyperref[thm:#1]{Theorem~\ref*{thm:#1}} and~\hyperref[thm:#2]{\ref*{thm:#2}}}
\newcommand{\lemlem}[2]{\hyperref[lem:#1]{Lemma~\ref*{lem:#1}} and~\hyperref[lem:#2]{\ref*{lem:#2}}}

{\endtcolorbox}

\newtheorem{theorem}{Theorem}[section]
\newtheorem{lemma}[theorem]{Lemma}

\newtheorem{corollary}[theorem]{Corollary}

\newtheorem{claim}[theorem]{Claim}

\newtheorem{definition}[theorem]{Definition}
\newtheorem{problemo}[theorem]{Problem}
\newtheorem{assumpt}[theorem]{Assumption}

\usepackage{color}
\definecolor{darkgreen}{rgb}{0,.5,0}
\definecolor{darkred}{rgb}{.7,.3,.3}
\definecolor{deepblue}{rgb}{0,.1,.7}

\newif\iflongversion
\newif\ifshortversion
\newif\ifediting
\editingtrue 

\def\ket#1{{\lvert}#1\rangle}
\def\bra#1{{\langle}#1\rvert}

\renewcommand{\(}{\left(}
\renewcommand{\)}{\right)}

\def\abs#1{\left| #1 \right|}

\def\norm#1{\left\| #1 \right\|}

\def\proj#1{\ket{#1}\bra{#1}}
\def\w{{\sf w}}

\def\p{{\sf p}}
\def\alt{{\sf alt}}

\usepackage{comment}

\usepackage{braket}
\usepackage{amssymb,bbm}
\usepackage{amsmath}

\newcommand{\f}{\theta}

\usepackage{latexsym}
\usepackage{amsthm}
\usepackage[capitalise]{cleveref}

\usepackage{comment}
\usepackage{mathdots}
\usepackage{url}
\usepackage{algorithm}
\usepackage{algorithmic}
\usepackage{systeme}
\usepackage{appendix}[toc,page]
\hypersetup{pdfpagemode=UseNone}

\theoremstyle{definition}
{

}

\mathchardef\mhyphen="2D

\newcommand{\brakett}[2]{\mbox{$\langle #1  | #2 \rangle$}}

\def\({\left(}
\def\){\right)}

\usepackage{graphicx}
\definecolor{greenn}{rgb}{0,0.8,0.2}
\definecolor{bluue}{rgb}{0.3,0,0.7}

\usepackage{ifdraft}
\ifdraft{\newcommand{\authnote}[3]{{\color{#3}{\text{ \bf #1:}} #2}}}{\newcommand{\authnote}[3]{}}

\newcommand{\N}{\Gamma}
\newcommand{\E}{\overrightarrow{E}}

\title{Multidimensional Electrical Networks and their Application to Exponential Speedups for Graph Problems}

\author{Jianqiang Li }
\affiliation{Department of Computer Science and Engineering, Pennsylvania State University, PA, \& \\ Department of Computer Science, Rice University, Huston, TX} \thanks{}
\email{ jl567@rice.edu}

\author{Sebastian Zur}
\affiliation{CWI \& QuSoft, the Netherlands}
\email{saz@cwi.nl}

\begin{document}
	
\maketitle

\newcommand{\st}{$s$-$t$ }
\newcommand{\palt}{\p^{\sf alt} }
\newcommand{\falt}{\f^{\sf alt} }
\newcommand{\du}{{\sf w}}
\newcommand{\B}{{\cal B}}
\newcommand{\A}{{\cal A}}
\newcommand{\Balt}{{\cal B}^{\sf alt}}
\begin{abstract}

 
 %

Recently, Apers and Piddock [TQC '23] strengthened the natural connection between quantum walks and electrical networks by considering Kirchhoff's Law and Ohm's Law.
In this work, we develop a new multidimensional electrical network by defining Alternative Kirchhoff's Law and Alternative Ohm's Law based on the multidimensional quantum walk framework by Jeffery and Zur [STOC '23].  In analogy to the connection between the (edge-vertex) incidence matrix of a graph and Kirchhoff's Law and Ohm's Law in an electrical network, we rebuild the connection between the alternative incidence matrix and Alternative Kirchhoff's Law and Alternative Ohm's Law.
This multidimensional electrical network framework enables generating an alternative electrical flow over the edges on graphs, which has the potential to be applied to a broader range of graph problems, benefiting both quantum and classical algorithm design.
 

We first use this framework to generate quantum alternative electrical flow states and use it to find a marked vertex in one-dimensional random hierarchical graphs as defined by Balasubramanian, Li, and Harrow [arXiv '23].  In this work, they generalised the well known exponential quantum-classical separation of the welded tree graph by Childs, Cleve, Deotto, Farhi, Gutmann, and Spielman [STOC '03] to random hierarchical graphs. Our result partially recovers their results with an arguably simpler analysis. 


Furthermore, this framework also allows us to demonstrate an exponential quantum speedup for the pathfinding problem in a type of regular graph, which we name the welded tree circuit graph. The exponential quantum advantage is obtained by efficiently generating quantum alternative electrical flow states and then sampling from them to find an 
$s$-$t$ path in the welded tree circuit graph. By comparison, Li [arXiv '23] constructed a non-regular graph based on welded trees and used the degree information to achieve a similar speedup.

\end{abstract}


\section{Introduction}

The connection between random walks, electrical networks, and Laplacian linear systems is of great significance in multiple aspects of theoretical computer science. When the edges of a graph are seen as resistors, the effective resistance between two vertices in an electrical network plays an important role in analysing the behavior of random walks (or more general Markov chains) on undirected graphs, such as their hitting time and cover time \cite{doyle1984RandomWalksAndElectriNetw,lovasz1993random}. By restating Kirchhoff's Law and Ohm's Law in electrical networks (graphs) in terms of the graph's Laplacian matrix and incidence matrix, one can compute the effective resistance and electrical flow between two vertices by solving a Laplacian linear system \cite{vishnoi2013lx}. In addition, the effective resistance and electrical flow used in electrical networks have been used to design new graph algorithms, such as graph sparsification via effective resistances \cite{spielman2008graph}, computing max flow via electrical flow \cite{christiano2011electrical}, and many more ~\cite{madry2013navigating,gao2023fully}. Furthermore, the mixing time of a random walk (Markov chain) on a connected graph is related to the smallest nonzero eigenvalue of the Laplacian matrix of the underlying graph \cite{spielman19spectral}.   
%





The connection between (discrete) quantum walks, electrical networks, and Laplacian linear systems is less understood and has evolved separately. Specifically, the connection between quantum walks and electrical networks was first established by \cite{belovs2013ElectricWalks}, where this connection was used to derive and analyse a phase estimation algorithm with the goal of detecting the existence of a marked vertex in a graph. Not until recently, this connection has been strengthened by Apers and Piddock \cite{piddock2019electricfind,apers2022elfs}  by considering Kirchhoff's Law and Ohm's Law. They showed that for the quantum walk operator based on the electrical network framework, if the phase value returned by the phase estimation algorithm from \cite{belovs2013ElectricWalks} is ``$0$'', indicating that there is a marked vertex, then the resulting state is actually a quantum state representing the electrical flow between the starting vertex and the marked vertex.  Separately, \cite{wang2017efficient} provided a quantum algorithm based on the Quantum Linear System (QLS) \cite{gilyen2018QSingValTransf} algorithm for solving a Laplacian related linear system in the analysis of large electrical networks, such as computing their electrical flow and effective resistance between two vertices. However, there is currently little known about the connection between quantum walks and the eigenvalues of the Laplacian matrix of the underlying graph. 

There are multiple examples of quantum walk based algorithms that can solve certain graph problems exponentially faster than any classical algorithm with oracle access to the graph, but currently none of them are based on electrical networks. For example, continuous quantum walks have been used to exhibit an exponential quantum-classical separation to solve the welded tree problem \cite{childs2003ExpSpeedupQW}. A welded tree graph consists of two full binary trees of depth $n$ and the leaves of both trees are connected via two disjoint perfect matching. There are a total number of $2^{n+1}-2$ vertices. Given an adjacency list oracle $O_G$ to the welded tree graphs $G$ and the name of one of the roots $s$, the goal of the welded tree problem is to output the name of the other root $t$. This exponential separation result of the welded tree problem has recently been recovered (and improved) by applying the multidimensional quantum walk framework \cite{jeffery2023multidimensional}
or using the conceptually simpler coined quantum walk \cite{li2024recovering}.
In addition, using a more refined analysis of the continuous quantum walk algorithm, \cite{balasubramanian2023exponential} generalised this exponential separation result on the welded tree graph to certain random hierarchical graphs where the goal is once again to find some special vertex $t$ given an initial vertex $s$. The idea of finding a marked vertex in certain graphs has also been used to show the exponential advantage of adiabatic quantum computation with no sign problem in \cite{gilyen2021sub}, whose graph is also related to the welded tree graph. Another graph constructed from welded tree graphs has also allowed for an exponential quantum-classical separation for a type of graph property testing problem \cite{bendavid2020symmetries}. Using the continuous quantum walk algorithm for the welded tree problem as a subroutine, \cite{li2023exponential} showed that exponential separation extends beyond finding a marked vertex and that this separation can also be shown for finding a \st path on the welded tree path graphs, which is constructed from welded tree graphs. 


Among the three types of quantum walk design paradigms mentioned that achieve this exponential speedup (continuous quantum walks, coined quantum walks, and multidimensional quantum walks), the multidimensional quantum walk has the closest relationship with the quantum walk operator in an electrical network. However, it is far from obvious how one can combine the design paradigm of multidimensional quantum walks, with Kirchhoff's Law and Ohm's Law in an electrical network.

\subsection{The Multidimensional Electrical Network}
In this work, we take the first steps in establishing this connection, which we name the multidimensional electrical network framework, achieved through generalising Kirchhoff's Law as Alternative Kirchhoff's Law and Ohm's Law as Alternative Ohm's Law. Roughly speaking, for each vertex other than the source and sink, Kirchhoff's Law forces the electrical flow to be orthogonal to a single vector, whereas Alternative Kirchhoff's Law forces the electrical flow to be orthogonal to a potentially larger subspace. Moreover, instead of associating each vertex with a single (unique) potential value for each vertex as in Ohm's Law, in Alternative Ohm's Law we associate (possibly distinct, but unique) potential values $\p_{(u,v)}$, $\p_{(v,u)}$ to each edge $(u,v)$. 
The ideas behind the definition of Alternative Kirchhoff's Law and Alternative Ohm's Law in the multidimensional electrical network are derived from the quantum walk based on an electrical network \cite{belovs2013ElectricWalks,piddock2019electricfind,apers2022elfs} and the alternative neighbourhood technique introduced by the multidimensional quantum walk \cite{jeffery2023multidimensional}. 



\paragraph{Multidimensional quantum walks and alternative electrical flows}
Similarly as in the quantum walk frameworks that are based on electrical networks, we model a graph $G=(V,E,\w)$ as an electrical network with each edge being assigned a positive weight $\w_{u,v}$, i.e. its conductance. This weight assignment gives rise to a weighted superposition of neighbours in the neighbourhood $\N(u)$ of any vertex $u$, normalised by the quantity $\w_u$, which is known as the \textit{star state} of $u$:
$$\frac{1}{\sqrt{\du_u}} \sum_{v \in \N(u)}\sqrt{\w_{u,v}}\ket{u,v}.$$
This state can be thought of as a quantum encoding of the probability to move from a vertex $u$ to each neighbour $v$. Since we build on top of the multidimensional quantum walk framework, we adopt their definitions for technical reasons, which are slightly different from what is convention in the literature. We discuss this difference in depth in \sec{prelim}, but the most important difference is that the edge set $E$ is described by a directed edge set $\E$ and that the definition of a star state depends on the boolean variable $\Delta_{u,v}$, which is equal to $0$ if $(u,v) \in\E$ and $1$ if $(v,u) \in \E$:
\begin{equation}\label{eq:star-state-intro}
    \ket{\psi_u} := \frac{1}{\sqrt{\du_u}} \sum_{v\in \N(u)}(-1)^{\Delta_{u,v}}\sqrt{\w_{u,v}}\ket{u,v}.
\end{equation}

We study \st flows on our graph $G$ between any two of its vertices $s$ and $t$. An \st flow assigns an amount of value of flow along each edge of the graph $G$. Kirchhoff's Law states that any \st flow is conserved at every vertex $u \in V \backslash \{s,t\}$, meaning that the amount of flow entering any vertex $u$ is equal to the amount of flow exiting $u$. More specifically, we are interested in the \st \textit{electrical flow} $\f$, which is the ``smallest'' \st flow, meaning it has the smallest energy out of all valid \st flows satisfying Kirchhoff's Law:
$${\cal E} (\f) :=\sum_{(u,v)\in\E}\frac{\f_{u,v}^2}{\w_{u,v}}.$$
The energy of the \st electrical flow $\f$ is called the \textit{effective resistance} ${\cal R}_{s,t}$ and this flow gives rise to the normalised electrical flow state $\ket{\f}$:
$$ \ket{\f} := \frac{1}{\sqrt{2{\cal R}_{s,t}}}\sum_{(u,v) \in \E}\frac{\f_{u,v}}{\sqrt{\w_{u,v}}}\left(\ket{u,v} + \ket{v,u}\right).$$

By viewing $G$ as an electrical network, this electrical flow precisely captures the electrical dynamics of how one unit of current send from $s$ to $t$ would traverse the electrical network $G$. Kirchhoff's Law states that this \st electrical flow is conserved at every vertex $u \in V \backslash \{s,t\}$, meaning that the amount of flow $\f$ coming into $u$ is equal to the amount of flow exiting $u$. This law can be equivalently read in terms of $\ket{\psi_u}$ and $\ket{\f}$, in which case it states that for every vertex $u \in V \backslash \{s,t\}$ we require 
$$\brakett{\psi_u}{\theta} = 0.$$ 

The \textit{quantum walk operator} $U_{\cal AB}=(2\Pi_{\cal A}-I)(2\Pi_{\cal B}-I)$ in \cite{belovs2013ElectricWalks} consists of a reflection around two spaces: the antisymmetric subspace ${\cal A}$ and the span of (almost all) star states:
$${\cal B} := \mathrm{span}\{\ket{\psi_u}: u \in V \backslash \{s,t\}\}.$$ 

The \st electrical flow $\f$ is special with regards to this quantum walk operator $U_{\cal AB}$, as its flow state $\ket{\f}$ lives in the $+1$-eigenspace of $U$. Moreover, it can also be written as a linear combination of projected star states $(I - \Pi_{\cal A})\ket{\psi_u}$ for $u \in V$. The coefficients in this linear combination are precisely the potentials $\p_u$ given by the \textit{potential vector} $\p$ corresponding to the \st electrical flow $\f$ such that together they satisfy Ohm's Law: $\p_{u}-\p_{v} = \f_{u,v}/\w_{u,v}$. By combining all these properties, \cite{piddock2019electricfind,apers2022elfs} showed that with the use of phase estimation on the quantum walk operator $U_{\cal AB}$, one can approximate $\ket{\f}$ and measure it, allowing one to (approximately) sample from the \st electrical flow $\f$, where an edge is obtained with a probability equal to the relative amount of flow that traverses this edge.

The exact complexity of approximating the \st electrical flow state in an electrical network is given in \cor{qwflows}, but it depends (apart from the approximation error) on two factors: the effective resistance ${\cal R}_{s,t}$ and the norm of the potential vector $\ket{\p}$, defined as
$$  \ket{\p}= \sqrt{\frac{2}{{\cal R}_{s,t}}}\sum_{u\in V\backslash \{ s,t\}} \p_u\sqrt{\du_u}\ket{\psi_u}.$$


The multidimensional quantum walk framework generalises the previously known quantum walks based on electrical networks and works by running phase estimation on the modified quantum walk operator $U_{\A\Balt}$, which reflects around the larger space that contains ${\B}$:
$${\Balt} := \mathrm{span}\{\mathrm{span}(\Psi_{\star}(u)): u \in V \backslash \{s,t\}\}.$$

By associating each vertex $u$ with a subspace $\mathrm{span}(\Psi_{\star}(u))$, the multidimensional quantum walk has shown advantages in solving certain problems, such as the welded tree problem and the $k$-distinctness problem in \cite{jeffery2023multidimensional} faster than the quantum walk. The algorithm works by applying phase estimation of the multidimensional quantum walk operator $U_{\A\Balt}$ and checking whether the returned phase value is ``$0$''. Especially for the welded tree problem, this phase value being ``$0$'' indicating that the root $t$ is ``marked'', which allows them to infer the name of $t$, which recovers the exponential quantum-classical separation for finding a marked vertex in the welded tree graph using a continuous quantum walk \cite{childs2003ExpSpeedupQW}. 

The reason for why multidimensional quantum walks achieve this potentially exponential speedup, whereas other discrete quantum walks can achieve up to a quadratic speedup compared to the underlying classical random walk, has to do with cost of calling $U_{\cal AB}$. The cost of applying this quantum walk operator, and hence the cost of this phase estimation procedure, relies on the cost of generating the star state $\ket{\psi_u}$. In \cite{jeffery2023multidimensional}, the authors deal with the case where it might be computationally costly to generate $\ket{\psi_u}$, but where the algorithm knows that $\ket{\psi_u}$ is one of a small set of easily preparable states $\Psi_{\star}(u) = \{\ket{\psi_{u,1}},\ket{\psi_{u,2}},\dots\}$, known as the \textit{alternative neighbourhoods} for $u$. To give some intuition as to why this might be useful, let's look at the definition of $\ket{\psi_u}$ in \eq{star-state-intro} again. If different neighbours $v_1,v_2 \in \N(u)$ have different amplitudes in $\ket{\psi_u}$, perhaps due to different weights/conductances $\w_{u,v_1},\w_{u,v_2}$, or perhaps a different sign due to $\Delta_{u,v_1},\Delta_{u,v_2}$, than to generate the star state $\ket{\psi_u}$, it would be necessary to distinguish the neighbours $v_1,v_2$ from each other. This might be computationally costly if we can access $G$ via an adjacency list oracle, where upon querying $u$ we learn $\N(u)$. Such problems can however also arise in a much more nuanced fashion outside of the oracle model, as exhibited in the algorithm tackling the $k$-distinctness problem in \cite{jeffery2023multidimensional}. Therefore by generating a reflection around $\Psi_{\star}(u)$, instead of $\ket{\psi_u}$, the cost of each step of the phase estimation algorithm is potentially drastically reduced.

To extend the previously discussed relationship between electrical networks and quantum walks to this new quantum walk operator $U_{\A\Balt}$, we introduce the \st alternative electrical flow $\falt$. This is again the ``smallest'' \st flow in terms of energy that satisfies Alternative Kirchhoff's Law, which requires $\ket{\falt}$ to be orthogonal to all of $\Psi_{\star}(u)$ instead of only to $\ket{\psi_u} \in \Psi_{\star}(u)$, ensuring that $\ket{\falt}$ lives in the $+1$-eigenspace of $U_{\A\Balt}$. This $\ket{\falt}$, whose energy is the \textit{alternative effective resistance} ${\cal R}_{s,t}^{\alt}$, can also be written as a linear combination of projected alternative neighbourhoods $(I - \Pi_{\cal A})\ket{\psi_{u,i}}$ for $u \in V$ and $\ket{\psi_{u,i}} \in \Psi_{\star}(u)$. As we have seen in the non-alternative neighbourhood case, this gives rise to an alternative potential vector $\palt$ and alternative potential state $\ket{\palt}$, this time not acting on vertices, but on edges instead. This $\palt$ is related to $\falt$ through Alternative Ohm's Law, which states that $\palt_{u,v}-\palt_{v,u}=\falt_{u,v}/\w_{u,v}$ and $\ket{\palt} \in \Balt$. By a similar analysis as in regular electrical networks, we show that with phase estimation on the quantum walk operator $U_{\A\Balt}$ we can approximate $\ket{\falt}$, allowing one to (approximately) sample from the \st alternative electrical flow $\falt$.

\begin{theorem}\label{thm:intromqwflows}
  Let $\Psi_{\star}$ be a collection of alternative neighborhoods on a network $G = (V,E,\w)$ and let $U_{\A\Balt}$ be the quantum walk operator with respect to $\Psi_{\star}$ as defined in \cref{eq:walk-alt}. Then by performing phase estimation on the initial state $\ket{\psi_s^+}$ with the operator $U_{\A\Balt}$ and precision $O\left(\frac{\epsilon^2}{\sqrt{{\cal R}_{s,t}^{\alt} \du_s}\norm{\ket{\palt}}}\right)$, the phase estimation algorithm outputs ``$0$'' with probability $\Theta\left(\frac{1}{{\cal R}_{s,t}^{\alt} \du_s}\right)$, leaving a state $\ket{\f'}$ satisfying
  $$ \frac{1}{2}\norm{\proj{\f'} - \proj{\falt}}_1 \leq \epsilon. $$
\end{theorem}




As can be seen in the above theorem, the complexity of generating the \st alternative electrical flow state in an multidimensional electrical network depends on the alternative versions of the effective resistance, ${\cal R}_{s,t}^{\alt}$, and (the norm of) the alternative potential vector $$
\ket{\palt}=\sqrt{\frac{2}{{\cal R}(\falt)}} \sum_{u \in V \setminus \{s\}}\sum_{v \in \N(u)}(-1)^{\Delta_{u,v}} \palt_{u,v} \sqrt{\w_{u,v}} \ket{u,v}.$$

With this new multidimensional electrical network framework, more specifically due to \thm{intromqwflows}, we show how our quantum algorithms can provide exponential quantum-classical separations for certain graph problems, such as finding a marked vertex and finding an \st path in some special graphs related to the welded tree graph. On the other hand, all previous known quantum algorithms based on electrical networks can provide only a quadratic speedup.

Our algorithms work by approximately sampling from the alternative electrical flow $\falt$, meaning we approximate the state $\ket{\falt}$ using \thm{intromqwflows} and then measure it to obtain an edge with probability scaling with the amount of alternative flow that traverses this edge. In both of these applications, the way that $\Psi_{\star}$ is extended is fairly natural, as it will be constructed from Fourier basis states.


\paragraph{Alternative incidence matrix and alternative electrical flow}
These new alternative electrical network laws and definitions may seem constructed in an ad-hoc fashion to fit with the analysis in \cite{piddock2019electricfind,apers2022elfs}, but we give proof that these are in fact natural definitions. It is well known in electrical network theory \cite{vishnoi2013lx} that both Kirchhoff's Law as well as Ohm's Law can be phrased as linear equations involving the edge-vertex incidence matrix $B$, whose entries contain the square root of the weights $\w_{u,v}$. These linear relations are useful to show important physical properties of the \st electrical flow $\f$ and its potential vector $\p$, such as their existence and the fact that $\p_s$ is equal to the energy of $\f$, also known as the effective resistance ${\cal R}_{s,t}$. By extending this incidence matrix $B$ in a natural fashion to also incorporate the alternative neighbourhoods $\Psi_{\star}$, we obtain the alternative incidence matrix $B_{\alt}$. For the actual technical definition of $B_{\alt}$, see \defin{incidence-alt}, but for now $B_{\alt}$ can be thought of as a generalised version of $B$. This matrix $B_{\alt}$ can then be substituted for $B$ in the linear equations that correspond to Kirchhoff Law and Ohm's Law, to recover Alternative Kirchhoff's Law and Alternative Ohm's Law respectively. As is done with regular electrical networks, we apply these linear equations to prove the existence of the \st alternative electrical flow $\falt$ and its alternative potential vector $\palt$, as well as the fact that the potential $\p_{s,u}$ along each edge adjacent to $s$ is equal to the alternative effective resistance ${\cal R}^{\alt}_{s,t}$. To summarise, we prove the following two results:

\begin{theorem}\label{thm:introelecflow-alt}
  Let $\falt$ be the \st alternative electrical flow in an electrical network $G = (V,E,\w)$ with respect to a collection of alternative neighbourhoods $\Psi_{\star}$. Let $B_{\alt}$ be the alternative incidence matrix of $G$. Then $W\falt$ is given by 
  \begin{equation*}
    W\falt = B_{\alt}^{T+}({\sf e}_s - {\sf e}_t),
  \end{equation*}
  where $W \in \mathbb{C}^{\E \times \E}$ is the diagonal matrix with entries $W_{(u,v),(u,v)} = 1/\sqrt{\w_{u,v}}$.
\end{theorem}

\begin{theorem}\label{thm:intropotential-alt}
   Let $\falt$ be the \st alternative electrical flow in an electrical network $G = (V,E,\w)$ with respect to a collection of alternative neighbourhoods $\Psi_{\star}$. Then there exists a unique alternative potential vector $\palt$ satisfying Alternative Ohm's Law such that $\palt_{s,u} = {\cal R}_{s,t}^{\alt}$ and $\palt_{t,v} = 0$ for each $u \in \N(s)$ and $v \in \N(t)$.
\end{theorem}

The perspective of the alternative incidence matrix of the multidimensional electrical network provides a new way to generate the alternative electrical flow by solving the alternative incidence linear system as indicated in \thm{introelecflow-alt}.

Classically, solving such an alternative incidence linear system takes a polynomial time in the size of the graph to get the alternative electrical state flow vector. The electrical flow obtained by solving the related Laplacian linear system has recently played an important role in designing new graph algorithms for max-flow problems \cite{christiano2011electrical,madry2013navigating,gao2023fully}.  It is conceivable that our alternative electrical flow could potentially be applied to design new graph algorithms for certain graph problems.

Quantumly, the complexity of solving the linear system to obtain the quantum alternative electrical flow state depends on the condition number of the alternative incidence matrix on a bounded degree graph \cite{gilyen2018QSingValTransf}. Although in this work we focus on using multidimensional quantum walk to generate the \st alternative electrical flow state $\ket{\falt}$ for our applications, this new way of generating $\ket{\falt}$ through the QLS algorithm could lead to new applications of the multidimensional electrical network for more types of problems. 


\subsection{Applications}

\paragraph{Finding a marked vertex in random hierarchical graphs}

We first apply this multidimensional quantum electrical network framework to one-dimensional random hierarchical graphs with nodes $S_0,S_1,\cdots, S_n$ as defined in \cite{balasubramanian2023exponential}. Given the initial vertex $s$, which is the unique element in $S_0$, the goal is to transverse the exponentially large (in $n$) one-dimensional random hierarchical graph to find the vertex $t$, the unique element in $S_n$. 

It has been shown that the continuous quantum walk approach can provide an exponential speedup (in $n$) in solving this problem compared to any classical algorithm with certain additional assumptions on the structures of the one-dimensional random hierarchical graph. In this paper, after applying minor differences to the assumptions regarding the structure of the graph, we show that multidimensional quantum walks also solve this problem in polynomial time by sampling from the alternative electrical flow state.

Instead of applying \thm{intromqwflows} directly, we describe and analyse our algorithm explicitly to provide more intuition on how multidimensional quantum walks can be used to approximate the alternative electrical flow. 

Interestingly, in this one-dimensional setting, the alternative electrical flow with respect to the multidimensional quantum walk operator on a one-dimensional unweighted random hierarchical graph matches the actual electrical flow of a weighted electrical network. We apply this result to the welded tree graph, which is an example of a one-dimensional random hierarchical graph, to provide an alternative quantum algorithm that solves the problem in polynomial time, recovering the exponential separation from \cite{childs2003ExpSpeedupQW}.

In the one-dimensional random hierarchical unweighted graph with an adjacency list oracle access of $G$, it is computationally costly to generate $\ket{\psi_u}$, but where the algorithm knows that $\ket{\psi_u}$ is one of a small set of easily preparable states $\Psi_{\star}(u) = \{\ket{\psi_{u,1}},\ket{\psi_{u,2}},\dots\}$, known as the \textit{alternative Fourier neighbourhoods} for $u$ (see \defin{fourier}).  With this multidimensional quantum walk operator, we show that the cost of generating the quantum alternative electrical flow state is a polynomial, while the cost of generating the quantum electrical flow state is exponentially large. 

Compared to the technical analysis from \cite{balasubramanian2023exponential} used for the continuous quantum walk approach, which contains more concepts from physics, our analysis is more suited for a computer science audience and arguably simpler.

\paragraph{Finding an \st path}

While there are known efficient quantum algorithms \cite{childs2003ExpSpeedupQW,jeffery2023multidimensional,li2024recovering} for finding the marked vertex in the welded tree graph, the problem of finding an \st path in the welded tree graph is notably difficult and still an open problem in the field of quantum query complexity \cite{aaronson2021open}. In addition, the problem of finding an \st path in isogeny graphs (a class of expander graphs) is assumed to be hard even for quantum algorithms, as the security of many isogeny-based cryptosystems are based on the hardness of solving this problem \cite{sean2018, Charles2009, costache2019ramanujan,wesolowski2022supersingular}.


Instead of finding an \st path in the welded tree graph or isogeny graphs,  \cite{li2023exponential} recently constructed a nonregular graph and exhibits an exponential quantum-classical separation in the context of pathfinding problems. The nonregular graph associates $n$ different welded tree graphs with an \st path of length $n$. This quantum algorithm uses the polynomial-time continuous quantum walk algorithm from \cite{childs2003ExpSpeedupQW} as a subroutine to output the \st path and the quantum algorithm relies heavily on the constructed graph being non-regular. This degree information in some way propagates the algorithm into the direction of $t$, which makes this technique infeasible for less structured graphs, such as isogeny graphs where every vertex has the same degree \cite{jao2011towards}.


In this work, we apply our new multidimensional electrical network framework to show an exponential quantum-classical separation to find an \st path in an unweighted regular graph. We construct a family of regular graphs, that we name welded tree circuit graphs (see \fig{weldedcircuitgraphsimple}), meaning that this graph has exponentially many vertices in $n$. 
Our framework allows us to approximate the alternative electrical flow state $\ket{\falt}$ in polynomial time for these types of graph. We then show that the overlap between each edge of an explicit \st path and the alternative electrical flow is at least an inverse polynomial. Therefore, by sampling from a polynomial number of copies of $\ket{\falt}$, we can obtain a polynomial sized subgraph that contains an \st path. A classical algorithm such as breadth first search or depth first search can then be used to traverse the subgraph and output this \st path. In \sec{path}, we explicitly compute what the \st alternative electrical flow looks like in the welded tree circuit graph. Unlike how the alternative electrical flow generated in one-dimensional unweighted random hierarchical networks matches the actual electrical flow in the electrical network, the alternative electrical flow in our pathfinding example is significantly different from any `real' electrical flow.

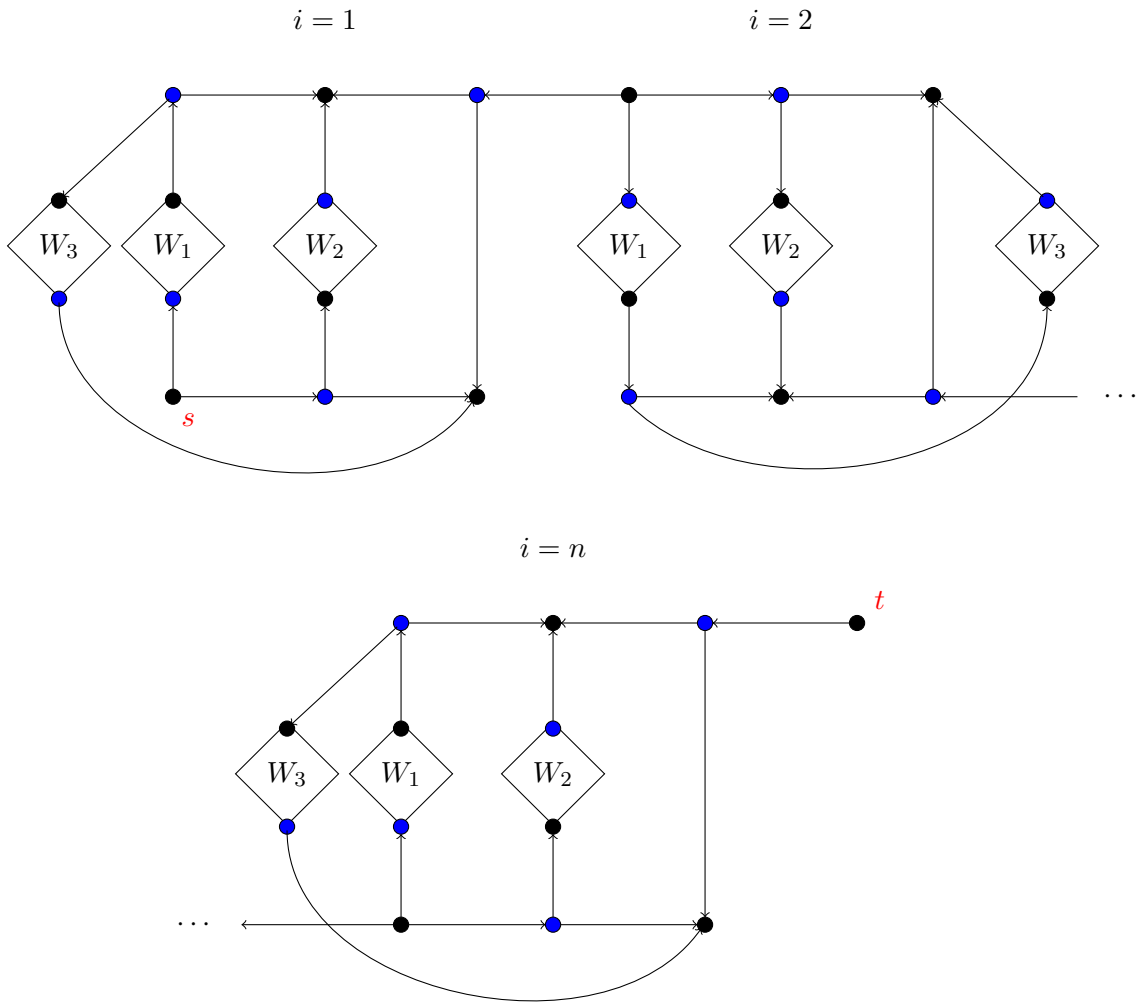
\begin{figure}[!ht]
\centering
\begin{tikzpicture}

\node at (2,5) {$i=1$};
\filldraw (0,0) circle (.1);     \node at (0.2,-0.3) {{\color{red}$s$}};
\filldraw[fill=blue] (2,0) circle (.1);          \node at (2,-0.3) {};
\filldraw (4,0) circle (.1);     \node at (4.6,0) {};
\filldraw[fill=blue] (0,4) circle (.1);          \node at (0,4.3) {};
\filldraw (2,4) circle (.1);     \node at (2,4.3) {};
\filldraw[fill=blue] (4,4) circle (.1);          \node at (4,4.3) {};

\node[diamond,
  draw = black,
  minimum width = 1cm,
  minimum height = 1.3cm] (d) at (0,2) {$W_1$};
\node[diamond,
  draw = black,
  minimum width = 1cm,
  minimum height = 1.3cm] (d) at (2,2) {$W_2$};
\node[diamond,
  draw = black,
  minimum width = 1cm,
  minimum height = 1.3cm] (d) at (-1.5,2) {$W_3$};

\filldraw[fill=blue] (0,1.3) circle (.1);         \node at (0.7,1.3) {};
\filldraw (2,1.3) circle (.1);    \node at (3,1.3) {};
\filldraw (0,2.6) circle (.1);    \node at (1,2.6) {};
\filldraw[fill=blue] (2,2.6) circle (.1);         \node at (2.9,2.6) {};
\filldraw[fill=blue] (-1.5,1.3) circle (.1);        \node at (-2.6,1.3) {};
\filldraw (-1.5,2.6) circle (.1);  \node at (-2.6,2.6) {};

\draw[->] (0.1,0)--(1.9,0);       \node at (1,0.3) {};
\draw[->] (0,0.1)--(0,1.2);       \node at (0.2,0.7) {};
\draw[->] (2.1,0)--(3.9,0);       \node at (3,0.3) {};
\draw[->] (2,0.1)--(2,1.2);       \node at (2.2,0.7) {};
\draw[->] (4,3.9)--(4,0.1);       \node at (4.3,2) {};
\draw[->] (0,2.7)--(0,3.9);       \node at (0.2,3.3) {};
\draw[->] (2,2.7)--(2,3.9);       \node at (2.2,3.3) {};
\draw[->] (-.05,3.95)--(-1.45,2.65);  \node at (-1,3.5) {};
\draw[->] (0.1,4)--(1.9,4);       \node at (1,4.3) {};
\draw[->] (3.9,4)--(2.1,4);       \node at (3,4.3) {};
\draw[->] (5.9,4)--(4.1,4);       \node at (5,4.3) {};
\draw[->] (-1.5,1.25) to[out=-90,in=-125] (3.95,-0.05); \node at (3.5,-1) {};

\node at (8,5) {$i=2$};
\filldraw[fill=blue] (6,0) circle (.1);     \node at (5.7,-0.3) {};
\filldraw (8,0) circle (.1);          \node at (8,-0.3) {};
\filldraw[fill=blue] (10,0) circle (.1);     \node at (10.6,0.3) {};
\filldraw (6,4) circle (.1);          \node at (6,4.3) {};
\filldraw[fill=blue] (8,4) circle (.1);     \node at (8,4.3) {};
\filldraw (10,4) circle (.1);          \node at (10,4.3) {};

\node[diamond,
  draw = black,
  minimum width = 1cm,
  minimum height = 1.3cm] (d) at (6,2) {$W_1$};
\node[diamond,
  draw = black,
  minimum width = 1cm,
  minimum height = 1.3cm] (d) at (8,2) {$W_2$};
\node[diamond,
  draw = black,
  minimum width = 1cm,
  minimum height = 1.3cm] (d) at (11.5,2) {$W_3$};

\filldraw (6,1.3) circle (.1);         \node at (7,1.3) {};
\filldraw[fill=blue] (8,1.3) circle (.1);    \node at (8.8,1.3) {};
\filldraw[fill=blue] (6,2.6) circle (.1);    \node at (6.7,2.6) {};
\filldraw (8,2.6) circle (.1);         \node at (9,2.6) {};
\filldraw (11.5,1.3) circle (.1);        \node at (12.6,1.3) {};
\filldraw[fill=blue] (11.5,2.6) circle (.1);  \node at (12.3,2.6) {};

\draw[->] (6.1,0)--(7.9,0);     \node at (7,0.3) {};
\draw[->] (6,1.2)--(6,0.1);     \node at (6.2,0.7) {};
\draw[->] (9.9,0)--(8.1,0);     \node at (9,0.3) {};
\draw[->] (8,1.2)--(8,0.1);     \node at (8.3,0.7) {};
\draw[->] (10,0.1)--(10,3.9);    \node at (10.3,2) {};
\draw[->] (6,3.9)--(6,2.7);     \node at (6.2,3.3) {};
\draw[->] (8,3.9)--(8,2.7);     \node at (8.2,3.3) {};
\draw[->] (11.45,2.65)--(10.05,3.95);\node at (11,3.5) {};
\draw[->] (6.1,4)--(7.9,4);     \node at (7,4.3) {};
\draw[->] (8.1,4)--(9.9,4);     \node at (9,4.3) {};
\draw[->] (11.9,0)--(10.1,0);    \node at (11.5,-0.3) {};
\draw[->] (6,-0.1) to[out=-45,in=-90] (11.5,1.2); \node at (6.5,-1) {};

\node at (12.5,0) {$\cdots$};

\node at (5,-2) {$i=n$};
\filldraw (3,-7) circle (.1);     \node at (3,-7.3) {};
\filldraw[fill=blue] (5,-7) circle (.1);          \node at (5,-7.3) {};
\filldraw (7,-7) circle (.1);     \node at (7.6,-7) {};
\filldraw[fill=blue] (3,-3) circle (.1);          \node at (3,-2.7) {};
\filldraw (5,-3) circle (.1);     \node at (5,-2.7) {};
\filldraw[fill=blue] (7,-3) circle (.1);          \node at (7,-2.7) {};
\filldraw (9,-3) circle (.1);          \node at (9.3,-2.7) {};
\draw[->] (9,-3)--(7.1,-3);       \node at (8,-2.7) {};

\node[diamond,
  draw = black,
  minimum width = 1cm,
  minimum height = 1.3cm] (d) at (3,-5) {$W_1$};
\node[diamond,
  draw = black,
  minimum width = 1cm,
  minimum height = 1.3cm] (d) at (5,-5) {$W_2$};
\node[diamond,
  draw = black,
  minimum width = 1cm,
  minimum height = 1.3cm] (d) at (1.5,-5) {$W_3$};

\filldraw[fill=blue] (3,-5.7) circle (.1);         \node at (3.7,-5.7) {};

\node at (9.3,-2.7) {{\color{red} $t$}};

\filldraw (5,-5.7) circle (.1);    \node at (6,-5.7) {};
\filldraw (3,-4.4) circle (.1);    \node at (4,-4.4) {};
\filldraw[fill=blue] (5,-4.4) circle (.1);         \node at (5.9,-4.4) {};
\filldraw[fill=blue] (1.5,-5.7) circle (.1);        \node at (0.4,-5.7) {};
\filldraw (1.5,-4.4) circle (.1);  \node at (0.4,-4.4) {};

\draw[->] (3.1,-7)--(4.9,-7);       \node at (4,-6.7) {};
\draw[->] (3,-6.9)--(3,-5.8);       \node at (3.2,-6.3) {};
\draw[->] (5.1,-7)--(6.9,-7);       \node at (6,-6.7) {};
\draw[->] (5,-6.9)--(5,-5.8);       \node at (5.2,-6.3) {};
\draw[->] (7,-3.1)--(7,-6.9);       \node at (7.3,-5) {};
\draw[->] (3,-4.3)--(3,-3.1);       \node at (3.2,-3.7) {};
\draw[->] (5,-4.3)--(5,-3.1);       \node at (5.2,-3.7) {};
\draw[->] (2.95,-3.05)--(1.55,-4.35);  \node at (2,-3.5) {};
\draw[->] (3.1,-3)--(4.9,-3);       \node at (4,-2.7) {};
\draw[->] (6.9,-3)--(5.1,-3);       \node at (6,-2.7) {};
\draw[->] (2.9,-7)--(0.9,-7);       \node at (1.5,-7.3) {};
\draw[->] (1.5,-5.75) to[out=-90,in=-125] (6.95,-7.05); \node at (6.5,-8) {};

\node at (0.3,-7) {$\cdots$};

\end{tikzpicture}
\caption{The Welded Tree Circuit Graph}\label{fig:weldedcircuitgraphsimple}
\end{figure}

The \st alternative electrical flow in the welded tree circuit graph essentially follows the \st alternative electrical flow in the following simple graph as in \fig{pathgraph1introduction}. The additional alternative neighborhoods guarantee that the \st alternative electrical flow has to split evenly at blue vertices that employ alternative Fourier neighbourhoods, which is a consequence of Alternative Kirchhoff's Law. Although multiple such \st flows exists, the alternative electrical flow has the minimum energy among them all. The energy of each such $\falt$ can be explicitly calculated to see that ${\cal E} (\falt)= 5y^2+4x^2+3$, and the energy is minimised when $x= 5/9$ (see \sec{exaG1} for more detail).
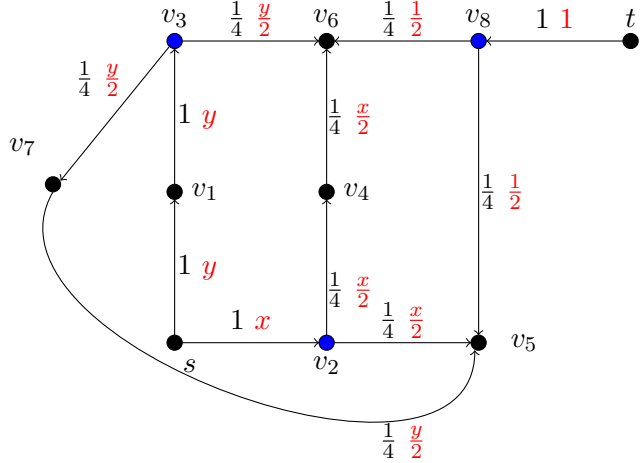
\begin{figure}[h!]
\centering
\begin{tikzpicture}

\filldraw (0,0) circle (.1);     \node at (0.2,-0.3) {$s$};
\filldraw[fill=blue] (2,0) circle (.1);          \node at (2,-0.3) {$v_2$};
\filldraw (4,0) circle (.1);     \node at (4.6,0) {$v_{5}$};
\filldraw[fill=blue] (0,4) circle (.1);          \node at (0,4.3) {$v_{3}$};
\filldraw (2,4) circle (.1);     \node at (2,4.3) {$v_{6}$};
\filldraw[fill=blue] (4,4) circle (.1);          \node at (4,4.3) {$v_{8}$};
\filldraw (6,4) circle (.1);  \node at (6,4.3) {$t$};

\filldraw (0,2) circle (.1);         \node at (0.4,2) {$v_{1}$};
\filldraw (2,2) circle (.1);    \node at (2.4,2) {$v_{4}$};
\filldraw (-1.6,2.1) circle (.1);  \node at (-2,2.6) {$v_{7}$};

\draw[->] (0.1,0)--(1.9,0);       \node at (1,0.3) {$1$ {\color{red} $x$}};
\draw[->] (0,0.1)--(0,1.9);       \node at (0.3,1) {$1$ {\color{red} $y$}};
\draw[->] (2.1,0)--(3.9,0);       \node at (3,0.3) {$\frac{1}{4}$ {\color{red} $\frac{x}{2}$}};
\draw[->] (2,0.1)--(2,1.9);       \node at (2.3,0.7) {$\frac{1}{4}$ {\color{red} $\frac{x}{2}$}};
\draw[->] (4,3.9)--(4,0.1);       \node at (4.3,2) {$\frac{1}{4}$ {\color{red} $\frac{1}{2}$}};
\draw[->] (0,2.1)--(0,3.9);       \node at (0.3,3) {$1$ {\color{red} $y$}};
\draw[->] (2,2.1)--(2,3.9);       \node at (2.3,3) {$\frac{1}{4}$ {\color{red} $\frac{x}{2}$}};
\draw[->] (-.05,3.95)--(-1.5,2.15);  \node at (-1,3.5) {$\frac{1}{4}$ {\color{red} $\frac{y}{2}$}};
\draw[->] (0.1,4)--(1.9,4);       \node at (1,4.3) {$\frac{1}{4}$ {\color{red} $\frac{y}{2}$}};
\draw[->] (3.9,4)--(2.1,4);       \node at (3,4.3) {$\frac{1}{4}$ {\color{red} $\frac{1}{2}$}};
\draw[->] (5.9,4)--(4.1,4);       \node at (5,4.3) {$1$ {\color{red} $1$}};
\draw[->] (-1.6,2) to[out=-120,in=-90] (3.95,-0.1); \node at (3.,-1.3) {$\frac{1}{4}$ {\color{red} $\frac{y}{2}$}};

\end{tikzpicture}
\caption{The graph $G_1$ with corresponding edge directions where the blue vertices have an additional alternative neighbour as defined in \eq{G-star-alt}. For each $(u,v) \in \E$, the weights $\w_{u,v}$ are denoted in black and the flow values $\falt_{u,v}$ in red for any valid unit \st alternative flow are parametrised by $x$ and $y = 1-x$.}\label{fig:pathgraph1introduction}
\end{figure}



We also give a classical lower bound on solving this specific pathfinding problem, which states that any classical algorithm will have to make an exponential number of queries to adjacency list oracle to output any \st path, under the folklore assumption that \st pathfinding is classically hard on welded tree graphs.

The idea of using alternative electrical flows to find an \st path shares similarity with the approach in \cite{Seaneletr},  \cite{jeffery2023quantum} and more recently \cite{wesolowski2024advances}, where they achieve a polynomial speedup to find an \st path for more general classes of graphs. Under the adjacency matrix model, \cite{jeffery2023quantum} suggests a new way to generate \st electrical flow states based on span programs. This approach can be used to sample an \st path, which improves the query complexity of the previous quantum algorithm by \cite{durr2006quantum} for certain types of graphs that have a unique \st path, but the resulting complexities scale at least linearly with the number of vertices, which poses a problem when dealing with exponentially large graphs. Under the adjacency list oracle model, \cite{Seaneletr} uses the Quantum Linear System (QLS) algorithm to generate \st electrical flow states \cite{wang2017efficient} and shows a quadratic advantage to find an \st path in certain graph related to the welded tree graph over both existing quantum and classical algorithms.

We summarise the result of the applications of the multidimensional electrical network framework for graph problems in \tab{applications}.

\begin{table}[h!]
    \centering
    \begin{tabular}{|c|c|c|c|c|c|c}
    \hline
        \textbf{Graphs} &   Input & Output &Techniques  \\
        \hline
          Welded Tree Graph \cite{childs2003ExpSpeedupQW}  & $s$ & $t$ & CQW \\
          \hline
          Random Hierarchical Graph \cite{balasubramanian2023exponential}  & $s$ & $t$ & CQW  \\ 
              \hline
             Welded Tree Path Graph \cite{li2023exponential}   & $s$ & $s$-$t$ path & CQW \\
                  \hline
    One-dimensional Random Hierarchical Graph  [{\color{red} This work}] & $s$ & $t$ &  MEN \\ \hline
    Welded Tree Circuit Graph [{\color{red} This work}] & $s$ & $s$-$t$ path & MEN\\ 
          \hline
    \end{tabular}
    \caption{CQW: Continuous Quantum Walk, MEN: Multidimensional Electrical Network}
    \label{tab:applications}
\end{table}

\subsection{Future Outlook}
One of the potential applications of the multidimensional electrical network framework in the field of quantum query complexity is the welded tree pathfinding problem \cite{aaronson2021open}. The goal of the welded tree pathfinding problem is to find an \st path in the welded tree graph and this problem has been proven to be difficult for certain types of quantum algorithms that are rooted \cite{childs2022quantum}. A rooted quantum algorithm needs to maintain all the paths from the starting vertex $s$ for the vertices within a memory space. In the above pathfinding example application, the multidimensional electrical network framework allows us to sample edges from alternative electrical flow states and then reconnect these to recover an \st path. Since such an algorithm would not be rooted, our multidimensional electrical network framework provides a potentially new avenue to tackle this open problem.




Although our applications all consider adjancency list access to the underlying graph, the multidimensional electrical network can also find applications in the non-oracle setting. For example, in the case of finding an \st path in exponentially sized isogeny graphs \cite{Charles2009}, the adjacency list oracle can be instantiated by computing isogenies between elliptic curves. This problem is believed to be hard even for quantum algorithms, but our new multidimensional electrical network framework gives some new hope, as our works presents an exponential speedup for this pathfinding problem on a regular graph.

The other potential application of the multidimensional electrical network framework is to find a killer application of quantum algorithms based on the QLS algorithm. Arguably, the most promising killer applications are believed to be based on quantum linear algebra with classical input and classical output. However, due to the recent dequantization results \cite{tang2018QuantumInspiredRecommSys,tang2018QInspiredClassAlgPCA,chia2019SampdSubLinLowRankFramework} and caveats of the QLS algorithm \cite{aaronson2015ReadTheFinePrint}, in order for a quantum algorithm to show a superpolynomial speedup in this context, the rank of the matrix in the linear system needs to be high and the condition number of the matrix should be polylogarithmic with respect to size of the matrix. Moreover, the only barrier of showing superpolynomial speedup for the recent quantum algorithms by \cite{chen2022quantum, ding2023limitations} for solving systems of polynomial equations, is because the condition number of their associated matrix is exponentially large. It is unlikely that a standard preconditioning technique will overcome this barrier. As exhibited in this work, the multidimensional electrical network framework allows us to generate the \st alternative electrical flow state in polynomial time in our pathfinding example graph from \fig{graph}, whereas the cost to generate the \st electrical flow state is exponentially large. Since the alternative electrical flow state is a solution to a linear system of equations, as shown in \thm{intromqwflows} and \thm{introelecflow-alt}, our framework may provide new way to get around the condition number barrier for quantum algorithms solving polynomial system. By modifying the system of linear equations, like we change the incidence matrix $B$ to the alternative incidence matrix $\Balt$, one could obtain a efficiently solvable linear system whose solution contains information about the solution to the original linear system. 

Many big open problems in the field of quantum computing can be reduced to the task of generating a specific type of quantum state, such as the graph isomorphism problem \cite{aharononv2007AdiabaticQStateGeneration}, lattice-related problems \cite{eldar2022efficient}, and the problem of computing the ground state of local Hamiltonians \cite{gharibian2015quantum}. However, limited progress has been made on how to generate those quantum states and exhibit exponential speedup on these types of problems. Alternatively, an equally important direction is to design new quantum algorithms to generate certain types of quantum states and use those states to exhibit, hopefully superpolynomial, speedups for certain other problems. Hopefully our new flow state generation techniques may also shed some light on those big open problems related to quantum state generation. 


By reformulating Alternative Kirchhoff's Law and Alternative Ohm's Law in terms of the alternative incidence matrix, see \sec{altincid}, our multidimensional electrical network framework could also have applications in the classical algorithm design for graph problems. Electrical networks have a strong duality to both classical random walks, as well as quantum random walks. In this work we study and recover the connection between (alternative) electrical networks and the generalised quantum walk paradigm of multidimensional quantum walks. It would be interesting to see if there is similar  classical analog of this multidimensional quantum walk by trying to map the (alternative) electrical network back to some classical Markovian process. In addition, Kirchhoff's Law and Ohm's Law were originally derived from real physical observations. It would be worthwhile to study the physical intuition behind our new Alternative Kirchhoff's Law and Alternative Ohm's Law. 

\paragraph{Organization: }The remainder of this article is organised as follows. In \sec{prelim}, we give preliminaries on graph theory, electrical networks and quantum walks. In \sec{multiinc} we show how the concepts related to electrical flow can be generalised to the multidimensional electrical network under the multidimensional quantum walk framework. This results in our new Alternative Kirchhoff's Law and Alternative Ohm's Law.
In addition, we rebuild the connection between the alternative incidence matrix and Alternative Kirchhoff's Law and Alternative Ohm's Law, showing that our new laws seem to be natural definitions and that they generalise known results regarding electrical networks.
In \sec{1d}, we apply the multidimensional electrical network to one-dimensional random hierarchical graphs and show how the framework allows us to sample exponentially faster from the electrical flow state than any classical algorithm can. In \sec{path}, using the multidimensional electrical network, we construct a pathfinding problem where we show that quantum walks can exhibit exponential speedups when it comes to pathfinding problems.

\section{Preliminaries}\label{sec:prelim}
\subsection{Graph theory and electrical networks}

In this section, we define graph-theoretic concepts and basic knowledge of electrical networks following \cite{vishnoi2013lx,jeffery2023multidimensional}. Although experienced readers will be familiar with these notions, we encourage the reader not to skip these definitions, as some of them are not completely standard compared to other works on quantum walks.

\begin{definition}[Network]\label{def:network}
  A network is a connected weighted graph $G = (V,E,\w)$ with a vertex set $V$, an (undirected) edge set $E$ and some weight function $\w:E\rightarrow\mathbb{R}_{>0}$. Since edges are undirected, we can equivalently describe the edges by some set $\E$ such that for all $(u,v)\in E$, exactly one of $(u,v)$ or $(v,u)$ is in $\E$. The choice of edge directions is arbitrary. Then we can view the weights as a function $\w:\E\rightarrow\mathbb{R}_{> 0}$, and for all $(u,v)\in \E$, define $\w_{v,u}=\w_{u,v}$. For convenience, we will define $\w_{u,v}=0$ for every pair of vertices such that $(u,v)\not\in E$. 
  \noindent For an implicit network $G$, and $u\in V$, we will let $\Gamma(u)$ denote the \emph{neighbourhood} of $u$:
  $$\Gamma (u):=\{v\in V:(u,v)\in E\}.$$
  We use the following notation for \emph{the out- and in-neighbourhoods} of $u\in V$:
  \begin{equation}
  \begin{split}
  \Gamma^+(u) &:= \{v\in\Gamma(u):(u,v)\in \E\}\\
  \Gamma^-(u) &:= \{v\in\Gamma(u):(v,u)\in \E\},
  \end{split}\label{eq:neighbourhoods}
  \end{equation}
\end{definition}

\begin{definition}[Flow, Circulation]\label{def:flow}
  A \emph{flow} on a network $G = (V,E,\w)$ is a real-valued function $\f:\E\rightarrow\mathbb{R}$, extended to edges in both directions by $\f_{u,v}=-\f_{v,u}$ for all $(u,v)\in\E$. For any flow $\f$ on $G$, vertex $u\in V$, and subset $A \subseteq V$ we define $\f_u=\sum_{v\in \Gamma(u)}\f_{u,v}$ as the flow coming out of $u$. If $\f_u=0$, we say flow is conserved at $u$. If flow is conserved at every vertex, we call $\f$ a \emph{circulation}. If $\f_u>0$, we call $u$ a \emph{source}, and if $\f_u<0$ we call $u$ a \emph{sink}.    
  A flow with a unique source $s$ and unique sink $t$ (satisfying $\f_s = -\f_t = - 1$) is called an \emph{(unit) \st flow}. The \emph{energy} of any flow $\f$ is 
  \begin{equation*}  
  {\cal E}(\f):=\sum_{(u,v)\in\E}\frac{\f_{u,v}^2}{\w_{u,v}}.\end{equation*}
  The \emph{effective resistance ${\cal R}_{s,t}$} is given by the minimal energy ${\cal E}(\f)$ over all unit flows $\f$ from $s$ to $t$. The \st \emph{electrical flow} is the unique unit \st flow that achieves this minimal energy.
\end{definition}

\begin{definition}[Potential]\label{def:potential}
  A \emph{potential vector} (also known as potential function) on a network $G = (V,E,\w)$ is a real-valued function $\p:V \rightarrow \mathbb{R}$ that assigns a potential $\p_u$ to each vertex $u \in V$.
\end{definition}

\begin{definition}[Electrical Network]\label{def:elec-network}
  Given a network $G = (V,E,\w)$ with a weight function $\w$, we can interpret every edge $(u,v)\in E$ as a resistor with resistance $1/\w_{u,v}$. This allows $G$ to be modeled as an \emph{electrical network}.
\end{definition}

Two fundamental laws related to electrical networks are Kirchhoff's Law (also known as Kirchhoff's Node Law) and Ohm's Law. The former states the definition of a \st flow, as in \defin{flow}:
\begin{definition}[Kirchhoff's Law]\label{def:kcl}
For any \st flow on an electrical network $G = (V,E,\w)$ with $s,t \in V$, the amount of electrical flow that enters any $u \in V \backslash \{s,t\}$ is equal to the amount of flow that exits $u$, that is, $\sum_{v\in\Gamma(u)} \f_{u,v}=0$.
\end{definition}

The latter states that if we inject a unit of current into $s$ and extract it from $t$ in the electrical network $G$, then there is an induced potential vector $\p$ which relates to the \st electrical flow $\f$:

\begin{definition}[Ohm's Law]\label{def:ohm} 
Let $\f$ be the \st electrical flow on an electrical network $G = (V,E,\w)$ with $s,t \in V$. Then there exists a potential vector $\p$ such that the potential difference between the two endpoints of any edge $(u,v) \in E$ is equal to the amount of electrical flow $\f_{u,v}$ along this edge multiplied with the resistance $1/\w_{u,v}$, that is, $\p_u-\p_v=\f_{u,v}/\w_{u,v}$. 
\end{definition}

The potential $\p$ induced by an \st electrical flow $\f$ in Ohm's Law is not unique and it is therefore convention to consider the potential $\p$ that assigns $\p_t = 0$, in which case $\p_s = {\cal R}_{s,t}$. 

\subsection{The incidence matrix, Kirchhoff's Law and Ohm's Law}\label{sec:inc}

We start by restating the connection between on one hand the incidence matrix of a network $G$ and on the other hand Kirchhoff's Law and Ohm's Law. We follow \cite[section 4]{vishnoi2013lx} in doing so.

\begin{definition}[The edge-vertex incidence matrix]\label{def:incidence}
  Let $G=(V,E,\w)$ be a network (See \defin{network}). The \textit{incidence matrix} $B \in \mathbb{C}^{\E \times V}$ of $G$, is the matrix whose rows are indexed by $(u,v) \in \E$, whose columns are indexed $u\in V$ and whose only non-zero entries are given by
  \begin{align*}
    &B_{(u,v),u} = \sqrt{\w_{u,v}}, &B_{(u,v),v} = -\sqrt{\w_{u,v}}.
  \end{align*} 
\end{definition}

Let $W \in \mathbb{C}^{\E \times \E}$ be the diagonal matrix with entries $W_{(u,v),(u,v)} = 1/\sqrt{\w_{u,v}}$. By considering a flow $\f$ on $G = (V,E,\w)$ not only as a function on $\E$, but also as a vector in $\mathbb{C}^{\E}$, we can multiply it with the matrix $W$ to obtain the weighted flow vector $W\f \in \mathbb{C}^{\E}$ with entries $(W\f)_{u,v} = \f_{u,v}/\sqrt{\w_{u,v}}$ for the row indexed by $(u,v) \in \E$. The norm of $W\f$ is therefore precisely given by $\sqrt{{\cal E}(\f)}$. By the introduction of $W\f$, we can rephrase Kirchhoff's Law from \defin{kcl} as a linear equation involving the incidence matrix $B$. Fix some ordering of the columns of $B$ of the form $s,u_1,\dots,u_2,t$ for some $u_1,u_2 \in V \backslash \{s,t\}$ and define the basis vectors ${\sf e}_i \in \mathbb{C}^{n}$ which have a $1$ at the $i$-th location and zero elsewhere.

\begin{definition}[Kirchhoff's Law (incidence matrix)]\label{defin:knl-inc} 
Let $\f$ be any unit \st flow on an electrical network $G = (V,E,\w)$. Let $B$ be the incidence matrix of $G$. Then $\f$ satisfies
  \begin{equation}\label{eq:matrix-knl}
  B^{T} W\f = \begin{bmatrix}
   \sum_{v \in \Gamma(s)} \f_{s,v} \\ \sum_{v \in \Gamma(u_1)} \f_{u_1,v} \\ \vdots \\\sum_{v \in \Gamma(u_2)} \f_{u_2,v} \\ \sum_{v \in \Gamma(t)} \f_{t,v} \end{bmatrix} = 
   \begin{bmatrix}
  1 \\ 0 \\ \vdots \\ 0 \\ -1 \end{bmatrix} = {\sf e}_s - {\sf e}_t.
  \end{equation}
\end{definition}

Recall from \defin{flow} that the \st electrical flow is the flow that minimises ${\cal E}(\f)$ for all unit \st flows $\f$. Since ${\cal E}(\f) = \norm{W\f}^2$, this means that the \st electrical flow corresponds to the `smallest' (in norm) solution to \eq{matrix-knl}, that is, the unique \st flow $\f$ such that its flow vector satisfies $W\f \in \textnormal{ker}(B^T)^{\perp}$. We can therefore recover $W\f$ by making use of the \textit{Moore-Penrose inverse} (also known as the pseudoinverse) of $B^T$, denoted by $B^{T+}$. For any matrix $A$, the Moore-Penrose inverse $A^+$ (not to be confused with the conjugate transpose $A^{\dagger}$), is the unique matrix satisfying
\begin{align} \label{eq:property}
  &AA^+A = A, &&A^+AA^+ = A^+, &&&(AA^+)^{\dagger} = AA^+, &&&&(A^+A)^{\dagger} = A^+A,
\end{align}
and it is well known that $A^{+}$ maps $\textnormal{ran}(A)$ to $\textnormal{ker}(A)^{\perp}$. Hence by left-multiplying both sides of \eq{matrix-knl} with $B^{T+}$, we recover the 
following important property of electrical networks:

\begin{theorem}[Theorem 4.7 in \cite{vishnoi2013lx}]\label{thm:elecflow}
  Let $\f$ be the \st electrical flow on a network $G = (V,E,\w)$. Let $B$ be the incidence matrix of $G$. Then its flow vector $W\f$ is given by 
  \begin{equation}\label{eq:transpose}
    W\f = B^{T+}({\sf e}_s - {\sf e}_t).
  \end{equation}
\end{theorem}

Just like we did with $\f$, we can also consider a potential vector $\p$ as a vector (hence the name) in $\mathbb{C}^{V}$ with entries $\p_{u}$ for the row indexed by $u \in V$. In doing so, we can rephrase Ohm's Law from \defin{ohm} as a linear equation involving the incidence matrix $B$. Fix some ordering of the rows of $B$ of the form $(u_1,v_1),\dots,(u_2,v_2) \in \E$.

\begin{definition}[Ohm's Law (incidence matrix)]\label{defin:ohm-inc} 
Let $\f$ be the \st electrical flow on an electrical network $G = (V,E,\w)$. Let $B$ be the incidence matrix of $G$. Then there exists a potential vector $\p$ such that
  \begin{equation}\label{eq:matrix-ohm}
  B\p = \begin{bmatrix}
   \sqrt{\w_{u_1,v_1}}\left(\p_{u_1}-\p_{v_1}\right)\\ \vdots \\ \sqrt{\w_{u_2,v_2}}\left(\p_{u_2}-\p_{v_2}\right)
  \end{bmatrix} = \begin{bmatrix}
   \frac{\f_{u_1,v_1}}{\sqrt{\w_{u_1,v_1}}}\\ \vdots \\ \frac{\f_{u_2,v_2}}{\sqrt{\w_{u_2,v_2}}}
  \end{bmatrix}=W\f .
\end{equation}
\end{definition}

We may assume that the potential vector $\p$ satisfying Ohm's Law to satisfy $\p_s = {\cal R}_{s,t}$ and $\p_t = 0$, which is easier to see from the incidence matrix perspective.

\begin{lemma}\label{lem:potential}
   Let $\f$ be the \st electrical flow on an electrical network $G = (V,E,\w)$ with effective resistance ${\cal R}_{s,t}$. Then there exists a potential vector $\p$ satisfying Ohm's Law such that $\p_s = {\cal R}_{s,t}$ and $\p_t = 0$. 
\end{lemma}
\begin{proof}
  From the incidence matrix $B$, we can obtain $B^TB$, which is known as \textit{the weighted Laplacian} of $G$. It is well known in spectral graph theory (see e.g. Lemma 2.2 in \cite{vishnoi2013lx}), that $B^TB$ has $0$ as an eigenvalue with multiplicity $1$, whose corresponding eigenvector is given by $\sum_{u \in V}e_u$. Since $\textnormal{ker}(B) = \textnormal{ker}(B^TB)$, not only does this mean that by setting $\p_t = 0$, we still have a valid solution to \eq{matrix-ohm}, but this actually makes the remaining solution unique. By left-multiplying both sides of \eq{transpose} with $(W\f)^T$ we obtain together with \eq{matrix-ohm} that
  \begin{equation}\label{eq:potential-s}
    {\cal R}_{s,t} = \norm{W\f}^2 = (W\f)^T B^{T+}({\sf e}_s - {\sf e}_t) = \p^T({\sf e}_s - {\sf e}_t) = \p_s - \p_t = \p_s.
  \end{equation}  
\end{proof}

With the Moore-Penrose inverse we can in fact recover the potential from \lem{potential}. To achieve this, we remove the last column of $B$ and last row of $\p$ to obtain $\overline{B}$ and $\overline{\p}$, effectively forcing $\p_t = 0$:
\begin{equation}\label{eq:matrix-pot}
  \p = \begin{bmatrix} \overline{\p} \\ 0 \end{bmatrix} = \begin{bmatrix} \overline{B}^{+}W\f \\ 0  \end{bmatrix}.
\end{equation}

\subsection{Quantum walks and electrical flow}\label{sec:walk}

There is a direct relationship between the analysis of random walks and electrical networks, see for example \cite{MR3616205}. The relationship between quantum walks and electrical networks was built for the first time by \cite{belovs2013ElectricWalks}, where electrical network theory was used to construct and analyse a phase estimation algorithm to detect whether a given graph contained a marked element. Recently, \cite{piddock2019electricfind,apers2022elfs} have shown that the resulting state after running this phase estimation is actually a quantum state representing the electrical flow between a starting vertex and the marked vertices. For a network $G = (V,E,\w)$ and vertices $s,t \in V$, let 
 
$$
\mathcal{H}=\mathrm{span}\{ \ket{u,v} | (u,v)\in E\}
$$ be the associated vector space of its edges. We emphasise here, especially for the readers familiar with \cite{jeffery2023multidimensional}, that each edge $(u,v) $ 'appears twice' in $\mathcal{H}$: both as $\ket{u,v}$ and $\ket{v,u}$, which are orthogonal states in $\mathcal{H}$. For each vertex $u\in V$, we let $\du_u=\sum_{v\in \N(u)} \w_{u,v}$ be the weighted degree of $u$. We use it to define the (normalised) \textit{star state} of $u$ as 
$$\ket{\psi_u} = \frac{1}{\sqrt{\du_u}}\left(\sum_{v\in \N^+(u)}\sqrt{\w_{u,v}}\ket{u,v}-\sum_{v\in \N^-(u)}\sqrt{\w_{u,v}}\ket{u,v}\right) = \frac{1}{\sqrt{\du_u}} \sum_{v\in \N(u)}(-1)^{\Delta_{u,v}}\sqrt{\w_{u,v}}\ket{u,v}.$$ 
Here for any $(u,v) \in E$, the quantity $\Delta_{u,v}$ is equal to $0$ if $(u,v) \in\E$ and $1$ if $(v,u) \in \E$. This definition of a star state is slightly different from most of the literature, where there is usually no sign-difference depending on whether $(u,v)$ is part of the directed edge set, but this will be necessary later on when working with the multidimensional quantum walk framework from \cite{jeffery2023multidimensional}. Now consider the following two subspaces of ${\cal H}$. Let 
$${\cal A} := \mathrm{span}\{\ket{\psi} \in {\cal H}: \brakett{u,v}{\psi} = -\brakett{v,u}{\psi}~~ \forall \ket{u,v}\in {\cal H}\}$$
be the \textit{antisymmetric subspace} of ${\cal H}$. Moreover, let ${\cal B} := \mathrm{span}\{\ket{\psi_u}: u \in V \backslash \{s,t\}\}$ be the \textit{star space} of ${\cal H}$. Then the \textit{quantum walk operator} $U_{{\cal AB}}$ is defined as 
\begin{equation}\label{eq:walk-op}
  U_{{\cal AB}}:=(2\Pi_{{\cal A}}-I)(2\Pi_{{\cal B}}-I) ,
\end{equation}
where $\Pi_{\cal A}$ and $\Pi_{\cal B}$ are orthogonal projectors onto ${\cal A}$ and ${\cal B}$ respectively. Note that 
\begin{align*}
  &2\Pi_{{\cal A}}-I = -{\sf SWAP}, &2\Pi_{{\cal B}}-I = 2\sum_{u\in V\backslash \{s,t\}} \proj{\psi_u} - I,
\end{align*}
where ${\sf SWAP}$ acts as ${\sf SWAP}\ket{u,v} = \ket{v,u}$ for any $\ket{u,v} \in {\cal H}$. For any star state $\ket{\psi_u}$, we write 
$$\ket{\psi_u^+} := \sqrt{2} (I-\Pi_{\cal A})\ket{\psi_u} = \frac{I + {\sf SWAP}}{\sqrt{2}}\ket{\psi_u}$$
for its normalised projection onto ${\cal A}^{\perp}$, which is also known as the \textit{symmetric subspace} of ${\cal H}$. For any flow $\f$, we define its associated (normalised) flow state in ${\cal H}$ as
\begin{equation}\label{eq:flowstate}
  \ket{\f} := \frac{1}{\sqrt{2{\cal E}(\f)}}\sum_{(u,v) \in \E}\frac{\f_{u,v}}{\sqrt{\w_{u,v}}}\left(\ket{u,v} + \ket{v,u}\right).
\end{equation}

In the case where $\f$ is the \st electrical flow, we define the (unnormalised) state associated with the induced potential vector $\p$ (with the convention that $\p_t=0$) as
\begin{equation}\label{eq:potstate}
  \ket{\p}= \sqrt{\frac{2}{{\cal R}_{s,t}}}\sum_{u\in V\backslash \{ s\}} \p_u\sqrt{\du_u}\ket{\psi_u}.
\end{equation}

In \cite{piddock2019electricfind,apers2022elfs}, this potential state $\ket{\p}$ is used to exhibit that by running phase estimation on the quantum walk operator $U_{\cal AB}$, we can obtain a close approximation to the flow state $\ket{\f}$. The precision required in this phase estimation algorithm scales with a quantity in \cite{apers2022elfs} is defined as the escape time $\mathsf{ET}_{s}$: 
$$\mathsf{ET}_{s} := \frac{1}{2}\norm{\ket{\p}}^2 = \frac{1}{{\cal R}_{s,t}} \sum_{u\in V}\p_u^2\du_u.$$ 
Since we will not be using the operational meaning of $\mathsf{ET}_{s}$ in this work, we will omit $\mathsf{ET}_{s}$ in the rest of this work and instead work with $\norm{\ket{\p}}$.

The following lemma tells us something about the output of a particular particular kind of quantum algorithm known as phase estimation~\cite{kitaev1996PhaseEst}. The input to the phase estimation algorithm is an initial state $\ket{\psi}$ and a unitary operator $U$. When the measured phase value is $0$, the phase estimation algorithm projects the initial state $\ket{\psi}$ onto the 1-eigenspace of the unitary $U$. Notably, the following result applies to a broader context beyond quantum walks and electrical network.

\begin{lemma}[Modified Lemma 8 in \cite{piddock2019electricfind} and Lemma 10 in \cite{apers2022elfs}]\label{lem:qwflows}
  Define the unitary $U_{\cal AB} = (2\Pi_{\cal A} - 1)(2\Pi_{\cal B} - 1)$ acting on a Hilbert space ${\cal H}$ for projectors $\Pi_{\cal A},\Pi_{\cal B}$ onto some subspaces ${\cal A}$ and ${\cal B}$ of ${\cal H}$ respectively. Let $\ket{\psi} = \sqrt{p}\ket{\varphi} + (I - \Pi_{\cal A})\ket{\phi}$ be a normalised quantum state such that $U_{\cal AB}\ket{\varphi} = \ket{\varphi}$ and $\ket{\phi}$ is a (unnormalised) vector satisfying $\Pi_{\cal B}\ket{\phi} = \ket{\phi}$. Then performing phase estimation on the state $\ket{\psi}$ with operator $U_{\cal AB}$ and precision $\delta$ outputs ``$0$'' with probability $p' \in [p,p + \frac{17\pi^2\delta\norm{\ket{\phi}}}{16}]$, leaving a state $\ket{\psi'}$ satisfying
  $$ \frac{1}{2}\norm{\proj{\psi'} - \proj{\varphi}}_1 \leq \sqrt{\frac{17\pi^2\delta\norm{\ket{\phi}}}{16p}}. $$
  Consequently, when the precision is $O\left(\frac{p\epsilon^2}{\norm{\ket{\phi}}}\right)$, the resulting state $\ket{\psi'}$ satisfies
  $$ \frac{1}{2}\norm{\proj{\psi'} - \proj{\varphi}}_1 \leq \epsilon. $$
\end{lemma}
\begin{proof} See \app{qwflows}. \end{proof}
This lemma is almost equivalent to Lemma 8 in \cite{piddock2019electricfind} and Lemma 10 in \cite{apers2022elfs}, but we have modified it slightly as we were unable to verify the constants in \cite{piddock2019electricfind, apers2022elfs} and the scaling with the precision in \cite{apers2022elfs}. The theory of electrical networks tells us if we consider the \st electrical flow $\f$, then we can apply \lem{qwflows} to approximate the \st electrical flow state $\ket{\f}$. 

\begin{corollary}\label{cor:qwflows}
  Let $U_{\cal AB}$ be the quantum walk operator as defined in \eq{walk-op}. Then by performing phase estimation on the initial state $\ket{\psi_s^+}$ with the operator $U_{\cal AB}$ and precision $O\left(\frac{\epsilon^2}{\sqrt{{\cal R}_{s,t} \du_s}\norm{\p}}\right)$, the phase estimation algorithm outputs ``$0$'' with probability $\Theta\left(\frac{1}{\sqrt{{\cal R}_{s,t} \du_s}}\right)$, leaving a state $\ket{\f'}$ satisfying
  $$ \frac{1}{2}\norm{\proj{\f'} - \proj{\f}}_1 \leq \epsilon. $$
\end{corollary}

We emphasise that the probability $p$ in \lem{qwflows} should not be confused with the potential $\p$ in \eq{potstate} and \cor{qwflows}.

Before we provide the proof of \cor{qwflows}, which can also be found in \cite{piddock2019electricfind,apers2022elfs}, we remark that it is possible to modify the network $G$ to ensure that ${\cal R}_{s,t} \du_s = \Theta(1)$, which is a standard tool used in quantum electrical networks \cite{belovs2013ElectricWalks}.

\begin{proof}
Firstly, by Kirchhoff's Law (see \defin{kcl}), we know the \st electrical flow $\f$ is conserved at each vertex $u \in V \setminus \{s,t\}$, which shows that $\Pi_{\cal B}\ket{\f} = 0$:
\begin{equation}\label{eq:flow-projB}
\begin{split}
  \langle \psi_u | \f \rangle &= \frac{1}{\sqrt{2{\cal R}_{s,t}\du_u}}\sum_{v\in \N(u)} (-1)^{\Delta_{u,v}} \sqrt{\w_{u,v}} \bra{u,v} \sum_{(u,v)\in \E} \frac{\f_{u,v}}{\sqrt{\w_{u,v}}} (\ket{u,v}+\ket{v,u})\\
  &= \frac{1}{\sqrt{2{\cal R}_{s,t}\du_u}}\left(\sum_{v\in \N^{+}(u)} \f_{u,v} + \sum_{v\in \N^{-}(u)} -\f_{v,u}\right) = \frac{1}{\sqrt{2{\cal R}_{s,t}}}\sum_{v\in \N(u)} \f_{u,v} =0.
\end{split}
\end{equation}

By Ohm's Law (see \defin{ohm}), we know that there exists a potential $\p$, with $\p_t = 0$, such that for each edge $(u,v) \in E$ we have $\p_u-\p_v=\frac{\f_{u,v}}{\w_{u,v}}$. This shows that $\Pi_{\cal A}\ket{\f} = 0$, which combined with the fact that $\Pi_{\cal B}\ket{\f} = 0$ shows that $\ket{\f}$ is indeed a normalised $+1$-eigenvector of $U_{\cal AB}$:
\begin{equation}\label{eq:flow-projA}
\begin{split}
   \ket{\f}&= \frac{1}{\sqrt{2{\cal R}_{s,t}}} \sum_{(u,v)\in \E} \frac{\f_{u,v}}{\sqrt{\w_{u,v}}} (\ket{u,v}+\ket{v,u}) \\
   &= \frac{1}{\sqrt{2{\cal R}_{s,t}}} \sum_{(u,v)\in \E} \left( \sqrt{\w_{u,v}}(\p_u-\p_v)\ket{u,v} + (\p_u-\p_v) \sqrt{\w_{u,v}}\ket{v,u}\right) \\
   &= \frac{1}{\sqrt{2{\cal R}_{s,t}}} \left(\sum_{u \in V} \p_u \sum_{v\in \N(u)} (-1)^{\Delta_{u,v}}\sqrt{\w_{u,v}}\ket{u,v} + {\sf SWAP}\sum_{u \in V} \p_u\sum_{v\in \N(u)} (-1)^{\Delta_{u,v}}\sqrt{\w_{u,v}}\ket{u,v}\right) \\
   &= (I-\Pi_{\cal A})\sqrt{\frac{2}{{\cal R}_{s,t}}} \sum_{u \in V} \p_u\sqrt{\du_u}\ket{\psi_u}.
\end{split}
\end{equation}

Not only does \eq{flow-projA} tells us that $\Pi_{\cal A}\ket{\f} = 0$, it also immediately shows us how to decompose $\ket{\f}$ to obtain the factor $\ket{\p}$, where we make use of the fact that $\p_s = {\cal R}_{s,t}$:

\begin{align*}
  \ket{\f} &= (I-\Pi_{\cal A})\sqrt{\frac{2}{{\cal R}_{s,t}}}\sum_{u \in V} \p_u\sqrt{\du_u}\ket{\psi_u} = (I-\Pi_{\cal A})\ket{\p} + (I-\Pi_{\cal A})\sqrt{\frac{2}{{\cal R}_{s,t}}}\p_s\sqrt{\du_s}\ket{\psi_s}\\
  &= (I-\Pi_{\cal A})\ket{\p} + \sqrt{{\cal R}_{s,t} \du_s}\ket{\psi_s^+},
\end{align*}
which we can rewrite to
\begin{equation}\label{eq:rel-flow-pot}
  \ket{\psi_s^+} = \frac{1}{\sqrt{{\cal R}_{s,t} \du_s}}\ket{\f} - (I-\Pi_{\cal A})\frac{1}{\sqrt{{\cal R}_{s,t} \du_s}}\ket{\p}.
\end{equation}

Lastly, since $\p_t = 0$, we immediately have by its definition in \eq{potstate} that $\ket{\p} \in {\cal B}$, meaning $\Pi_{\cal B}\ket{\p} = \ket{\p}$. Hence by applying \lem{qwflows} with the parameters $\ket{\psi} = \ket{\psi_s^+}$, $\ket{\varphi} = \ket{\f}$, $\ket{\phi} = -\frac{1}{\sqrt{{\cal R}_{s,t} \du_s}}\ket{\p}$ and $p = \frac{1}{{\cal R}_{s,t} \du_s}$, we find that the resulting state after running phase estimation on the quantum walk operator $U_{\cal AB}$ with initial state $\ket{\psi_s^+}$ is approximately the \st electrical flow state. 
  
\end{proof}

 \section{Multidimensional electrical networks and the alternative incidence matrix }\label{sec:multiinc}


In this section, based on the multidimensional quantum walk framework \cite{jeffery2023multidimensional}, we extend the electrical network to the multidimensional electrical network by generalising Kirchhoff's Law and Ohm's Law as Alternative Kirchhoff's Law and Alternative Ohm's Law, respectively. One of the key techniques used in the multidimensional quantum walk framework is the introduction of alternative neighbourhoods, where each vertex is associated with a subspace instead of a single vector (its star state) as was the case in \sec{walk}. 
 
\subsection{Alternative neighbourhoods}

As discussed and motivated in the introduction, the multidimensional quantum walk framework modifies the quantum walk operator through the use of \textit{alternative neighbourhoods}.

\begin{definition}[Alternative Neighbourhoods]\label{def:alternative}
For a network $G = (V,E,\w)$ and for each vertex $u \in V$, a set of \emph{alternative neighbourhoods} is a collection of states $\Psi_\star(u)$ such that $\ket{\psi_u} \in \Psi_\star(u)$ and
$$\Psi_\star=\{\Psi_\star(u)\subset \mathrm{span}\{\ket{u,v}:v\in \N(u)\}: u\in V\}$$
We view the states of $\Psi_\star(u)$ as different possibilities for $\ket{\psi_u}$, only one of which is ``correct''. We say we can \emph{generate $\Psi_\star$ in complexity ${\sf A}_\star$} if there is a map $U_\star$ that can be implemented with complexity ${\sf A}_\star$ and for each $u\in V$, an orthonormal basis $\overline{\Psi}(u)=\{\ket{{\psi}_{u,0}},\dots,\ket{{\psi}_{u,a_u-1}}\}$ of size $a_u < \abs{\N(u)}$ for $\mathrm{span}\{\Psi_\star(u)\}$, such that for all $i\in \{0,\dots,a_u-1\}$,
$U_{\star}\ket{u,i}=\ket{\psi_{u,i}}.$ 
\end{definition}

In \defin{alternative} we never exclude the possibility that the dimension $a_u$ of the alternative neighbourhood $\Psi_{\star}(u)$ is equal to one, in which case $\Psi_{\star}(u) = \{\ket{\psi_{u,0}}\} = \{\ket{\psi_{u}}\}$. If that is the case, we will say that $u$ has no additional alternative neighbourhoods. These alternative neighbourhoods were introduced in \cite{jeffery2023multidimensional} to tackle the case where it might be computationally easier to generate $\Psi_{\star}(u)$ instead of $\ket{\psi_u}$. This can happen for example when dealing with an adjacency list oracle, where a single query to a vertex $u$ does not allow us to distinguish its neighbours, which is needed when the star state $\ket{\psi_{u}}$ that we want to generate has different weights $\w_{u,v}$ for different neighbours $v$. In the well known welded tree problem \cite{childs2003ExpSpeedupQW} we are indeed dealing with such an adjacency list oracle and this hurdle is tackled in \cite{jeffery2023multidimensional} using alternative neighbourhoods. This specific approach for welded tree will be discussed and generalised in \sec{1d}.

By modifying the quantum walk operator $U_{\cal AB}$ to reflect around the span of $\Psi_\star$ instead of the span of all star states $\ket{\psi_u}$, this reduces the cost of applying the walk operator $U_{\cal AB}$. As a result, one can reduce the precision needed in the phase estimation algorithm by reducing the weight of the graph, which directly reduces $\norm{\ket{\palt}}$, at the cost of increasing the effective resistance ${\cal R}_{s,t}$, without incurring an additional cost in calling $U_{\cal AB}$. 

The addition of these alternative neighbourhoods in $\Psi_{\star}$ modifies the quantum walk operator $U_{\A\Balt}$, by increasing the star space ${\cal B}$:
$$\Balt = \mathrm{span}\{\ket{\psi_{u,i}}: u \in V \backslash \{s,t\}, i \in \{0,\dots,a_u-1\}\}.$$
\noindent Through this modification, the quantum walk operator $U_{{\cal AB}}$ with respect to $\Psi_{\star}$ is altered to
\begin{equation}\label{eq:walk-alt}
  U_{\A\Balt}=(2\Pi_{{\cal A}}-I)(2\Pi_{\Balt}-I) ,
\end{equation}
where $\Pi_{\cal A}$ and $\Pi_{\Balt}$ are orthogonal projectors onto $\A$ and $\Balt$ respectively, meaning
\begin{align*}
  &2\Pi_{{\cal A}}-I = -{\sf SWAP}, &2\Pi_{\Balt}-I = 2\sum_{u\in V\backslash \{s,t\}}\sum_{i = 0}^{a_u-1} \proj{\psi_{u,i}} - I.
\end{align*}

We would like to be able to apply \lem{qwflows} to this more general walk operator as well, meaning we want to find an alternative unit \st flow $\falt$, an (unnormalised) state $\ket{\palt}$ and (normalised) state $\ket{\psi}$ such that the following conditions are satisfied:
\begin{enumerate}
  \item $U\ket{\falt}=\ket{\falt}$.
  \item $(I-\Pi_{\cal A})\ket{\palt}+ \sqrt{{\cal E(\falt)}\du_s}\ket{\psi}=\ket{\falt}$.
  \item $\Pi_{\Balt} \ket{\palt} =\ket{\palt}$.
\end{enumerate}

For simplicity we will assume in the rest of this work that $s$ and $t$ do not contain any additional alternative neighbourhoods, as it greatly simplifies notation and intuition. In our applications in \sec{1d} and \sec{path} these simplifying assumptions will also hold. 

\subsection{Alternative Kirchhoff's Law}\label{sec:kal}

Recall the definition of a flow state from \eq{flowstate} for any flow $\f$. By construction, $\ket{\f}$ lives in the symmetric subspace ${\cal A}^{\perp}$, since
$$\Pi_{\cal A}\left(\ket{u,v}+\ket{v,u}\right) = \frac{I - {\sf SWAP}}{2}\left(\ket{u,v}+\ket{v,u}\right) = 0.$$

Hence any \st flow $\f$ that we select will satisfy $\Pi_{\cal A}\ket{\f} = 0$. For the flow state $\ket{\f}$ to live in the $+1$-eigenspace of $U$, it rests us to find some $\f$ such that $\Pi_{\Balt}\ket{\f} = 0$. In \eq{flow-projB}, we used Kirchhoff's Law for this goal, which showed that for any \st flow $\f$ and vertex $u \in V \backslash \{s,t\}$, we have $\brakett{\psi_u}{\f} = 0$. However, in the multidimensional electrical network, it must be orthogonal to all states in $\Balt$ instead of $\B$. That is, the state $\ket{\falt}$ must be orthogonal to all of $\mathrm{span}(\Psi_{\star}(u))$ for every $u \in V \backslash \{s,t\}$. One can interpret this as the flow being conserved for ``all'' alternative neighborhoods. We therefore modify Kirchhoff's Law to be \textit{Alternative Kirchhoff's Law}.

\begin{definition}[Alternative Kirchhoff's Law]\label{def:kcl-alt}
For any \st alternative flow $\falt$ with respect to a collection of alternative neighbourhoods $\Psi_{\star}$ on an electrical network $G = (V,E,\w)$ with $s,t \in V$, the corresponding flow state $\ket{\falt}$ is orthogonal to $\mathrm{span}(\Psi_{\star}(u))$ for every $u \in V \backslash \{s,t\}$, that is, $\langle {\psi}_{u,i} |\f \rangle =0$ for each $i\in \{0,1,\cdots,a_u-1\}$.
\end{definition}

We refer to any unit \st flow satisfying Alternative Kirchhoff's Law as an alternative unit \st flow. Similarly as in \defin{flow}, we define the \st alternative electrical flow with respect to $\Psi_{\star}$ as the alternative unit \st flow achieving minimal energy:
\begin{definition}[Alternative Electrical Flow]\label{def:flow-alt}
For a collection of alternative neighbourhoods $\Psi_{\star}$ on an electrical network $G = (V,E,\w)$ with $s,t \in V$, the \emph{\st alternative electrical flow} is the alternative unit \st flow with minimal energy ${\cal E}(\falt)$. We call this minimal energy the \emph{alternative effective resistance} ${\cal R}_{s,t}^{\alt}$.
\end{definition}

Right now it might seem as if this is ill-defined, as at first glance there could very well be multiple alternative unit \st flows that achieve the minimal energy ${\cal R}_{s,t}^{\alt}$, but we prove in \thm{elecflow-alt} that the \st alternative electrical flow is indeed unique (as long as any alternative unit \st flow exists at all). It might be that the \st electrical flow also satisfies Alternative Kirchhoff's Law, meaning that it coincides with the \st alternative electrical flow. We show an example of this in \sec{1d} and this allows us to apply \lem{qwflows} directly using similar parameters as in \cor{qwflows}. The other side of the spectrum is that there might not be any \st flow at all that satisfies Alternative Kirchhoff's Law, in which case the \st alternative electrical flow does not exist. We show an example of this shortly. The most likely scenario however is that we are right in the middle where the \st electrical flow and \st alternative electrical flow do not coincide, meaning we can not rely on Ohm's Law. 

\subsection{Alternative Ohm's Law} 

To apply \lem{qwflows}, we still need to find an (unnormalised) state $\ket{\palt}$ and (normalised) state $\ket{\psi}$ such that
\begin{enumerate}
  \item $(I-\Pi_{\cal A})\ket{\palt}+ \sqrt{{\cal E(\falt)}\du_s}\ket{\psi}=\ket{\falt}$.
  \item $\Pi_{\cal \Balt} \ket{\palt} =\ket{\palt}$.
\end{enumerate} 

In the case that the \st alternative electrical flow $\falt$ does not overlap with the \st electrical flow, we will not be able to find a potential vector $\p$ defined on the vertices $V$ satisfying Ohm's Law. So instead we will be looking for a potential vector $\palt$ on the edges $E$, meaning it assigns a potential $\palt_{u,v}$ to each edge $(u,v) \in E$. 
\begin{definition}[Alternative Potential]\label{def:alterpotential}
  An \emph{alternative potential vector} (or alternative potential function) on a network $G = (V,E,\w)$ is a real-valued function $\palt: E \rightarrow \mathbb{R}$ that assigns a potential $\p_{u,v}$ to each ordered pair $(u,v) \in E$.
\end{definition}

Similarly to how the potential vector satisfied $\p_s = {\cal R}_{s,t}$ and $\p_t = 0$, we require the alternative potential vector $\palt$ to satisfy $\palt_{s,v} = {\cal R}_{s,t}^{\alt}$ and $\palt_{t,v} = 0$ for every $v \in \N(s)$ (resp. $v \in \N(t)$). We define its corresponding state in ${\cal H}$ as
\begin{equation}\label{eq:potential-edge}
  \ket{\palt}=\sqrt{\frac{2}{{\cal R}(\falt)}} \sum_{u \in V \setminus \{s\}}\sum_{v \in \N(u)}(-1)^{\Delta_{u,v}} \palt_{u,v} \sqrt{\w_{u,v}} \ket{u,v}.
\end{equation}

\begin{definition}[Alternative Ohm's Law]\label{def:ohm-alt}
Let $\falt$ be the \st alternative electrical flow with respect to a collection of alternative neighbourhoods $\Psi_{\star}$ on an electrical network $G = (V,E,\w)$ with $s,t \in V$. Then there exists an alternative potential vector $\palt$ that assigns a potential $\palt_{u,v}$ on each edge $(u,v) \in E$ such that the associated state $\ket{\palt}$ (see \eq{potential-edge}) satisfies $\Pi_{\Balt}\ket{\palt} = \ket{\palt}$ and the potential difference between $(u,v)$ and $(v,u)$ is equal to the amount of electrical flow $\falt_{u,v}$ along $(u,v)$ multiplied with the resistance $1/\w_{u,v}$, that is, $\palt_{u,v} - \palt_{v,u} = \falt_{u,v}/\w_{u,v}$.
\end{definition}

We have not yet introduced the necessarily tools to show that there always exists a potential vector $\palt$ satisfying Alternative Ohm's Law, we will do this in \thm{potential-alt}. Note that if it were not for the extra condition $\Pi_{\Balt}\ket{\palt} = \ket{\palt}$, there is no interplay between the variables $\palt_{u,v},\palt_{v,u}$ for different edges and one could always find an alternate potential satisfying Alternative Ohm's Law. However, we shall see in \thm{potential-alt} that the condition $\Pi_{\Balt}\ket{\palt} = \ket{\palt}$ gives rise to a unique alternative potential. In the following examples and applications, we therefore show existence by explicitly constructing $\ket{\palt}$. If the potential vector $\palt$ satisfies Alternative Ohm's Law, then $\ket{\palt}$ is precisely the state we need to apply \lem{qwflows}:

\begin{align*}
  \ket{\falt} &= \frac{1}{\sqrt{2{\cal R}_{s,t}^{\alt}}} \sum_{(u,v)\in \E} \frac{\f_{u,v}}{\sqrt{\w_{u,v}}} (\ket{u,v}+\ket{v,u}) \\
  &= \frac{1}{\sqrt{2{\cal R}_{s,t}^{\alt}}} \sum_{(u,v)\in \E} \left(\sqrt{\w_{u,v}}(\palt_{u,v}-\palt_{v,u})\ket{u,v} + \sqrt{\w_{u,v}}(\palt_{u,v}-\palt_{v,u})\ket{v,u}\right) \\
  &= \frac{1}{\sqrt{2{\cal R}_{s,t}^{\alt}}} \left(\sum_{(u,v)\in \E} \sqrt{\w_{u,v}}(\palt_{u,v}\ket{u,v} - \palt_{v,u} \ket{v,u}) + {\sf SWAP}\sum_{(u,v)\in \E} \sqrt{\w_{u,v}}(\palt_{u,v}\ket{u,v} - \palt_{v,u} \ket{v,u})\right) \\
  &= (I-\Pi_{\cal A}) \sqrt{\frac{2}{{\cal R}_{s,t}^{\alt}}} \sum_{(u,v)\in \E} \sqrt{\w_{u,v}}(\palt_{u,v}\ket{u,v} - \palt_{v,u} \ket{v,u})\\
  &= (I-\Pi_{\cal A})\ket{\palt} + (I-\Pi_{\cal A}) \sqrt{\frac{2}{{\cal R}_{s,t}^{\alt}}}\sum_{v\in\N(s)} (-1)^{\Delta_{s,v}}\palt_{s,v}\sqrt{\w_{s,v}}\ket{s,v} \\
  &= (I-\Pi_{\cal A})\ket{\palt} + \sqrt{{{\cal R}_{s,t}^{\alt}}\du_s}\ket{\psi_s^+}. 
\end{align*}

In the following examples and applications where we explicitly construct the state $\ket{\palt}$, we need to verify that it satisfies $\Pi_{\Balt}\ket{\palt} = \ket{\palt}$. To assist in this verification, we introduce the states 
$$\ket{\palt_{|u}}= \sqrt{\frac{2}{{\cal R}(\falt)}}(\ket{u}\bra{u}\otimes I)\ket{\palt}$$
for $u \in V \setminus \{s\}$. To verify whether $\Pi_{\Balt}\ket{\palt}= \ket{\palt}$, it will be sufficient to verify whether each $\ket{\palt_{|u}}$ lies in $\mathrm{span}\{\Psi_\star(u)\}$, since we can decompose $\ket{\palt}$ as
\begin{equation}\label{eq:pot-reduced}
\begin{split}
  \ket{\palt}&= \sqrt{\frac{2}{{\cal R}(\falt)}} \sum_{u \in V \setminus \{s\}}\sum_{v \in \N(u)}(-1)^{\Delta_{u,v}} \palt_{u,v} \sqrt{\w_{u,v}} \ket{u,v}\\
  &=\sqrt{\frac{2}{{\cal R}(\falt)}} \sum_{u \in V \setminus \{s\}} \ket{\palt_{|u}}.
\end{split}
\end{equation}

In the special case where $u$ has no additional alternative neighbourhoods, for $\ket{\palt_{|u}}$ to lay in $\mathrm{span}\{\Psi_\star(u)\} = \mathrm{span}\{\ket{\psi_u}\}$, the edge potentials $\p_{u,v}$ must be the same for each $v \in \N(u)$.

\subsection{Examples} \label{sec:examples}

Having rebuilt the connection between the alternative potential vector and \st alternative electrical flow in the multidimensional quantum electrical network framework, we now provide some intuition for these new definitions by providing a few examples.

Consider the network $G = (V,E,\w)$ with the vertex set $V = \{s,x,y,t\}$ and directed edge set $\E = \{(s,x),(x,y),(x,t),(y,t)\}$, where each edge $(u,v) \in \E$ has weight $\w_{u,v} = 1/4$, except for the edge $(s,x)$, which has weight $\w_{s,x}=1$. This is visualised in \fig{normal}. These directions and weight assignments give rise to the following star states for each of our $4$ vertices:
\begin{align*}
  &\ket{\psi_s} = \ket{s,x}, &&\ket{\psi_x} =\sqrt{\frac{2}{3}} \left(-\ket{x,s} + \frac{1}{2}\ket{x,y} + \frac{1}{2}\ket{x,t}\right),\\
  &\ket{\psi_y} =\sqrt{2}\left( -\frac{1}{2}\ket{y,x} + \frac{1}{2}\ket{y,t}\right), &&\ket{\psi_t} =\sqrt{2}\left( -\frac{1}{2}\ket{t,x} - \frac{1}{2}\ket{t,y}\right).
\end{align*}

In \fig{normal} we show the \st electrical flow $\f$ on $G$ and the corresponding potential vector $\p$. It is straightforward to verify that $\f$ and $\p$ satisfy Ohm's Law, meaning $\p_u - \p_v = \frac{\f_{u,v}}{\w_{u,v}}$. 

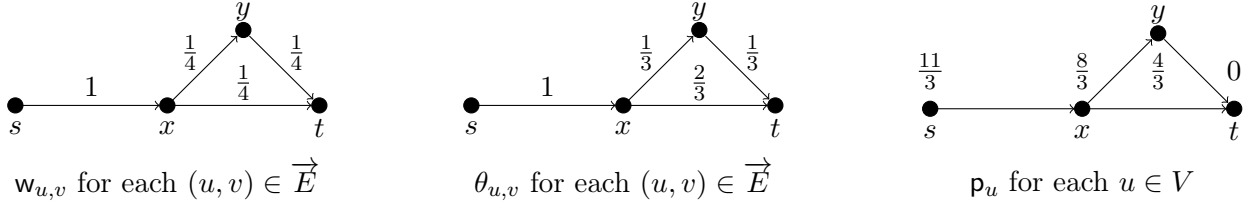
\begin{figure}
\centering
\begin{tikzpicture}
\node at (0,0) {\begin{tikzpicture}
  \draw[->] (-2,0)--(-0.1,0);
  \draw[->] (0,0)--(1.9,0);
  \draw[->] (1.08,.92)--(1.92,.08);
  \draw[->] (0,0)--(.92,.92);
  
  \filldraw (-2,0) circle (.1);
  \filldraw (0,0) circle (.1);
  \filldraw (2,0) circle (.1);
  \filldraw (1,1) circle (.1);
  
  \node at (-2,-0.3) {$s$};
  \node at (0,-0.3) {$x$};
  \node at (1,1.25) {$y$};
  \node at (2,-0.3) {$t$};

  \node at (-1,0.25) {$1$};
  \node at (1,0.3) {$\frac{1}{4}$};
  \node at (0.3,0.7) {$\frac{1}{4}$};
  \node at (1.7,0.7) {$\frac{1}{4}$};

  \node at (0,-1) {$\w_{u,v}$ for each $(u,v) \in \E$};
  \end{tikzpicture}};
  
  \node at (6,0) {\begin{tikzpicture}
  \draw[->] (-2,0)--(-0.1,0);
  \draw[->] (0,0)--(1.9,0);
  \draw[->] (1.08,.92)--(1.92,.08);
  \draw[->] (0,0)--(.92,.92);
  
  \filldraw (-2,0) circle (.1);
  \filldraw (0,0) circle (.1);
  \filldraw (2,0) circle (.1);
  \filldraw (1,1) circle (.1);
  
  \node at (-2,-0.3) {$s$};
  \node at (0,-0.3) {$x$};
  \node at (1,1.25) {$y$};
  \node at (2,-0.3) {$t$};

  \node at (-1,0.25) {$1$};
  \node at (1,0.3) {$\frac{2}{3}$};
  \node at (0.3,0.7) {$\frac{1}{3}$};
  \node at (1.7,0.7) {$\frac{1}{3}$};
  \node at (0,-1) {$\f_{u,v}$ for each $(u,v) \in \E$};
  \end{tikzpicture}};

  \node at (12,0){\begin{tikzpicture}
  \draw[->] (-2,0)--(-0.1,0);
  \draw[->] (0,0)--(1.9,0);
  \draw[->] (1.08,.92)--(1.92,.08);
  \draw[->] (0,0)--(.92,.92);
  
  \filldraw (-2,0) circle (.1);
  \filldraw (0,0) circle (.1);
  \filldraw (2,0) circle (.1);
  \filldraw (1,1) circle (.1);
  
  \node at (-2,-0.3) {$s$};
  \node at (0,-0.3) {$x$};
  \node at (1,1.25) {$y$};
  \node at (2,-0.3) {$t$};

  \node at (-2,0.5) {$\frac{11}{3}$};
  \node at (-0,0.5) {$\frac{8}{3}$};
  \node at (1,0.5) {$\frac{4}{3}$};
  \node at (2,0.5) {$0$};
  \node at (0,-1) {$\p_u$ for each $u \in V$};
  \end{tikzpicture}};
  
\end{tikzpicture}
\caption{Graph $G$ with its \st electrical flow $\f$ and corresponding potential $\p$ at each vertex.}\label{fig:normal}
\end{figure}

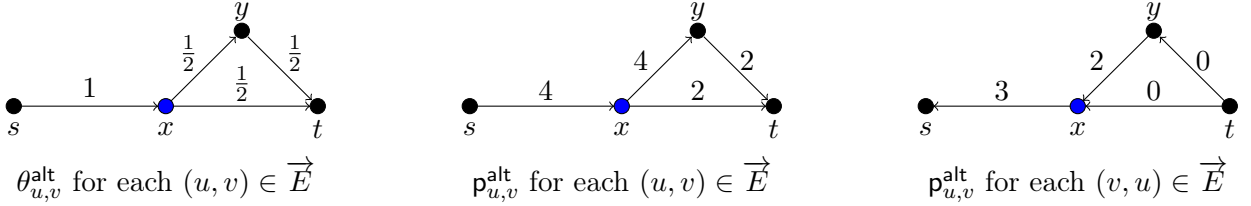
\begin{figure}
\centering
\begin{tikzpicture}
\node at (0,0) {\begin{tikzpicture}
  \draw[->] (-2,0)--(-0.1,0);
  \draw[->] (0,0)--(1.9,0);
  \draw[->] (1.08,.92)--(1.92,.08);
  \draw[->] (0,0)--(.92,.92);
  
  \filldraw (-2,0) circle (.1);
  \filldraw[fill=blue] (0,0) circle (.1);
  \filldraw (2,0) circle (.1);
  \filldraw (1,1) circle (.1);
  
  \node at (-2,-0.3) {$s$};
  \node at (0,-0.3) {$x$};
  \node at (1,1.25) {$y$};
  \node at (2,-0.3) {$t$};

  \node at (-1,0.25) {$1$};
  \node at (1,0.3) {$\frac{1}{2}$};
  \node at (0.3,0.7) {$\frac{1}{2}$};
  \node at (1.7,0.7) {$\frac{1}{2}$};

  \node at (0,-1) {$\falt_{u,v}$ for each $(u,v) \in \E$};
  \end{tikzpicture}};
  
  \node at (6,0) {\begin{tikzpicture}
  \draw[->] (-2,0)--(-0.1,0);
  \draw[->] (0,0)--(1.9,0);
  \draw[->] (1.08,.92)--(1.92,.08);
  \draw[->] (0,0)--(.92,.92);
  
  \filldraw (-2,0) circle (.1);
  \filldraw[fill=blue] (0,0) circle (.1);
  \filldraw (2,0) circle (.1);
  \filldraw (1,1) circle (.1);
  
  \node at (-2,-0.3) {$s$};
  \node at (0,-0.3) {$x$};
  \node at (1,1.25) {$y$};
  \node at (2,-0.3) {$t$};

  \node at (-1,0.2) {$4$};
  \node at (1.65,0.6) {$2$};
  \node at (1,0.2) {$2$};
  \node at (0.25,0.6) {$4$};

  \node at (0,-1) {$\palt_{u,v}$ for each $(u,v) \in \E$};
  \end{tikzpicture}};  

  \node at (12,0) {\begin{tikzpicture}
  \draw[->] (0.0,0)--(-1.9,0);
  \draw[->] (2,0)--(0.1,0);
  \draw[->] (2,0)--(1.08,.92);
  \draw[->] (1,1)--(0.07,0.07);
  
  \filldraw (-2,0) circle (.1);
  \filldraw[fill=blue] (0,0) circle (.1);
  \filldraw (2,0) circle (.1);
  \filldraw (1,1) circle (.1);
  
  \node at (-2,-0.3) {$s$};
  \node at (0,-0.3) {$x$};
  \node at (1,1.25) {$y$};
  \node at (2,-0.3) {$t$};

  \node at (-1,0.2) {$3$};
  \node at (1.65,0.6) {$0$};
  \node at (1,0.2) {$0$};
  \node at (0.25,0.6) {$2$};
  
  \node at (0,-1) {$\palt_{u,v}$ for each $(v,u) \in \E$};
  \end{tikzpicture}}; 
  
\end{tikzpicture}
\caption{Graph $G$ where the blue vertex $x$ has an additional alternative neighbourhood. The \st alternative electrical flow $\falt$ with respect to this extra alternative neighbourhood is displayed, as well as the corresponding potential vector $\palt$. }\label{fig:alternative}
\end{figure}

We now consider the case where only the vertex $x \in V$ contains an additional alternative neighbourhood: let $\Psi_{\star}(x) = \{\ket{\psi_x},\ket{\psi_x^{\alt}}\}$ where
$$ \ket{\psi_x^{\alt}} = \sqrt{\frac{2}{3}} (\frac{1}{2}\ket{x,s} -\ket{x,y} + \frac{1}{2}\ket{x,t}),$$ 
visualised in \fig{alternative}. Alternative Kirchhoff's Law states that the flow state $\ket{\falt}$ of any unit \st flow $\falt$ must additionally be orthogonal to $\ket{\psi_x^{\textrm{alt}}}$. Together with being orthogonal to all the star states, meaning that the flow $\falt$ is conserved at the vertices $x$ and $y$, this leaves us with only a single option for $\falt$. This flow is visualised in \fig{alternative} and the corresponding flow vector is given by
\begin{align*}
  \ket{\falt} &= \frac{1}{\sqrt{2{\cal R}_{s,t}^{\alt}}}\sum_{(u,v) \in \E}\frac{\f_{u,v}}{\sqrt{\w_{u,v}}}\left(\ket{u,v} + \ket{v,u}\right) \\
  &= \frac{1}{\sqrt{8}}\left(\frac{1}{1}\left(\ket{s,x} + \ket{x,s}\right) + \frac{1/2}{1/2}\left(\ket{x,y} + \ket{y,x}\right) + \frac{1/2}{1/2}\left(\ket{x,t} + \ket{y,t}\right) + \frac{1/2}{1/2}\left(\ket{y,t} + \ket{t,y}\right)\right) \\
  &= \frac{1}{\sqrt{8}}\left(\ket{s,x} + \ket{x,s} + \ket{x,y} + \ket{y,x} + \ket{x,t} + \ket{t,x} + \ket{y,t} + \ket{t,y}\right)
\end{align*}

Since this $\falt$ is the only unit \st flow satisfying Alternative Kirchhoff's Law, it is by default the \st alternative electrical flow. For its alternative potential vector $\palt$, we construct $\ket{\palt}$ from the bottom up by creating the states from \eq{pot-reduced}:
\begin{align*}
  &\ket{\palt_{|s}} = 4 \ket{s,x}, &&\ket{\palt_{|x}} = - 3\ket{x,s} + 4\sqrt{\frac{1}{4}}\ket{x,y} + 2 \sqrt{\frac{1}{4}} \ket{x,t},\\
  &\ket{\palt_{|y}} = - 2\sqrt{\frac{1}{4}} \ket{y,x} + 2\sqrt{\frac{1}{4}} \ket{y,t}, &&\ket{\palt_{|t}} = - 0 \sqrt{\frac{1}{4}}\ket{t,x} - 0 \sqrt{\frac{1}{4}}\ket{t,y}.
\end{align*}
Note that constructing $\ket{\palt_{|s}}$ was not necessary, but we have added it for completeness. Each such $\ket{\palt_{|u}}$ for $u \in V \setminus \{s\}$ lies in $\mathrm{span}\{\Psi_\star(u)\}$ respectively. The alternative potential $\palt$ (see \fig{alternative} for all the edge potentials) satisfies $\palt_{s,x} = {\cal R}_{s,t}^{\alt} = 4$ and $\palt_{t,x} = \palt_{t,y} = 0$, as well as Alternative Ohm's Law, meaning that each $(u,v) \in E)$ satisfies $\palt_{u,v} - \palt_{v,u} = \falt_{u,v}/\w_{u,v}$. We have therefore found the alternative potential vector $\palt$ whose associated state $\ket{\palt}$ satisfies $\Pi_{\Balt}\ket{\palt} = \ket{\palt}$:
\begin{align*}
  \ket{\palt} &=\sqrt{\frac{2}{{\cal R}(\falt)}} \sum_{u \in V \setminus \{s\}}\sum_{v \in \N(u)}(-1)^{\Delta_{u,v}} \palt_{u,v} \sqrt{\w_{u,v}} \ket{u,v} \\
  &= -3\ket{x,s} + 2\ket{x,y}+ \ket{x,t} - \ket{y,x} + \ket{y,t} \\
  &= -\sqrt{\frac{3}{2}}\left(\frac{8}{3}\ket{\psi_x} + \frac{2}{3}\ket{\psi_x^{\alt}}\right) + \sqrt{2}\ket{\psi_y} + 0\ket{\psi_t}.
\end{align*}

As mentioned in \sec{kal}, depending on the alternative neighbourhoods in $\Psi_{\star}$, the \st alternative electrical flow might not exist, which is in contrast with regular electrical networks. As such a counterexample, we modify $G$ once more, this time removing the edge $(y,t)$ from $\E$. It is clear that any unit \st flow $\falt$ must satisfy $\falt_{s,x} = \falt_{x,t} = 1$ and $\falt_{x,y} = 0$, but in doing so, it will not satisfy Alternative Kirchhoff's Law, as the associated state $\ket{\falt}$ is not orthogonal to $\ket{\psi_x^{\alt}}$:
$$ \brakett{\psi_x^{\alt}}{\falt} = \frac{1}{\sqrt{2{\cal R}(\falt)}}\sqrt{\frac{2}{3}} = \sqrt{\frac{1}{6}}.$$
 
 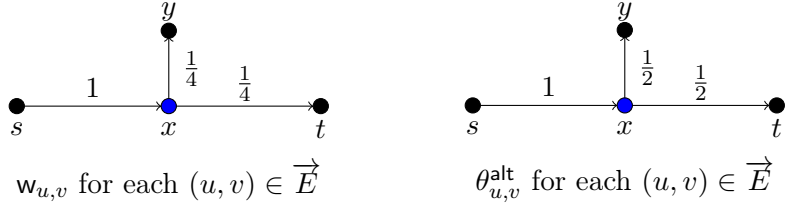
\begin{figure}
\centering
\begin{tikzpicture}
\node at (0,0) {\begin{tikzpicture}
  \draw[->] (-2,0)--(-0.1,0);
  \draw[->] (0,0)--(1.9,0);
  \draw[->] (0,0)--(0,.92);
  
  \filldraw (-2,0) circle (.1);
  \filldraw[fill=blue] (0,0) circle (.1);
  \filldraw (2,0) circle (.1);
  \filldraw (0,1) circle (.1);
  
  \node at (-2,-0.3) {$s$};
  \node at (0,-0.3) {$x$};
  \node at (0,1.25) {$y$};
  \node at (2,-0.3) {$t$};

  \node at (-1,0.25) {$1$};
  \node at (1,0.3) {$\frac{1}{4}$};
  \node at (0.3,0.5) {$\frac{1}{4}$};
  \node at (0,-1) {$\w_{u,v}$ for each $(u,v) \in \E$};
  \end{tikzpicture}};
   \node at (6,0) { \begin{tikzpicture}
  \draw[->] (-2,0)--(-0.1,0);
  \draw[->] (0,0)--(1.9,0);
  \draw[->] (0,0)--(0,.92);
  
  \filldraw (-2,0) circle (.1);
  \filldraw[fill=blue] (0,0) circle (.1);
  \filldraw (2,0) circle (.1);
  \filldraw (0,1) circle (.1);
  
  \node at (-2,-0.3) {$s$};
  \node at (0,-0.3) {$x$};
  \node at (0,1.25) {$y$};
  \node at (2,-0.3) {$t$};

  \node at (-1,0.25) {$1$};
  \node at (1,0.3) {$\frac{1}{2}$};
  \node at (0.3,0.5) {$\frac{1}{2}$};
  \node at (0,-1) {$\falt_{u,v}$ for each $(u,v) \in \E$};
  \end{tikzpicture}};
  
\end{tikzpicture}
\caption{Graph $G$ where the blue vertex $x$ has an additional alternative neighbourhood $\ket{\psi_x^{\textrm{alt}}}$. There is no unit flow from $s$ to $t$ satisfying Alternative Kirchhoff's Law possible in this graph.}\label{fig:counterexam}
\end{figure}

\renewcommand*{\arraystretch}{1.5}

\subsection{The alternative incidence matrix, Alternative Kirchhoff's Law and Alternative Ohm's Law} \label{sec:altincid}

In this section, inspired by the connection between the electrical network $G=(V,E,\w)$ and the incidence matrix $B$ of $G$ from \sec{inc}, we rebuild the connection between the multidimensional electrical network and its alternative incidence matrix $B_{\alt}$. We then use this connection to prove the uniqueness of the \st alternative flow $\falt$ and the existence of the alternative potential $\palt$ that satisfy Alternative Ohm's Law. 

What is it that makes the \st electrical flow $\f$ special, making it satisfy Ohm's Law. Why is Ohm's Law not necessarily true for our \st alternative flow? Even though all flow states live in the symmetric subspace ${\cal A}^{\perp}$ by construction, we saw in \eq{flow-projA} that the flow state $\ket{\f}$ of the \st electrical flow $\f$ can be written as 
$$ \ket{\f} = (I-\Pi_{\cal A})\sqrt{\frac{2}{{\cal R}_{s,t}}} \sum_{u \in V} \p_u\sqrt{\du_u}\ket{\psi_u},$$
meaning that $\ket{\f}$ in fact lives in the \textit{the symmetric star space} of ${\cal H}$, which is contained in ${\cal A}^{\perp}$:
\begin{align}\label{eq:star-space}
  H^{+\star} := \mathrm{span}\{(I - \Pi_{\cal A})\ket{\psi_u}: u \in V\} = \mathrm{span}\{\ket{\psi_u^+}: u \in V\}.
\end{align}

Out of all \st flows, the \st electrical flow is the unique unit flows such that $\ket{\f}$ is the only corresponding flow state that is an element of $H^{+\star}$ (see e.g. \cite{MR3616205}). We will not give a formal proof of this statement, but the intuition is that any other \st flow has a higher energy, i.e. higher norm, which is due to containing a component that is orthogonal to all of $H^{+\star}$, namely a circulation. The column space of the incidence matrix $B$ is in fact isomorphic to $H^{+\star}$, where the column of $B$ indexed by $u \in V$ represents $\sqrt{\du_u}\ket{\psi^+_u}$ through the isometry
\begin{equation}\label{eq:isometry}
  {\cal V}: \mathbb{C}^{|\E|} \mapsto {\cal A}^{\perp}, \text{ where } {\cal V}\left(u,v\right) = \sqrt{2}(I - \Pi_{\cal A})\ket{u,v} = \frac{1}{\sqrt{2}}\left(\ket{u,v} + \ket{v,u}\right). 
\end{equation}

Through the addition of alternative neighbourhoods (see \defin{alternative}), the space $H_G^{+\star}$ is effectively enlarged. Define
\begin{equation}
  V^{\alt} := \{(u,i) \in V \times \mathbb{N}: i \in \{0,1,\cdots,a_u-1\}\}.
\end{equation}
Instead of only considering the span of all $\ket{\psi_u^+}$ for $u \in V$, we now consider the span of all alternative neighbourhoods projected onto the symmetric subspace, meaning $\ket{\psi_{u,i}^+} := \sqrt{2}(I - \Pi_{\cal A})\ket{\psi_{u,i}}$ for $(u,i) \in V^{\alt}$:
\begin{align}\label{eq:alt-space}
  H^{+\alt} := \mathrm{span}\{\ket{\psi_{u,i}^+}: u \in V, i\in \{0,1,\cdots,a_u-1\}\}.
\end{align}

By modifying the incidence matrix $B$ to ensure that its column space still represents the newly modified $H_G^{+\alt}$, we obtain the alternative incidence matrix $B_{\alt}$.
\begin{definition}[Alternative incidence matrix]\label{def:incidence-alt}
  Let $G$ be a network and let $\Psi_{\star}$ be a collection of alternative neighbourhoods. Let $\{\ket{{\psi}_{u,0}},\dots,\ket{{\psi}_{u,a_u-1}}\}$ be an orthonormal basis for each $\Psi_{\star}(u) \in \Psi_{\star}$. The \textit{alternative incidence matrix} $B_{\alt} \in \mathbb{C}^{\E \times V^{\alt}}$ of $G$ is the matrix whose rows range over $(u,v) \in \E$, whose columns range over $(u,i)\in V^{\alt}$ and whose only non-zero entries are of the form
  \begin{align*}
    &{B_{\alt}}_{(u,v),(u,i)} = \sqrt{\du_u}\brakett{u,v}{\psi_{u,i}}, &{B_{\alt}}_{(u,v),(v,j)} = \sqrt{\du_v}\brakett{u,v}{\psi_{v,j}}.
  \end{align*}
\end{definition}

By \defin{alternative} we may assume that each $\ket{{\psi}_{u,i}}$ only has real coefficients and that $\ket{{\psi}_{u,0}} = \ket{\psi_u}$. By substituting $B$ with $B_{\alt}$ in both \eq{matrix-knl} and \eq{matrix-ohm}, we can recover both Alternative Kirchhoff's Law and Alternative Ohm's Law, showing that these are indeed their natural definitions with respect to perspective of the incidence matrix. Fix some ordering of the columns of $B$ of the form $s,(u_1,i_1),\dots,(u_2,i_2),t$ for some $u_1,u_2 \in V \backslash \{s,t\}$ such that $(u_1,i_1),(u_2,i_2) \in V^{\alt}$.

\begin{definition}[Alternative Kirchhoff's Law (incidence matrix)]\label{defin:knl-inc-alt} 
Let $\falt$ be any alternative unit \st flow on an electrical network $G = (V,E,\w)$ with respect to a collection of alternative neighbourhoods $\Psi_{\star}$. Let $B_{\alt}$ be the alternative incidence matrix of $G$. Then $\falt$ satisfies
  \begin{equation}\label{eq:matrix-knl-alt}
  B_{\alt}^{T} W\falt = \begin{bmatrix}
   \sum_{v \in \Gamma(s)} \falt_{s,v} \\ \sum_{v \in \Gamma(u_1)}\frac{\falt_{u_1,v}}{\sqrt{\w_{u_1,v}}}\sqrt{\du_{u_1}}\brakett{u_1,v}{\psi_{u_1,i_1}}\\ \vdots \\\sum_{v \in \Gamma(u_2)}\frac{\falt_{u_2,v}}{\sqrt{\w_{u_2,v}}}\sqrt{\du_{u_2}}\brakett{u_2,v}{\psi_{u_2,i_2}} \\ \sum_{v \in \Gamma(t)} \falt_{t,v} \end{bmatrix} = 
   \begin{bmatrix}
  1 \\ 0 \\ \vdots \\ 0 \\ -1 \end{bmatrix} = {\sf e}_s - {\sf e}_t.
  \end{equation}
\end{definition}

Recall from \defin{flow-alt} that the \st alternative electrical flow is the flow that minimises ${\cal E}(\falt)$ for all alternative unit \st flows $\falt$ (if any such flow exists). By applying the Moore-Penrose inverse of $B_{\alt}^T$ to \eq{matrix-knl-alt}, we prove that the \st electrical flow is unique and thus well defined:

\begin{theorem}\label{thm:elecflow-alt}
  Let $\falt$ be the \st alternative electrical flow on a network $G = (V,E,\w)$. Let $B_{\alt}$ be the alternative incidence matrix of $G$. Then $W\falt$ is given by 
  \begin{equation}\label{eq:transpose-alt}
    W\falt = B_{\alt}^{T+}({\sf e}_s - {\sf e}_t).
  \end{equation}
\end{theorem}

Recall the isometry ${\cal V}$ defined in \eq{isometry}. The column of $B_{\alt}$ indexed by $(u,i) \in V^{\alt}$ is equal to ${\cal V}^{-1}\left(\sqrt{\du_u}\ket{\psi_{u,i}^+}\right)$, meaning that the column space of $B_{\alt}$ is equal to ${\cal V}^{-1}\left(H^{+\alt}\right)$. Moreover, the column space of $B_{\alt}$ is equal to the column space of $B_{\alt}^{T+}$, due to the properties of the Moore-Penrose inverse in \eq{property}. Combined with the fact that the state $\ket{\falt}$ is related to the vector $W\falt$ via the equality $\sqrt{{\cal R}_{s,t}^{\alt}}\ket{\falt} = {\cal V}\left(W\falt\right)$, we find that $\ket{\falt}$ is an element of $H^{+\alt}$. This means that there exist coefficients $\palt_{(u,i)}$ such that 
\begin{equation}\label{eq:flow-alt-projA}
\begin{split}
  \ket{\falt} = \frac{1}{\sqrt{{{\cal R}^{\alt}_{s,t}}}} \sum_{u \in V} \sum_{i=0}^{a_u-1}\palt_{(u,i)}\sqrt{\du_u}\ket{\psi_{u,i}^+}.
\end{split}
\end{equation}

The notation $\palt_{(u,i)}$ seems to hint that these coefficients are related to the alternative potential vector $\palt$. This is indeed the case: by defining the potential vector $\palt$ as
\begin{equation}\label{eq:pot-edge}
  \palt_{u,v} := \frac{(-1)^{\Delta_{u,v}}}{\sqrt{\w_{u,v}}}\sum_{i=0}^{a_u-1}\palt_{(u,i)}\sqrt{\du_u}\brakett{u,v}{\psi_{u,i}},
\end{equation}
we guarantee that the state $\ket{\palt}$ satisfies $\Pi_{\Balt}\ket{\palt}$:
\begin{equation*}
\begin{split}
  \ket{\palt} &=\sqrt{\frac{2}{{\cal R}(\falt)}} \sum_{u \in V \setminus \{s\}}\sum_{v \in \N(u)}(-1)^{\Delta_{u,v}} \palt_{u,v} \sqrt{\w_{u,v}} \ket{u,v} \\
  &= \sqrt{\frac{2}{{\cal R}^{\alt}_{s,t}}} \sum_{u \in V \setminus \{s\}}\sum_{v\in \N(u)} \sum_{i=0}^{a_u-1}\palt_{(u,i)}\sqrt{\du_u}\brakett{u,v}{\psi_{u,i}}\ket{u,v} \\
  &= \sqrt{\frac{2}{{\cal R}^{\alt}_{s,t}}} \sum_{u \in V \setminus \{s\}} \sum_{i=0}^{a_u-1}\palt_{(u,i)}\sqrt{\du_u}\ket{\psi_{u,i}}.
\end{split}
\end{equation*}

Due to the coefficients $\palt_{(u,i)}$, we can therefore consider the alternative potential vector $\palt$ as a vector in $\mathbb{C}^{V^{\alt}}$ with entries $\palt_{(u,i)}$ for the row indexed by $(u,i) \in V^{\alt}$. By substituting $B$ with $B_{\alt}$ in \eq{matrix-ohm} and combining this with \eq{pot-edge}, we recover Alternative Ohm's Law:

\begin{definition}[Alternative Ohm's Law (incidence matrix)]\label{defin:ohm-inc-alt} 
Let $\falt$ be any alternative unit \st flow on an electrical network $G = (V,E,\w)$ with respect to a collection of alternative neighbourhoods $\Psi_{\star}$. Let $B_{\alt}$ be the alternative incidence matrix of $G$. Then there exists an alternative potential vector $\palt$ such that $\Pi_{\Balt}\ket{\palt} = \ket{\palt}$ and
\begin{equation}\label{eq:matrix-ohm-alt}
  B_{\alt}\palt = \begin{bmatrix}
   \sqrt{\w_{u_1,v_1}}\left(\palt_{u_1,v_1}-\palt_{v_1,u_1}\right)\\ \vdots \\ \sqrt{\w_{u_2,v_2}}\left(\palt_{u_2,v_2}-\palt_{v_2,u_2}\right)
  \end{bmatrix} = \begin{bmatrix}
   \frac{\f_{u_1,v_1}}{\sqrt{\w_{u_1,v_1}}}\\ \vdots \\ \frac{\f_{u_2,v_2}}{\sqrt{\w_{u_2,v_2}}}
  \end{bmatrix}=W\falt.
\end{equation}
\end{definition}

Just like with the potential vector $\p$, we may assume that the alternative potential vector $\palt$ satisfying Alternative Ohm's Law also satisfies $\palt_s = {\cal R}_{s,t}^{\alt}$ and $\palt_t = 0$

\begin{theorem}\label{thm:potential-alt}
   Let $\falt$ be the \st alternative electrical flow on an electrical network $G = (V,E,\w)$ with respect to a collection of alternative neighbourhoods $\Psi_{\star}$. Let $B_{\alt}$ be the alternative incidence matrix of $G$. Then there exists a unique alternative potential vector $\palt$ satisfying Alternative Ohm's Law such that $\palt_s = {\cal R}_{s,t}^{\alt}$ and $\palt_t = 0$. 
\end{theorem}
\begin{proof}
  Recall from \lem{potential} that $B$ has $\sum_{u \in V}e_u$ as an $0$-eigenvector. It is therefore straightforward to see that $\sum_{u \in V}e_{u,0}$ is a $0$-eigenvector of $B_{\alt}$. This allows us to apply the same trick as in \eq{matrix-pot}, meaning we remove the last column of $B_{\alt}$ and last row of $\palt$ to obtain $\overline{B_{\alt}}$ and $\overline{\palt}$, forcing $\palt_t = 0$ for the solution satisfying \eq{matrix-ohm-alt}:
  \begin{equation}\label{eq:matrix-pot-alt}
    \palt = \begin{bmatrix} \overline{\palt} \\ 0 \end{bmatrix} = \begin{bmatrix} \overline{B_{\alt}}^{+}W\falt \\ 0  \end{bmatrix}.
  \end{equation}

  \noindent By left-multiplying both sides of \eq{transpose-alt} with $(W{\falt})^T$ we obtain together with \eq{matrix-ohm-alt} that
  \begin{equation}\label{eq:potential-s-alt}
    {\cal R}_{s,t}^{\alt} = \norm{W\falt}^2 = {(W\falt)}^T B^{T+}({\sf e}_s - {\sf e}_t) = {\palt}^T({\sf e}_s - {\sf e}_t) = \palt_s - \palt_t = \palt_s.
  \end{equation}  
\end{proof}

Due to \thm{potential-alt}, we may now apply \lem{qwflows} with the parameters $\ket{\psi} = \ket{\psi_s^+}$, $\ket{\varphi} = \ket{\falt}$, $\ket{\phi} = -\frac{1}{\sqrt{{\cal R}_{s,t}^{\alt} \du_s}}\ket{\palt}$ and $p = \frac{1}{{\cal R}_{s,t}^{\alt} \du_s}$, proving the following generalisation of \cor{qwflows}:
\begin{theorem}\label{thm:qwflows}
  Let $\Psi_{\star}$ be a collection of alternative neighbourhoods on a network $G = (V,E,\w)$ and let $U_{\A\Balt}$ be the quantum walk operator with respect to $\Psi_{\star}$ as defined in \eq{walk-alt}. Then by performing phase estimation on the initial state $\ket{\psi_s^+}$ with the operator $U_{\A\Balt}$ and precision $O\left(\frac{\epsilon^2}{\sqrt{{\cal R}_{s,t}^{\alt} \du_s}\norm{\ket{\palt}}}\right)$, the phase estimation algorithm outputs ``$0$'' with probability $\Theta\left(\frac{1}{{\cal R}_{s,t}^{\alt} \du_s}\right)$, leaving a state $\ket{\f'}$ satisfying
  $$ \frac{1}{2}\norm{\proj{\f'} - \proj{\falt}}_1 \leq \epsilon. $$
\end{theorem}


\subsection{Examples}

We will now show how these results apply to the examples \fig{normal} and \fig{alternative} from section \sec{examples}, which we have restated here in \fig{normal-2} and \fig{alternative-2}. Consider the graph $G$, consisting of the vertex set $V = \{s,x,y,t\}$ and directed edge set $\E = \{(s,x),(x,y),(x,t),(y,t)\}$, where each edge $(u,v) \in \E$ has weight $\w_{u,v} = 1/4$, except for the edge $(s,x)$, which has weight $\w_{s,x}=1$. This graph is visualised in \fig{normal-2}. These directions and weight assignments give rise to the following star states for each of our $4$ vertices:
\begin{align*}
  &\ket{\psi_s} = \ket{s,x}, &&\ket{\psi_x} =\sqrt{\frac{2}{3}}\left(-\ket{x,s} + \frac{1}{2}\ket{x,y} + \frac{1}{2}\ket{x,t}\right),\\
  &\ket{\psi_y} =\sqrt{2}\left(-\frac{1}{2}\ket{y,x} + \frac{1}{2}\ket{y,t}\right), &&\ket{\psi_t} =\sqrt{2}\left(-\frac{1}{2}\ket{t,x} - \frac{1}{2}\ket{t,y}\right).
\end{align*}

By ordering the directed edges as $(s,x),(x,y),(x,t),(y,t)$ and the vertices as $s,x,y,t$, we have that the incidence matrix $B$ of $G$ and the Moore-Penrose inverse $B^{T+}$ of its transpose are equal to
\begin{align}
  &B = \begin{bmatrix}
   1 & -1 & 0 & 0\\ 0 & \frac{1}{2} & -\frac{1}{2} & 0 \\ 0 & \frac{1}{2} & 0 & -\frac{1}{2} \\ 0 & 0 & \frac{1}{2} & -\frac{1}{2}
  \end{bmatrix}, &B^{T+} = \begin{bmatrix}
   \frac{3}{4} & -\frac{1}{4} & -\frac{1}{4} & -\frac{1}{4}\\ \frac{1}{2} & \frac{1}{2} & -\frac{5}{6} & -\frac{1}{6} \\ \frac{1}{2} & \frac{1}{2} & -\frac{1}{6} & -\frac{5}{6} \\ 0 & 0 & \frac{2}{3} & -\frac{2}{3}
  \end{bmatrix}.
\end{align}

\noindent The weighted diagonal matrix $W$ is given by
\begin{align}
  &W = \begin{bmatrix}
   1 & 0 & 0 & 0\\ 0 & \frac{1}{2} & 0 & 0 \\ 0 & 0 & \frac{1}{2}& 0 \\ 0 & 0 & 0 & \frac{1}{2}
  \end{bmatrix} .
\end{align}
We can recover the electrical flow $\f$ in \fig{normal-2} using \thm{elecflow} to derive 
\begin{align*}
  &W\f = \begin{bmatrix}
  \frac{\f_{s,x}}{\sqrt{\w_{s,x}}} \\ \frac{\f_{x,y}}{\sqrt{\w_{x,y}}} \\ \frac{\f_{x,t}}{\sqrt{\w_{x,t}}}\\ \frac{\f_{y,t}}{\sqrt{\w_{y,t}}}
  \end{bmatrix} = B^{T+}({\sf e}_s - {\sf e}_t) = \begin{bmatrix}
   \frac{3}{4} & -\frac{1}{4} & -\frac{1}{4} & -\frac{1}{4}\\ \frac{1}{2} & \frac{1}{2} & -\frac{5}{6} & -\frac{1}{6} \\ \frac{1}{2} & \frac{1}{2} & -\frac{1}{6} & -\frac{5}{6} \\ 0 & 0 & \frac{2}{3} & -\frac{2}{3}
  \end{bmatrix}\begin{bmatrix}
   1 \\ 0 \\ 0 \\ -1
  \end{bmatrix} = \begin{bmatrix}
   1 \\ \frac{2}{3} \\ \frac{4}{3} \\ \frac{2}{3}
  \end{bmatrix}.
\end{align*}

This means that ${\cal R}_{s,t} = 1 + \frac{4}{9} + \frac{16}{9} +\frac{4}{9} = \frac{11}{3}$. By invoking \eq{matrix-pot}, where the matrix $\overline{B}$ and its Moore-Penrose inverse $\overline{B}^+$ are equal to
\begin{align}
  &\overline{B} = \begin{bmatrix}
   1 & -1 & 0\\ 0 & \frac{1}{2} & -\frac{1}{2}\\ 0 & \frac{1}{2} & 0 \\ 0 & 0 & \frac{1}{2}
  \end{bmatrix}, &\overline{B}^+ = \begin{bmatrix}
   1 & \frac{2}{3} & \frac{4}{3} & \frac{2}{3}\\ 0 & \frac{2}{3} & \frac{4}{3} & \frac{2}{3} \\ 0 & -\frac{2}{3} & \frac{2}{3} & \frac{4}{3}
  \end{bmatrix},
\end{align}

\noindent we obtain that the potential at each vertex is given by
\begin{equation}\label{eq:poten}
  \p = \begin{bmatrix} \overline{p} \\ 0  \end{bmatrix} = \begin{bmatrix} \overline{B}^{+}W\f \\ 0 \end{bmatrix} = \begin{bmatrix} \begin{bmatrix}
   1 & \frac{2}{3} & \frac{4}{3} & \frac{2}{3}\\ 0 & \frac{2}{3} & \frac{4}{3} & \frac{2}{3} \\ 0 & -\frac{2}{3} & \frac{2}{3} & \frac{4}{3}
  \end{bmatrix} \begin{bmatrix}
   1 \\ \frac{2}{3} \\ \frac{4}{3} \\ \frac{2}{3}
  \end{bmatrix}\\ 0 \end{bmatrix} = \begin{bmatrix}
   \frac{11}{3} \\ \frac{8}{3} \\ \frac{4}{3} \\ 0
  \end{bmatrix},
\end{equation}
meaning that the potential state $\ket{\p}$ is equal to
\begin{equation}
\begin{split}
  \ket{\p} = \sqrt{\frac{2}{{\cal R}_{s,t}}}\sum_{u \in V}\p_u\sqrt{\du_u}\ket{\psi_u} = \frac{11}{3}\ket{s,x} -\frac{8}{3}\ket{x,s} + \frac{4}{3}\ket{x,y}+ \frac{4}{3}\ket{x,t} - \frac{2}{3}\ket{y,x} + \frac{2}{3}\ket{y,t}.
\end{split}
\end{equation}

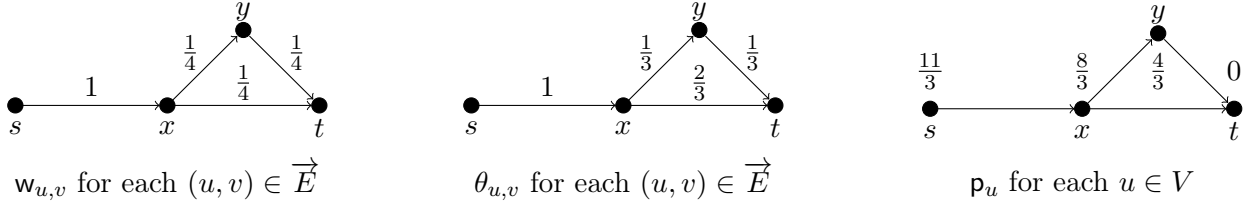
\begin{figure}
\centering
\begin{tikzpicture}
\node at (0,0) {\begin{tikzpicture}
  \draw[->] (-2,0)--(-0.1,0);
  \draw[->] (0,0)--(1.9,0);
  \draw[->] (1.08,.92)--(1.92,.08);
  \draw[->] (0,0)--(.92,.92);
  
  \filldraw (-2,0) circle (.1);
  \filldraw (0,0) circle (.1);
  \filldraw (2,0) circle (.1);
  \filldraw (1,1) circle (.1);
  
  \node at (-2,-0.3) {$s$};
  \node at (0,-0.3) {$x$};
  \node at (1,1.25) {$y$};
  \node at (2,-0.3) {$t$};

  \node at (-1,0.25) {$1$};
  \node at (1,0.3) {$\frac{1}{4}$};
  \node at (0.3,0.7) {$\frac{1}{4}$};
  \node at (1.7,0.7) {$\frac{1}{4}$};

  \node at (0,-1) {$\w_{u,v}$ for each $(u,v) \in \E$};
  \end{tikzpicture}};
  
  \node at (6,0) {\begin{tikzpicture}
  \draw[->] (-2,0)--(-0.1,0);
  \draw[->] (0,0)--(1.9,0);
  \draw[->] (1.08,.92)--(1.92,.08);
  \draw[->] (0,0)--(.92,.92);
  
  \filldraw (-2,0) circle (.1);
  \filldraw (0,0) circle (.1);
  \filldraw (2,0) circle (.1);
  \filldraw (1,1) circle (.1);
  
  \node at (-2,-0.3) {$s$};
  \node at (0,-0.3) {$x$};
  \node at (1,1.25) {$y$};
  \node at (2,-0.3) {$t$};

  \node at (-1,0.25) {$1$};
  \node at (1,0.3) {$\frac{2}{3}$};
  \node at (0.3,0.7) {$\frac{1}{3}$};
  \node at (1.7,0.7) {$\frac{1}{3}$};
  \node at (0,-1) {$\f_{u,v}$ for each $(u,v) \in \E$};
  \end{tikzpicture}};

  \node at (12,0){\begin{tikzpicture}
  \draw[->] (-2,0)--(-0.1,0);
  \draw[->] (0,0)--(1.9,0);
  \draw[->] (1.08,.92)--(1.92,.08);
  \draw[->] (0,0)--(.92,.92);
  
  \filldraw (-2,0) circle (.1);
  \filldraw (0,0) circle (.1);
  \filldraw (2,0) circle (.1);
  \filldraw (1,1) circle (.1);
  
  \node at (-2,-0.3) {$s$};
  \node at (0,-0.3) {$x$};
  \node at (1,1.25) {$y$};
  \node at (2,-0.3) {$t$};

  \node at (-2,0.5) {$\frac{11}{3}$};
  \node at (-0,0.5) {$\frac{8}{3}$};
  \node at (1,0.5) {$\frac{4}{3}$};
  \node at (2,0.5) {$0$};
  \node at (0,-1) {$\p_u$ for each $u \in V$};
  \end{tikzpicture}};
  
\end{tikzpicture}
\caption{Graph $G$ with its \st electrical flow $\f$ and corresponding potential $\p$ at each vertex.}\label{fig:normal-2}
\end{figure}

We now consider the case where the vertex $x \in V$ contains an additional alternative neighbourhood: let $\Psi_{\star}(x) = \{\ket{\psi_x},\ket{\psi_x^{\alt}}\}$ where
$$ \ket{\psi_x^{\textrm{alt}}} = \sqrt{\frac{2}{3}} (\frac{1}{2}\ket{s,x} -\ket{x,y} + \frac{1}{2}\ket{x,t}),$$ 
visualised in \fig{alternative-2}. By taking 
$$\sqrt{\du_x}\ket{\psi_{x,1}} = \sqrt{\frac{3}{2}} \sqrt{\frac{1}{2}} (-\ket{x,y} + \ket{x,t}) = \frac{1}{2}\sqrt{3}(-\ket{x,y} + \ket{x,t}),$$
we find that $\{\ket{\psi_{x,0}} = \ket{\psi_x},\ket{\psi_{x,1}}\}$ forms an orthonormal basis for $\Psi_{\star}(x)$. For this basis we find that the alternative incidence matrix $B_{\alt}$ of $G$ and $\Psi_{\star}$ and the Moore-Penrose inverse $B_{\alt}^{T+}$ of its transpose are equal to
\renewcommand*{\arraystretch}{1.5}
\begin{align}
  &B_{\alt} = \begin{bmatrix}
   1 & -1 & 0 & 0 & 0\\ 0 & \frac{1}{2} & -\frac{1}{2}\sqrt{3}& -\frac{1}{2} & 0 \\ 0 & \frac{1}{2} & \frac{1}{2}\sqrt{3} & 0 & -\frac{1}{2} \\ 0 & 0 & 0 & \frac{1}{2} & -\frac{1}{2}
  \end{bmatrix}, &B_{\alt}^{T+} = \begin{bmatrix}
   \frac{3}{4} & -\frac{1}{4} & 0 & -\frac{1}{4} & -\frac{1}{4}\\ \frac{1}{2} & \frac{1}{2} & -\frac{1}{3}\sqrt{3}& -\frac{1}{2} & -\frac{1}{2} \\ \frac{1}{2} & \frac{1}{2} & \frac{1}{3}\sqrt{3} & -\frac{1}{2} & -\frac{1}{2} \\ 0 & 0 & -\frac{1}{3}\sqrt{3} & 1 & -1
  \end{bmatrix}.
\end{align}

\noindent We can recover the electrical flow $\falt$ with respect to $\Psi_{\star}$ in \fig{alternative-2} using \thm{elecflow-alt} to derive 
\begin{align*}
  W\falt &= \begin{bmatrix}
  \frac{\falt_{s,x}}{\sqrt{\w_{s,x}}} \\ \frac{\falt_{x,y}}{\sqrt{\w_{x,y}}} \\ \frac{\falt_{x,t}}{\sqrt{\w_{x,t}}}\\ \frac{\falt_{y,t}}{\sqrt{\w_{y,t}}}
  \end{bmatrix} = B_{\alt}^{T+}({\sf e}_s - {\sf e}_t) =\begin{bmatrix}
   \frac{3}{4} & -\frac{1}{4} & 0 & -\frac{1}{4} & -\frac{1}{4}\\ \frac{1}{2} & \frac{1}{2} & -\frac{1}{3}\sqrt{3}& -\frac{1}{2} & -\frac{1}{2} \\ \frac{1}{2} & \frac{1}{2} & \frac{1}{3}\sqrt{3} & -\frac{1}{2} & -\frac{1}{2} \\ 0 & 0 & -\frac{1}{3}\sqrt{3} & 1 & -1
\end{bmatrix}\begin{bmatrix}
   1 \\ 0 \\ 0 \\ 0 \\ -1
  \end{bmatrix} = \begin{bmatrix}
   1 \\ 1 \\ 1 \\ 1
  \end{bmatrix}.
\end{align*}

This means that ${\cal R}^{\alt}_{s,t} = 1 + 1 + 1 +1 = 4$. By invoking \eq{matrix-pot-alt}, where the matrix $\overline{B_{\alt}}$ and its Moore-Penrose inverse $\overline{B}_{\alt}^+$ are equal to
\begin{align}
  &\overline{B_{\alt}} = \begin{bmatrix}
   1 & -1 & 0 & 0\\ 0 & \frac{1}{2} & -\frac{1}{2}\sqrt{3}& -\frac{1}{2}\\ 0 & \frac{1}{2} & \frac{1}{2}\sqrt{3} & 0\\ 0 & 0 & 0 & \frac{1}{2}
  \end{bmatrix}, &\overline{B_{\alt}}^+ = \begin{bmatrix}
   1 & 1& 1 & 1\\ 0 & 1 & 1 & 1 \\ 0 & -\frac{1}{2}\sqrt{3} & \frac{1}{2}\sqrt{3} & -\frac{1}{2}\sqrt{3} \\ 0 & 0 & 0 & 2
  \end{bmatrix},
\end{align}
we obtain that the alternative potential at each alternative neighbourhood is given by
\begin{equation}\label{eq:alterpotential}
  \palt = \begin{bmatrix} \overline{\palt} \\ 0  \end{bmatrix} = \begin{bmatrix} \overline{B_{\alt}}^{+}W \falt_{s,t} \\ 0 \end{bmatrix} = \begin{bmatrix} \begin{bmatrix}
   1 & 1& 1 & 1\\ 0 & 1 & 1 & 1 \\ 0 & -\frac{1}{3}\sqrt{3} & \frac{1}{3}\sqrt{3} & -\frac{1}{3}\sqrt{3} \\ 0 & 0 & 0 & 2
  \end{bmatrix} \begin{bmatrix}
   1 \\ 1 \\ 1 \\ 1
  \end{bmatrix}\\ 0 \end{bmatrix} = \begin{bmatrix}
   4 \\ 3 \\ -\frac{1}{3}\sqrt{3} \\ 2 \\ 0
  \end{bmatrix},
\end{equation}
meaning that the alternative potential state $\ket{\palt}$ is equal to
\begin{equation}
\begin{split}
  \ket{\palt} &= \sqrt{\frac{2}{{\cal R}^{\alt}_{s,t}}}\sum_{u \in V, i \in \{0,\dots,a_u-1\}}\palt_{(u,i)}\sqrt{\du_u}\ket{\psi_{(u,i)}} \\
  &= 4\ket{s,x} -3\ket{x,s} + \frac{3}{2}\ket{x,y}+ \frac{3}{2}\ket{x,t} + \frac{1}{2}\ket{x,y} - \frac{1}{2}\ket{x,t} - \ket{y,x} + \ket{y,t} \\
  &= 4\ket{s,x} -3\ket{x,s} + 2\ket{x,y}+ \ket{x,t} - \ket{y,x} + \ket{y,t} .
\end{split}
\end{equation}

\begin{figure}[h!]
\centering
\begin{tikzpicture}
\node at (0,0) {\begin{tikzpicture}
  \draw[->] (-2,0)--(-0.1,0);
  \draw[->] (0,0)--(1.9,0);
  \draw[->] (1.08,.92)--(1.92,.08);
  \draw[->] (0,0)--(.92,.92);
  
  \filldraw (-2,0) circle (.1);
  \filldraw[fill=blue] (0,0) circle (.1);
  \filldraw (2,0) circle (.1);
  \filldraw (1,1) circle (.1);
  
  \node at (-2,-0.3) {$s$};
  \node at (0,-0.3) {$x$};
  \node at (1,1.25) {$y$};
  \node at (2,-0.3) {$t$};

  \node at (-1,0.25) {$1$};
  \node at (1,0.3) {$\frac{1}{2}$};
  \node at (0.3,0.7) {$\frac{1}{2}$};
  \node at (1.7,0.7) {$\frac{1}{2}$};

  \node at (0,-1) {$\falt_{u,v}$ for each $(u,v) \in \E$};
  \end{tikzpicture}};
  
  \node at (6,0) {\begin{tikzpicture}
  \draw[->] (-2,0)--(-0.1,0);
  \draw[->] (0,0)--(1.9,0);
  \draw[->] (1.08,.92)--(1.92,.08);
  \draw[->] (0,0)--(.92,.92);
  
  \filldraw (-2,0) circle (.1);
  \filldraw[fill=blue] (0,0) circle (.1);
  \filldraw (2,0) circle (.1);
  \filldraw (1,1) circle (.1);
  
  \node at (-2,-0.3) {$s$};
  \node at (0,-0.3) {$x$};
  \node at (1,1.25) {$y$};
  \node at (2,-0.3) {$t$};

  \node at (-1,0.2) {$4$};
  \node at (1.65,0.6) {$2$};
  \node at (1,0.2) {$2$};
  \node at (0.25,0.6) {$4$};

  \node at (0,-1) {$\palt_{u,v}$ for each $(u,v) \in \E$};
  \end{tikzpicture}};  

  \node at (12,0) {\begin{tikzpicture}
  \draw[->] (0.0,0)--(-1.9,0);
  \draw[->] (2,0)--(0.1,0);
  \draw[->] (2,0)--(1.08,.92);
  \draw[->] (1,1)--(0.07,0.07);
  
  \filldraw (-2,0) circle (.1);
  \filldraw[fill=blue] (0,0) circle (.1);
  \filldraw (2,0) circle (.1);
  \filldraw (1,1) circle (.1);
  
  \node at (-2,-0.3) {$s$};
  \node at (0,-0.3) {$x$};
  \node at (1,1.25) {$y$};
  \node at (2,-0.3) {$t$};

  \node at (-1,0.2) {$3$};
  \node at (1.65,0.6) {$0$};
  \node at (1,0.2) {$0$};
  \node at (0.25,0.6) {$2$};
  
  \node at (0,-1) {$\palt_{u,v}$ for each $(v,u) \in \E$};
  \end{tikzpicture}}; 
  
\end{tikzpicture}
\caption{Graph $G$ where the blue vertex $x$ has an additional alternative neighbourhood. The the \st alternative electrical flow $\falt$ with respect to this extra alternative neighbourhood is displayed, as well as the corresponding potential vector $\palt$. }\label{fig:alternative-2}
\end{figure}
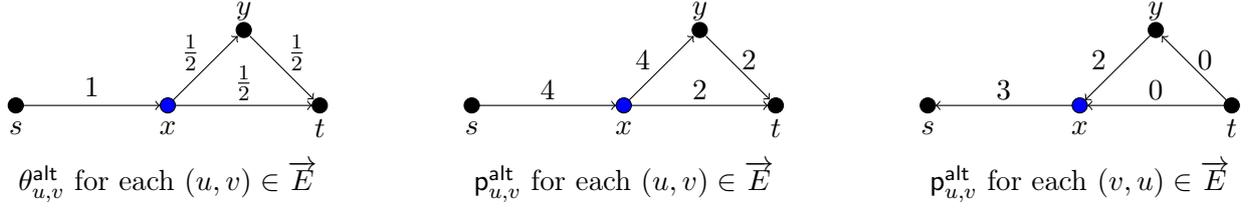

\section{Electrical flow sampling on one-dimensional random hierarchical graphs}\label{sec:1d}

Recently, \cite{balasubramanian2023exponential} have shown that there is an exponential separation between quantum and classical algorithms in finding a marked vertex in one-dimensional random hierarchical graphs, which is a generalization of the result of the welded tree problem \cite{childs2003ExpSpeedupQW}. In this section, we show that for one-dimensional random hierarchical graphs, we can efficiently generate a set of alternative neighbourhoods $\Psi_{\star}$ such that the resulting \st alternative electrical flow matches the \st electrical flow, meaning it satisfies Ohm's Law. We show that this allows us to invoke \lem{qwflows} with similar parameters as in \cor{qwflows}, allowing us to efficiently approximate the \st electrical flow and sample from it to find a marked vertex, recovering some of the results from \cite{balasubramanian2023exponential}. Note that we could have invoked \thm{qwflows} directly, but instead aim for this approach to provide more intuition on how our results generalise \cor{qwflows}.

Following \cite{balasubramanian2023exponential}, we now define the one-dimensional random hierarchical graph model with nodes $S_0,S_1,\dots, S_{n}$. 

\begin{definition}[Hierarchical graph on a line supergraph $\mathcal{G}$] A \emph{hierarchical graph on a line supergraph} $\mathcal{G}=(\mathcal{V} = \{0,\dots,n\},\mathcal{E})$ of length $n$ is defined by a set of nodes $S_{v}$ for each $v\in \mathcal{V}$ and a set of edges $ E_{u,v}$ for each $(u,v)\in \mathcal{E}$ such that $s_v=|S_v|$ and $e_{(u,v)}=|E_{u,v}|$. There are two special start and exit nodes $S_0 = \{s\}$ and $S_n = \{t\}$, meaning $s_{0}=s_{n}=1$.
Define $V=\bigcup_{v\in \mathcal{V}} S_v, E=\bigcup_{(u,v)\in \mathcal{E}} E_{u,v}$ and $G=(V,E)$. For each $(u,v)\in E(G)$, the edge set $E_{u,v}$ denotes the set of edges between the nodes between $S_u$ and $S_v$. 
\end{definition}

In \sec{weexample}, we discuss the welded tree graph, which is an example of a one-dimensional random hierarchical graph. 

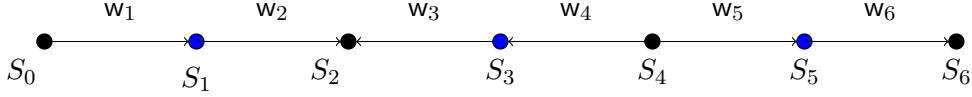
\begin{figure}
  \centering
  \begin{tikzpicture}
    \filldraw[fill] (-5.5,0) circle (.1); 	\node at (-5.8,-.4) {$S_0$};

	\draw[->] (-5.5,0) -- (-3.6,0); \node at (-4.5,0.4) {$\w_1$};

\filldraw[fill=blue] (-3.5,0) circle (.1); 
  \node at (-3.5,-0.5) {$S_1$};

  \draw[->] (-3.5,0) -- (-1.6,0); 	
   \node at (-2.5,0.4) {$\w_2$};

 \filldraw[fill] (-1.5,0) circle (.1); 	\node at (-1.8,-.4) {$S_2$};

 \draw[->] (0.5,0) -- (-1.4,0); 	
   \node at (-0.5,0.4) {$\w_3$};

\filldraw[fill=blue] (0.5,0) circle (.1);
 \node at (0.5,-0.4) {$ S_{3}$};

 \draw[->] (2.5,0) -- (0.6,0); 	
   \node at (1.5,0.4) {$\w_4$};
 
\filldraw[fill] (2.5,0) circle (.1);
 \node at (2.5,-0.4) {$ S_{4}$};

 \draw[->] (2.5,0) -- (4.4,0); 	
   \node at (3.5,0.4) {$\w_5$};
   
\filldraw[fill=blue] (4.5,0) circle (.1);
 \node at (4.5,-0.4) {$ S_{5}$};
 
 \draw[->] (4.5,0) -- (6.4,0); 	
   \node at (5.5,0.4) {$\w_6$};
 
\filldraw[fill] (6.5,0) circle (.1);
 \node at (6.5,-0.4) {$ S_{6}$};
	 
  \end{tikzpicture}
  \caption{A line supergraph $\mathcal{G}$ with nodes $S_0,S_1,\dots, S_{6}$. The black nodes are subsets of $V_{\mathrm{even}}$, where the edge directions are reversed and where all adjacent edges have the same weight and direction.}
  \label{fig:pathgraph}
\end{figure}

\begin{definition}[Balanced hierarchical graph] A hierarchical graph on a supergraph $G$ is said to be \emph{balanced} if for every $(u,v)\in E(G)$, the number of edges connecting a fixed node $\alpha \in S_u$ to nodes in $S_v$ is the same for each $\alpha$. 
\end{definition}

\begin{definition}[Edge-edge ratio]
Consider a hierarchical graph on the line supergraph $G$ which has nodes $S_0, S_1,\dots, S_{n}$ where each node $S_i$ contains $s_0, s_1,\dots, s_{n}$ many vertices. Let $e_k$ and $\mathcal{E}_k$ denote the number of edges and the set of edges between the nodes $S_{k-1}$ and $S_k$ respectively.  Then the \emph{edge ratios} $r_k$ for $k\in \{1,\dots, n-1\} $ are defined as 
\begin{equation*}
 r_k=\frac{e_{k+1}}{e_{k}}. 
\end{equation*}
\end{definition}

\begin{definition}[Edge-vertex ratio] A hierarchical graph on the line supergraph $G=(V,E)$ which has nodes $S_0,S_1,\dots,S_{n}$ possesses \emph{edge-vertex ratios} $\kappa_0,\kappa_1,\dots,\kappa_{n}$ given by 
\begin{equation}
  \kappa_j=\frac{e_j}{s_j}.
\end{equation}
\end{definition}

For a $D$-regular random balanced hierarchical graph on a line supergraph $G=(V,E)$, we have 
$e_i+e_{i+1}= e_i+r_ie_i=\kappa_is_i +r_i \kappa_is_i =Ds_i$ and $\kappa_i(1+r_i)=D.$ Let $\ell=\Theta(n)$ be an integer such that $2^{\ell} \gg |V|$, where $|V|$ is the number of vertices in the one-dimensional random hierarchical graph $G$. This does impose the restriction that $|V|$ can be at most exponential in $n$. To each vertex in $V$, we assign a random name from the set $\{0,1\}^{\ell}$. To access the neighbours of a particular vertex, we are given quantum access to an adjacency list oracle $O_G$ for the graph $G$. Given an $\ell$-bit string $\sigma \in \{0,1\}^{\ell}$ corresponding to a vertex $u \in V$, the adjacency list oracle $O_G$ provides the bit strings of the neighbouring vertices in $\N(u)$. If $\sigma$ does not correspond to any vertex, which will most often be the case than not since $2^{\ell} \gg |V|$, the oracle instead returns $\perp$. This oracle structure effectively forces any algorithm to start in $s$ and traverse the graph $G$ from there, as it is infeasible to try and guess the name of any other vertex in $V$. 

\begin{problemo}[One-dimensional random hierarchical graph problem]\label{probl:random}
We are given an adjacency list oracle $O_G$ to the one-dimensional random hierarchical graph $G$ ($D$-regular) on the line supergraph of length $n$ and the possibility to check whether any vertex $u$ is equal to $t$. Given the $\ell$-bit string associated to the starting vertex $s \in \{0,1\}^{\ell}$, the goal is to output the $\ell$-bit string corresponding to the other root $t$.
\end{problemo}

Before we can use \lem{qwflows} to tackle this problem, we must turn $G$ into an electrical network (see \defin{elec-network}), meaning we have to assign a weight and direction to each of its edges. We assign all edges in $E_{k} = \bigcup_{(u,v)\in \mathcal{E}_k} E_{u,v}$ the same weight for $k \in \{1,2,\dots, n\}$ and the weight $\w_k$ changes every two layers (starting at $\w_2$). Without loss of generality, we assume that $n$ is an even number and set $\w_1=1$ and

\begin{equation}\label{eq:weight-super}
\w_k= \prod_{i=1}^{\lfloor k/2 \rfloor} \left(\frac{1}{r_{2i-1}}\right)^2.
\end{equation} 

\noindent For each vertex $u\in S_i$ where $i\in\{1,\dots,n-1\}$, we find that 
$$\du_u=\sum_{v\in \Gamma(u)} \w_{u,v} = \kappa_i\w_i+(D-\kappa_i)\w_{i+1}.$$
We define the set of directed edges as follows:
\begin{equation}\label{eq:direct-super}
\E=\bigcup_{k\text{ mod } 4\in\{1,2\}}\{(u,v): u\in S_{k-1}, v\in S_k
 \} \cup \bigcup_{k\text{ mod } 4 \in\{0,3\}}\{(u,v): v\in S_{k-1}, u\in S_k\}.
\end{equation}
See \fig{pathgraph} for an example of a line supergraph where this edge orientation and weight assignments is visualised. By viewing $G$ as an electrical network, it is straightforward to directly compute the effective resistance ${\cal R}_{s,t}$ via the resistance laws for electrical circuits in series and parallel \cite{siebert1986circuits}. As a result we find for the weight assignment from \eq{weight-super} that
\begin{equation}\label{eq:eff-resis}
 {\cal R}_{s,t}= \frac{1}{D}+\sum_{k=2}^{n} \frac{1}{e_k\w_k}.
\end{equation}

Since one-dimensional random hierarchical graphs generalise the welded tree graph, it should come as no surprise that we will use a collection of alternative neighbourhood that generalises the one used in \cite{jeffery2023multidimensional} to traverse the welded tree graph:
\begin{definition}[Alternative Fourier Neighbourhood]\label{def:fourier}
  Let $G$ be a network. For any vertex $u\in V(G)$ with neighbours $\Gamma(u) = \{v_0,v_1,\dots v_{D-1}\}$. Let $\omega_D=\exp( 2\pi i/D)$ be the $D$-th root of unity. Then for each $j\in \{0,1,2, \dots, D-1 \}$, the $j$'th Fourier basis state is given by: 
  \[
  \ket{ \hat{\psi}^j_u} := \frac{1}{\sqrt{D}} \sum_{i=0}^{D-1} \omega_D^{i\cdot j}\ket{u,v_i} \]
  \noindent We define the \emph{alternative Fourier neighbourhood of dimension $D$} of the vertex $u$ as 
  $$\hat{\Psi}_\star(u)= \{\ket{\hat{\psi}^1_u},\ket{\hat{\psi}^1_u}, \dots \ket{\hat{\psi}^{D-1}_u}\}.$$

\end{definition}

Recall that we defined the weights in \eq{weight-super} and the edge directions in \eq{direct-super} in an alternating fashion. This induces a partition of $V$ into $V_{\mathrm{even}} = \bigcup_{v\in \mathcal{V}: v \text{ is even }} S_v$ and $V_{\mathrm{odd}} = \bigcup_{v\in \mathcal{V}: v \text{ is odd }} S_v$. We can assume without loss of generality that we know for any $u \in V$ whether it belongs to $V_{\mathrm{even}}$ or $V_{\mathrm{odd}}$ by keeping track of the parity of the distance from $s$ that is initially 0, and flips every time the algorithm takes a step. For a more detailed argument why this assumption is without loss of generality, we refer the reader to the end of Section 4 in \cite{jeffery2023multidimensional}. Note that at each vertex in $V_{\mathrm{odd}}$ the edge directions are reversed and all adjacent edges have the same weight (see \fig{pathgraph}). It is therefore straightforward to generate the star state $\ket{\psi_u}$ for each $u \in V_{\mathrm{odd}}$, since $\ket{\psi_u} \propto \ket{\hat{\psi}^0_u}$. For these vertices, we therefore do not consider any additional alternative neighbourhoods, meaning $\Psi_{\star}(u) = \{ \ket{\psi_u}\}$. For any $u \in V_{\mathrm{even}} \backslash \{t\}$, we let the set of alternative neighbourhoods be the alternative Fourier neighbourhood (see \defin{fourier}): $\Psi_{\star}(u) = \hat{\Psi}_\star(u).$ This means that for every $u \in V_{\mathrm{even}} \backslash \{t\}$, generating $\Psi_{\star}(u)$ does not depend on $u$ (apart from making the query to obtain its neighbours). This allows us to prove the following lemma:

\begin{lemma}\label{lem:generate-fourier}
The quantum walk operator $U_{\A\Balt}$ as defined in \eq{walk-alt} can be implemented in $O(1)$ queries to $O_G$ and $O(nD)$ elementary operations. 
\end{lemma}
\begin{proof}
  The unitary $U_{\A\Balt}$ consists of the two reflections $2\Pi_{\cal A} - 1$ and $2\Pi_{\Balt} - 1$. Since the former is (up to a sign difference) equal to the ${\sf SWAP}$ operator on two registers, each containing bit strings of length $\ell = \Theta(n)$, it can be implemented in $0$ queries and $O(n)$ elementary operations. The cost of implementing $2\Pi_{\Balt} - 1$ follows almost directly from the proof of Lemma 4.4 in \cite{jeffery2023multidimensional}, which proves the $D = 3$ case. By considering general $D$ in their proof, it still holds that we only need $O(1)$ queries to $O_G$ to apply $2\Pi_{\Balt} - 1$. The number of elementary operations needed in their proof is in the general case dominated by the cost of the following operation (needed to generate the state $\ket{\hat{\psi}^j_u}$), which for each $j \in \{0,\dots,D-1\}$ applies the map
  $$\ket{j}\left(\sum_{i=0}^{D-1} \omega_D^{i\cdot j}\ket{i,0}\right)\ket{v_0,v_1,\dots v_{D-1}} \mapsto \ket{j}\left(\sum_{i=0}^{D-1} \omega_D^{i\cdot j}\ket{i,v_i}\right)\ket{v_0,v_1,\dots v_{D-1}}.$$
  By conditioning on the value $i$, we can copy over the $i$'th value in the $\ket{v_0,v_1,\dots v_{D-1}}$ register, but this will require $O(nD)$ elementary operations. This far exceeds the complexity of implementing the Quantum Fourier Transform $F_D$, which requires $O(\log(D)\log\log(D))$ elementary operations \cite{hales2000improved}.
\end{proof}

We now show how to apply \lem{qwflows} with the quantum walk operator $U_{\A\Balt}$, where $\Balt = \mathrm{span}\{\mathrm{span}(\Psi_{\star}(u)) : u \in V \backslash \{s,t\}\}$. We choose the same parameters as in \cor{qwflows}. This means that for $\ket{\f}$ we choose the state corresponding to the \st electrical flow $\f$ on $G$, which sends one unit of flow from $s$ to $t$ by evenly distributing the one unit of flow available at each layer $S_{i}$ to the next layer $S_{i+1}$ for each layer $i \in \{0,\dots,n-1\}$. By \eq{flowstate} we obtain that 
\begin{equation}\label{eq:elecflow1d}
  \ket{\f}= \frac{1}{\sqrt{{2\cal R}_{s,t}}}\sum_{(u,v)\in \E} \frac{\f_{u,v}}{\sqrt{\w_{u,v}}} (\ket{u,v}+\ket{v,u}) = \frac{1}{\sqrt{{2\cal R}_{s,t}}}\sum_{k = 1}^{n}\sum_{(u,v)\in \E_k} (-1)^{\Delta_{u,v}}\frac{1}{e_k\sqrt{\w_k}} (\ket{u,v}+\ket{v,u}),
\end{equation} 
and it is straightforward to verify that $\ket{\f}$ is normalised using \eq{eff-resis}, confirming that $\f$ is in indeed the \st electrical flow. By \eq{rel-flow-pot} we know that for its corresponding potential vector $\p$ and with potential state $\ket{\p} \in {\cal B} \subseteq \Balt$ (see \eq{potstate}) we have
\begin{equation*}
  \ket{\psi_s^+} = \frac{1}{\sqrt{{\cal R}_{s,t} D}}\ket{\f} - (I-\Pi_{\cal A})\frac{1}{\sqrt{{\cal R}_{s,t} D}}\ket{\p}.
\end{equation*}

Hence we can apply \lem{qwflows} by choosing $\ket{\psi} = \ket{\psi_s^+}$, $\ket{\phi} = -\frac{1}{\sqrt{{\cal R}_{s,t} Ds}}\ket{\p}$ and $p = \frac{1}{{\cal R}_{s,t} D}$ for the remaining parameters, if we manage to show that $\Pi_{\Balt}\ket{\f} = 0$. We achieve this with the following claim.

\begin{claim}\label{clm:orthogonal-flow}
For any $u \in V$, define $\ket{\f_u}= (\ket{u}\bra{u}\otimes I)\ket{\f}$. If $u \in V_{\mathrm{even}}$, then $\ket{\f_u} \propto \ket{\hat{\psi}^0(u)}$. If $u \in V_{\mathrm{odd}}$, then $\ket{\f_u} \propto \sum_{v \in \Gamma(u)} \f_{u,v}\ket{u,v}$. As a consequence, for every $u \in V$ and $\ket{\psi_{\star}} \in \Psi_{\star}(u)$ the state $\ket{\f_u}$ satisfies $\brakett{\psi_{\star}}{\f_u} =0$.
\end{claim}
\begin{proof}
By construction of $\ket{\f}$ (see \eq{elecflow1d}), we see for any $u \in V$ that the state $\ket{\f_u}$ is equal to
$$\ket{\f_u}= \frac{1}{\sqrt{{2\cal R}_{s,t}}}\sum_{v \in \Gamma(u)}\frac{\f_{u,v}}{\sqrt{\w_{u,v}}} \ket{u,v}.$$

Let $k$ such that $u \in S_k$ and let $v_1,v_2,\dots, v_{l} \in \Gamma(u) \cap S_{k-1}$ be the neighbours of $u$ that lay in the node $S_{k-1}$ and similarly let $v_{l+1},\dots, v_{D}\in \Gamma(u) \cap S_{k+1}$ be the neighbours of $u$ that lay in the node $S_{k+1}$, where $l= D/(1+r_k)$, which is in fact an integer. This means that $\f_{u,v_i} = (-1)^{\Delta_{u,v_i}}/e_k$ for $i \in [l]$ and $\f_{u,v_i} = (-1)^{\Delta_{u,v_i}}/e_{k+1}$ for $i \in [D]\backslash[l]$. If $u\in V_{\mathrm{even}}$, then the weights (see \eq{weight-super}) satisfy 
$$\sqrt{\frac{\w_{k+1}}{\w_{k}}}=\frac{e_{k}}{e_{k+1}}=\frac{1}{r_{k}},$$
meaning
$$\frac{1}{e_k\sqrt{\w_k}}=\frac{1}{e_{k+1}\sqrt{\w_{k+1}}}.$$
Additionally, since $u \in V_{\mathrm{even}}$, it holds that $(-1)^{\Delta_{u,v_i}}(-1)^{\Delta_{u,v_j}} = -1$ for any $i \in [\ell]$ and $j \in [D]\backslash[\ell]$, meaning $\ket{\f_u} \propto \ket{\hat{\psi}^0(u)}$.
Since for $u\in V_{\mathrm{even}}$ we defined $\Psi_{\star}(u)$ to be the alternative Fourier neighbourhood (see \defin{fourier}) and the Fourier basis states form an orthonormal basis, it follows that $\brakett{\psi_{\star}}{\f_u} =0 $.

Now if instead $u\in V_{\mathrm{odd}}$, then we know that $\w_k=\w_{k+1}$ and $(-1)^{\Delta_{u,v_i}}(-1)^{\Delta_{u,v_j}} = 1$ for any $i \in [\ell]$ and $j \in [D]\backslash[\ell]$. So $\ket{\f_u} \propto \sum_{v \in \Gamma(u)} \f_{u,v}\ket{u,v}$. Since for $u\in V_{\mathrm{odd}}$ we defined $\Psi_{\star}(u) = \{\ket{\psi_u} = \ket{\hat{\psi}^0(u)}\}$, it follows by the conservation of the flow $\f$ that $\brakett{\psi_u}{\f_u} = \sum_{v \in \Gamma(u)}\f_{u,v} = 0$.
\end{proof}

Knowing that we can apply \lem{qwflows} for our multidimensional electrical network, we now show how to use this information to solve \probl{random}.

\subsection{The algorithm}
In this section, we provide a quantum algorithm that approximates the \st electrical flow state and samples from it to find the ending vertex $t \in V$ in a one-dimensional random hierarchical graph. As an example of such a one-dimensional random hierarchical graph, we then apply our algorithm to the welded tree graph.

\begin{algorithm}
\caption{Solving the one-dimensional random hierarchical graph problem}
\begin{algorithmic}\label{alg:1drandomvertexfinding}
\REQUIRE One-dimensional random hierarchical graph $G = (V,E)$ with adjacency list oracle $O_G$, the $\ell$-bit string corresponding to the starting vertex $s \in V$, a success probability parameter $\delta$.
\ENSURE The $\ell$-bit string corresponding to the ending vertex $t \in V$.
\begin{enumerate}
  \STATE Set $i = 1$, $T_1 = \Theta({\cal R}_{s,t} D)$ and $T_2 = \Theta({\cal R}_{s,t} D\w_n\log(1/\delta))$. 
  \STATE For $j = 0$ to $T_1$, run phase estimation on the multidimensional quantum walk operator $U_{\A\Balt}$ and state $\ket{\psi_s^+}$ to precision $O(\frac{\epsilon^2}{\sqrt{{\cal R}_{s,t} \du_s}\norm{\ket{\p}}})$, where $\epsilon=\frac{1}{2{\cal R}_{s,t} D\w_n}$, and measure the phase register. If the output is ``$0$'', return the resulting state $\ket{\f'}$ and immediately continue to Step $3$. 
  \STATE Measure $\ket{\f'}$ to obtain an outcome $\ket{u,v}$, representing the edge $(u,v) \in E$. Check if $u$ or $v$ is equal to $t$ and if this is the case, return the $\ell$-bit string corresponding to $t$. Otherwise, if $i < T_2$, increment $i$ by $1$ and return to Step $2$.
\end{enumerate}
\end{algorithmic}
\end{algorithm} 

\begin{theorem}\label{thm:randomgraph}
  Let $G$ be a $D$-regular one-dimensional random hierarchical graph on the line supergraph of length $n$ with edge ratios $r_0,\dots,r_{n-1}$. Let $\w_n= \prod_{k=1}^{\lfloor n/2 \rfloor} (\frac{1}{r_{2k-1}})^2$ and let each vertex in $G$ be identified by an $\ell$-bit string where $\ell = \Theta(n)$. Given access to an adjacency list oracle $O_G$ to the graph $G$, there exists a quantum algorithm that solves \cref{probl:random} with success probability $1 - O(\delta)$ with
 \begin{align*}
    &O\left(\|\ket{{\sf p}}\|{\cal R}_{s,t}^{4.5} D^{4.5} {\sf w}_n^3\log(1/\delta)\right) \text{ queries}&O\left(n\|\ket{{\sf p}}\|{\cal R}_{s,t}^{4.5} D^{5.5} {\sf w}_n^3\log(1/\delta)\right) \text{ time}
    \end{align*}
\end{theorem}
\begin{proof}
  The proof consists of a cost and success probability analysis of \cref{alg:1drandomvertexfinding}. By \cref{lem:qwflows}, each run of the phase estimation algorithm in Step $2$ succeeds with probability at least $\Theta\left(\frac{1}{{\cal R}_{s,t} D}\right)$. Hence, the probability that at least a single out of the $T_1= \Theta({\cal R}_{s,t} D)$ runs succeed is constant.

  Suppose that we had a perfect copy of $\ket{\theta}$, then after measuring it we would obtain an edge $(u,v)\in E$ containing the vertex $t$ with probability 
  $$\frac{1}{{\cal R}_{s,t}}\sum_{u \in \Gamma(t)} \frac{\f_{u,t}^2}{\w_{u,t}} = \frac{1}{{\cal R}_{s,t} D\w_n}.$$
  
  \noindent Instead, we have access to a state $\ket{\f'}$, which by \cref{lem:qwflows} satisfies 
  $$ \frac{1}{2}\norm{\proj{\f'} - \proj{\f}}_1 \leq \epsilon = \frac{1}{2{\cal R}_{s,t} D\w_n}.$$
  Hence by measuring $\ket{\f'}$, we obtain an edge $(u,v) \in E$ that contains the vertex $t$ with probability at least $\Theta\left(\frac{1}{{\cal R}_{s,t} D\w_n}\right)$. The probability that a single out of the at most $T_2= \Theta({\cal R}_{s,t} D\w_n\log(1/\delta))$ repetitions succeeds in returning the vertex $t$ is therefore at least
  $$1 - \left(1 - O\left(\frac{1}{{\cal R}_{s,t} D\w_n}\right)\right)^{T_2} \geq 1 - O(\delta).$$

  \noindent For the cost of Step $2$, each iteration of the phase estimation requires 
  $$O\left(\frac{\norm{\ket{\p}}\sqrt{{\cal R}_{s,t} D}}{\epsilon^2}\right) = O(\norm{\ket{\p}}{\cal R}_{s,t}^{2.5} D^{2.5} \w_n^2)$$
  calls to $U_{\A\Balt}$. By \cref{lem:generate-fourier}, each such call has a cost of $O(1)$ queries and $O(nD)$ elementary operations. Since we can set up the initial state $\ket{\psi_s}$ in the same cost and we run at most $T_1\cdot T_2$ iterations of phase estimation, we find that the total contribution of Step $2$ to the cost is 
\begin{align*}
    &O\left(\|\ket{{\sf p}}\|{\cal R}_{s,t}^{4.5} D^{4.5} {\sf w}_n^3\log(1/\delta)\right) \text{ queries}&O\left(n\|\ket{{\sf p}}\|{\cal R}_{s,t}^{4.5} D^{5.5} {\sf w}_n^3\log(1/\delta)\right) \text{ time}
\end{align*}

  For the cost of Step $3$, we must only verify whether $u$ or $v$ is equal to $t$, which can be done in zero queries and $O(\ell) = O(n)$ elementary operations. So the cost of Step $2$ dominates the total cost of the algorithm.
\end{proof}

\subsubsection{Welded tree problem}\label{sec:weexample}

As an example to show the power of this electrical flow sampling approach, we show that \alg{1drandomvertexfinding} can be used to solve the welded tree problem in polynomial time, thus achieving an exponential speedup compared to any classical algorithms, which was originally shown in \cite{childs2003ExpSpeedupQW}. 

A welded tree graph consists of two full binary trees of depth $h$ and contains $2^{h+2}-2$ vertices. See \fig{welded-tree} for an example of such a graph. The leaves of both trees are connected via two disjoint perfect matchings. This makes it a one-dimensional random hierarchical graph on the line supergraph of length $n=2h+1$. For each $k \in \{0,\dots,2h+1\}$, every node $S_k$ contains 
\begin{equation*}
  s_k=\left\{\begin{array}{ll}
  2^{k} & \mbox{if }k\in\{0,\dots,h\}\\
  2^{2h+1-k} & \mbox{if }k\in\{h+1,\dots,2h+1\},
\end{array}\right.
\end{equation*}
vertices, meaning that its edge ratios are equal to
\begin{equation*}
  r_k=\left\{\begin{array}{ll}
  2 & \mbox{if }k\in\{1,\dots,h\}\\
  \frac{1}{2}& \mbox{if }k\in\{h+1,\dots,2h+1\}.
\end{array}\right.
\end{equation*}

Since $V = 2^{h+2}-2$, we find that $\ell = 2h$ satisfies $2^{\ell} \gg |V|$, meaning each vertex is assigned a $2h$-bit string as an identifier.

\begin{problemo}[The welded tree problem]\label{probl:welded} Given an adjacency list oracle $O_G$ for the welded tree graph $G$ of depth $h$ and the $2h$-bit string associated to the starting vertex $s \in \{0,1\}^{2h}$, the goal is to output the $2h$-bit string associated to the other root $t$.
\end{problemo}

Before we apply \thm{randomgraph} to the welded tree graph, we first obtain a little more insight about its weights $\w_k$, where we assume without loss of generality that $h$ is odd:
\begin{equation}\label{eq:w-welded}
  \w_k=\left\{\begin{array}{ll}
  2^{-2\lfloor k/2\rfloor} & \mbox{if }k\in\{1,\dots,h+1\}\\
  2^{-2(h+1-\lfloor k/2\rfloor)} & \mbox{if }k\in\{h+2,\dots,2h+1\}.
\end{array}\right.
\end{equation}

\begin{theorem}\label{thm:welded}
  Given an adjacency list oracle $O_G$ to the welded tree graph $G$, there exists a quantum algorithm that solves \cref{probl:welded} with success probability $1 - O(\delta)$ and cost
  \begin{align*}
    &O\left(n^{5.5}\log(1/\delta)\right) \text{ queries}, &O\left(n^{6.5}\log(1/\delta)\right) \text{ time}
    \end{align*}
  
\end{theorem}
\begin{proof}
  The theorem can be derived by bounding the quantities ${\cal R}_{s,t}, D, \w_n$ and $\norm{\ket{\p}}$ in \cref{thm:randomgraph}. From \cref{eq:w-welded} we see that $\w_n = 1/2$. Additionally, the effective resistance from \cref{eq:eff-resis} can be computed to find that ${\cal R}_{s,t} =\Theta(n)$. Since $D = 3$ and $\p_s={\cal R}_{s,t}$ is the largest potential value, we only need to bound $\norm{\ket{\p}}$:
  $$\norm{\ket{\p}}^2 = \frac{2}{{\cal R}_{s,t}}\sum_{k=0}^{n}\sum_{u \in S_k} \p_u^2\du_u \leq {\cal R}_{s,t} \sum_{k=0}^{n} s_k \du_u=O(n^2).$$

\end{proof}

The result of \thm{randomgraph} is worse than the state of the art algorithm for the welded tree problem by \cite{jeffery2023multidimensional}, which has cost $O(n)$ queries and $O(n^2)$ time. By approximating the electrical flow, we infer much more information than is actually needed to recover the bit string associated to $t$, but it exemplifies how sampling from the electrical flow can provide an exponential speedup.

\begin{figure}
\centering
\begin{tikzpicture}

\filldraw (-5.5,0) circle (.1); 	\node at (-5.8,.4) {$s$};

	\draw[->] (-5.4,0) -- (-3.6,2); 	\node at (-4.5,1.35) {$\frac{1}{4}$};
	\draw[->] (-5.4,0) -- (-3.6,-2);	\node at (-4.5,-1.35) {$\frac{1}{4}$};

\filldraw[fill=blue] (-3.5,2) circle (.1); 
\filldraw[fill=blue] (-3.5,-2) circle (.1);

	\draw[<-](-2.1,3) -- (-3.4,2);	\node at (-2.75,2.85) {$\frac{1}{4}$};
	\draw[<-] (-2.1,1) --(-3.4,2);
	\draw[<-] (-2.1,-3)-- (-3.4,-2);
	\draw[<-] (-2.1,-1)-- (-3.4,-2) ;

\filldraw (-2,3) circle (.1);
\filldraw (-2,1) circle (.1);
\filldraw (-2,-1) circle (.1);
\filldraw (-2,-3) circle (.1);

	\draw[<-] (-1.9,3) -- (-1.1,3.75);		\node at (-1.6,3.65) {$\frac{1}{16}$};
	\draw[<-] (-1.9,3) -- (-1.1,2.25);
	 \draw[<-] (-1.9,1) -- (-1.43,1.5);
	 \draw[<-] (-1.9,1) -- (-1.43,.5);
	 \draw[<-] (-1.9,-1) -- (-1.43,-.5);
	 \draw[<-] (-1.9,-1) -- (-1.43,-1.5);
	\draw[<-] (-1.9,-3) -- (-1.1,-3.75);
	\draw[<-] (-1.9,-3) -- (-1.1,-2.25);

\filldraw[fill=blue] (-1,3.75) circle (.1);
\filldraw[fill=blue] (-1,2.25) circle (.1);
\filldraw[fill=blue] (-1,-2.25) circle (.1);
\filldraw[fill=blue] (-1,-3.75) circle (.1);

\draw[<-] (-.9,3.75) -- (.9,3.75);		\node at (0,4.2) {$\frac{1}{16}$};
\draw[<-] (-.9,3.75) -- (.9,-2.25);
\draw[<-] (-.9,2.25) -- (.9,2.25);
\draw[<-] (-.9,2.25) -- (.9,3.75);

\filldraw[fill=blue] (5.5,0) circle (.1); \node at (5.9,.4) {$t$};

	\draw[->] (5.4,0) -- (3.6,2);	\node at (4.5,1.35) {$1$};
	\draw[->] (5.4,0) -- (3.6,-2);

\filldraw (3.5,2) circle (.1); 
\filldraw (3.5,-2) circle (.1);

	\draw[<-] (3.4,2) -- (2.1,3);	\node at (2.75,2.85) {$\frac{1}{4}$};
	\draw[<-] (3.4,2) -- (2.1,1);
	\draw[<-] (3.4,-2) -- (2.1,-3);
	\draw[<-] (3.4,-2) -- (2.1,-1);

\filldraw[fill=blue] (2,3) circle (.1);
\filldraw[fill=blue] (2,1) circle (.1);
\filldraw[fill=blue] (2,-1) circle (.1);
\filldraw[fill=blue] (2,-3) circle (.1);

	\draw[<-](1.9,3) -- (1.1,3.75);	\node at (1.6,3.65) {$\frac{1}{4}$};
	\draw[<-] (1.9,3) -- (1.1,2.25);
	 \draw[<-] (1.9,1) -- (1.43,1.5);
	 \draw[<-] (1.9,1) -- (1.43,.5);
	 \draw[<-] (1.9,-1) -- (1.43,-.5);
	 \draw[<-] (1.9,-1) -- (1.43,-1.5);
	\draw[<-] (1.9,-3) -- (1.1,-3.75);
	\draw[<-] (1.9,-3) -- (1.1,-2.25);

\filldraw (1,3.75) circle (.1);
\filldraw (1,2.25) circle (.1);
\filldraw (1,-2.25) circle (.1);
\filldraw (1,-3.75) circle (.1);

\end{tikzpicture}
\caption{The welded tree graph with depth $h=3$: the black vertices are the vertices in $V_{\mathrm{even}}$, where the edge directions are reversed and where all adjacent edges have the same weight and direction.}\label{fig:welded-tree}
\end{figure}
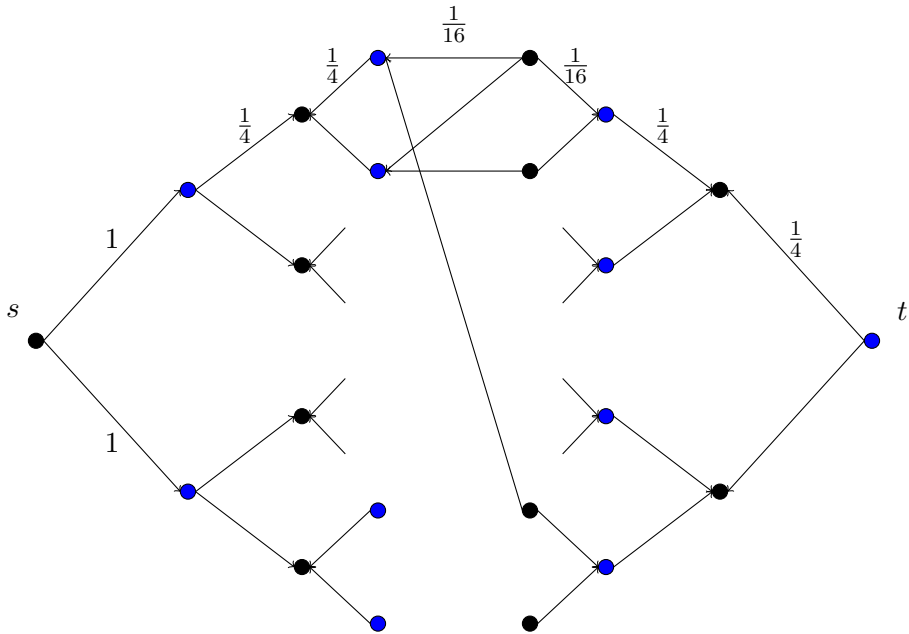

\section{An exponential speedup for pathfinding using alternative electrical flow sampling}\label{sec:path}

In this section, we show that the electrical flow in a multidimensional electrical network can also be used to show an exponential quantum-classical separation for the pathfinding problem relative to an oracle. We achieve this by constructing, and sampling from the \st \textit{alternative} electrical flow that we defined in \defin{flow-alt}, which is the flow achieving minimal energy out of all unit \st flows satisfying Alternative Kirchhoff's law, and we show that it also satisfies Alternative Ohm's Law through explicitly constructing the alternative potential $\palt$. In all of this section we assume that the parameter $n$ is odd for readability, but the everything can be slightly modified to also hold for even $n$.

\subsection{Example graph $G_1$} \label{sec:exaG1}

Since the graph that we will try to find an \st path for is quite large, we start by analysing the \st alternative flow and alternative potential for smaller graphs that will form the building blocks for the larger graph. We start with a network $G_1 = (V,E,\w)$, whose vertex set is given by $V = \{s,v_1,v_2,v_3,v_4,v_5,v_6,v_7,v_8,t\}$. We have visualised $G_1$, with its directed edge set and weights in \fig{pathgraph1}. These directions and weights give rise to the star states $\ket{\psi_u}$ for each $u \in V$, but we will also consider the following additional alternative neighbourhoods for the vertices $v_2,v_3,v_8 \in V$:
\begin{equation}\label{eq:G-star-alt}
\begin{split}
  \ket{\psi_{v_2}^{\textrm{alt}}} &= \sqrt{\frac{2}{3}} \left(-\ket{v_2,v_4} +\frac{1}{2}\ket{v_2,s} +\frac{1}{2}\ket{v_2,v_5}\right),\\
  \ket{\psi_{v_3}^{\textrm{alt}}} &= \sqrt{\frac{2}{3}} \left(\frac{1}{2}\ket{v_3,v_1} -\ket{v_3,v_6} + \frac{1}{2}\ket{v_3,v_7}\right),\\
  \ket{\psi_{v_8}^{\textrm{alt}}} &= \sqrt{\frac{2}{3}} \left(\frac{1}{2}\ket{v_8,t} -\ket{v_8,v_5} + \frac{1}{2}\ket{v_8,v_6}\right).
\end{split}
\end{equation}

Any \st alternative unit flow $\falt$ must be conserved at every vertex and satisfy $\falt_{s,v_1} = x$, $\falt_{s,v_2} = y$ for some $x,y \in [0,1]$ such that $x+y = 1$. For $\falt$ to also satisfy Alternative Kirchhoff's Law (see \defin{kcl-alt}), the flow coming into any vertex $v_2,v_3,v_8$ through the edge with the highest weight, must evenly be distributed along the other two neighbours. This is visualised in \fig{pathgraph1} and we end up with a single parameter $x$ (because $y = 1-x$) that parametrises all possible \st alternative unit flows $\falt$ on $G_1$. The energy of each such $\falt$ can be explicitly calculated to see that ${\cal E} (\falt)= 5y^2+4x^2+3$, and the energy is therefore minimised for $x= 5/9$, resulting in the alternative effective resistance to be ${\cal R}_{s,t}^{\alt} = 47/9$.

\begin{figure}[h!]
\centering
\begin{tikzpicture}

\filldraw (0,0) circle (.1);     \node at (0.2,-0.3) {$s$};
\filldraw[fill=blue] (2,0) circle (.1);          \node at (2,-0.3) {$v_2$};
\filldraw (4,0) circle (.1);     \node at (4.6,0) {$v_{5}$};
\filldraw[fill=blue] (0,4) circle (.1);          \node at (0,4.3) {$v_{3}$};
\filldraw (2,4) circle (.1);     \node at (2,4.3) {$v_{6}$};
\filldraw[fill=blue] (4,4) circle (.1);          \node at (4,4.3) {$v_{8}$};
\filldraw (6,4) circle (.1);  \node at (6,4.3) {$t$};

\filldraw (0,2) circle (.1);         \node at (0.4,2) {$v_{1}$};
\filldraw (2,2) circle (.1);    \node at (2.4,2) {$v_{4}$};
\filldraw (-1.6,2.1) circle (.1);  \node at (-2,2.6) {$v_{7}$};

\draw[->] (0.1,0)--(1.9,0);       \node at (1,0.3) {$1$ {\color{red} $x$}};
\draw[->] (0,0.1)--(0,1.9);       \node at (0.3,1) {$1$ {\color{red} $y$}};
\draw[->] (2.1,0)--(3.9,0);       \node at (3,0.3) {$\frac{1}{4}$ {\color{red} $\frac{x}{2}$}};
\draw[->] (2,0.1)--(2,1.9);       \node at (2.3,0.7) {$\frac{1}{4}$ {\color{red} $\frac{x}{2}$}};
\draw[->] (4,3.9)--(4,0.1);       \node at (4.3,2) {$\frac{1}{4}$ {\color{red} $\frac{1}{2}$}};
\draw[->] (0,2.1)--(0,3.9);       \node at (0.3,3) {$1$ {\color{red} $y$}};
\draw[->] (2,2.1)--(2,3.9);       \node at (2.3,3) {$\frac{1}{4}$ {\color{red} $\frac{x}{2}$}};
\draw[->] (-.05,3.95)--(-1.5,2.15);  \node at (-1,3.5) {$\frac{1}{4}$ {\color{red} $\frac{y}{2}$}};
\draw[->] (0.1,4)--(1.9,4);       \node at (1,4.3) {$\frac{1}{4}$ {\color{red} $\frac{y}{2}$}};
\draw[->] (3.9,4)--(2.1,4);       \node at (3,4.3) {$\frac{1}{4}$ {\color{red} $\frac{1}{2}$}};
\draw[->] (5.9,4)--(4.1,4);       \node at (5,4.3) {$1$ {\color{red} $1$}};
\draw[->] (-1.6,2) to[out=-120,in=-90] (3.95,-0.1); \node at (3.,-1.3) {$\frac{1}{4}$ {\color{red} $\frac{y}{2}$}};

\end{tikzpicture}
\caption{The graph $G_1$ with corresponding edge directions where the blue vertices have an additional alternative neighbour as defined in \eq{G-star-alt}. For each $(u,v) \in \E$, the weights $\w_{u,v}$ are denoted in black and the flow values $\falt_{u,v}$ in red for any valid unit \st alternative flow are parametrised by $x$ and $y = 1-x$.}\label{fig:pathgraph1}
\end{figure}
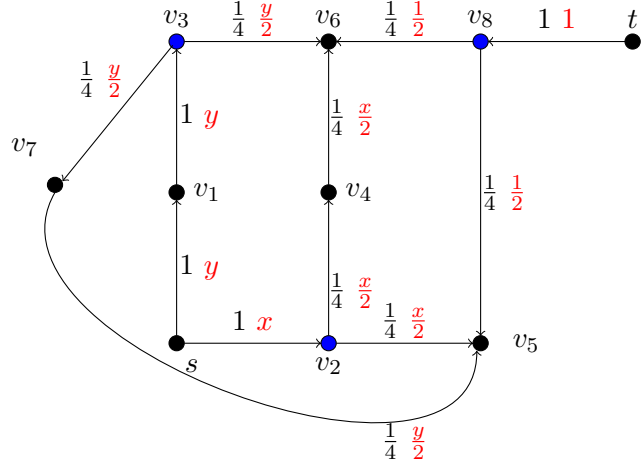

We now explicitly construct the alternative potential $\palt$ corresponding to this \st alternative electrical flow, that satisfies $\palt_{s,v_1} = \palt_{s,v_2} = {\cal R}_{s,t}^{\alt} = 47/9$, $\palt_{t,v_8} = 0$ and Alternative Ohm's Law (see \defin{ohm-alt}). We do this by constructing the states $\ket{\palt_{|u}} \in \mathrm{span}\{\Psi_\star(u)\}$ from \eq{pot-reduced}, where once again we add $\ket{\palt_{|s}}$ for completeness:

\begin{small}
\begin{align*}
  &\ket{\palt_{|s}} = \frac{47}{9} \ket{s,v_1} + \frac{47}{9} \ket{s,v_2}, &&\ket{\palt_{|v_1}} = - \frac{43}{9} \ket{v_1,s} + \frac{43}{9} \ket{v_1,v_3},\\
  &\ket{\palt_{|v_2}} = - \frac{42}{9} \ket{v_2,s} + \frac{38}{9}\sqrt{\frac{1}{4}} \ket{v_2,v_4} + \frac{46}{9}\sqrt{\frac{1}{4}}\ket{v_2,v_5}, 
  &&\ket{\palt_{|v_3}} = -\frac{39}{9}\ket{v_3,v_1} + \frac{52}{9} \sqrt{\frac{1}{4}}\ket{v_3,v_7} + \frac{26}{9} \sqrt{\frac{1}{4}}\ket{v_3, v_6},\\
   &\ket{\palt_{|v_4}} = - 2\sqrt{\frac{1}{4}} \ket{v,u} + 2\sqrt{\frac{1}{4}} \ket{v,t}, 
   &&\ket{\palt_{|v_5}} = - 4 \sqrt{\frac{1}{4}}\ket{v_5,v_8} - 4 \sqrt{\frac{1}{4}}\ket{v_5,v_7}-4\sqrt{\frac{1}{4}}\ket{v_5,v_2},\\
  &\ket{\palt_{|v_6}} = - 2\sqrt{\frac{1}{4}} \ket{v_6,v_8} - 2\sqrt{\frac{1}{4}} \ket{v_6,v_4}- 2 \sqrt{\frac{1}{4}} \ket{v_6,v_3}, &&\ket{\palt_{|v_7}} = - \frac{44}{9}\sqrt{\frac{1}{4}}\ket{v_7,v_3} + \frac{44}{9} \sqrt{\frac{1}{4}}\ket{v_7,v_5},\\
  &\ket{\palt_{|v_8}} = - \ket{v_8,t} + 0 \sqrt{\frac{1}{4}} \ket{v_8,v_6})+ 2\sqrt{\frac{1}{4}} \ket{v_8,v_5}, 
  &&\ket{\palt_{|t}} = 0 \ket{t,v_8}.
  \end{align*}
\end{small}

It is straightforward to verify that these states indeed satisfy Alternative Ohm's Law as well as the equations $\palt_{s,v_1} = \palt_{s,v_2} = 47/9$, $\palt_{t,v_8} = 0$. It is also clear that $\ket{\palt_{|u}} \in \mathrm{span}\{\Psi_\star(u)\}$ for every $u$ without additional alternative neighbourhoods, i.e. $u \in \{s,v_1,v_4,v_5,v_6,v_7,t\}$, since all edge potentials are the same. For $u \in \{v_2,v_3,v_8 \}$, we can confirm that $\ket{\palt_{|u}} \in \mathrm{span}\{\Psi_\star(u)\}$ by calculating that all the amplitudes of $\ket{\palt_{|u}}$ sum to $0$. 

\subsection{Example graph $G_2$}\label{sec:exG2}
The second example graph $G_2 = (V,E,\w)$ (see \fig{pathgraph2}) is build by combining the graph $G_1$ (see \fig{pathgraph1}) with three welded tree graphs $W_1,W_2,W_3$ (see \fig{welded-tree}). The ``starting'' root of these three welded tree graphs are $w_1, w_4$ and $w_6$ respectively. In the next section we will compose our final graph for the pathfinding example from multiple such $G_2$ graphs. 

As discussed in \sec{weexample}, the welded tree graph is an example of a one-dimensional random hierarchical graph with nodes $\{S_0, S_1,\cdots, S_n\}$. We additionally saw that for the weight assignments, edge directions and alternative neighbourhoods in \sec{weexample}, we ended up with an \st electrical flow that matched the \st alternative electrical flow, as it also satisfied Alternative Kirchhoff's Law. From the perspective of electrical networks, we can therefore interpret each welded tree graph $W_i$ as an edge of resistance ${\cal R}_i$, but we will formalise this intuition shortly. The weights and directions of $W_1$ in $G_2$ match those from \sec{weexample}, so ${\cal R}_1= {\cal R}$, where ${\cal R}$ is the effective resistance of a welded tree graph of depth $n$ (see \eq{eff-resis}). The weights of $W_2$ and $W_3$ have been multiplied by a factor of $1/4$, and their edge directions are reversed (because their respective roots are $w_4$ and $w_6$), so we have ${\cal R}_2={\cal R}_3 = 4{\cal R}$.

In $G_2$, the motivation for the alternative neighbourhoods, edge directions and weight assignments in the network $G_1$ become clear. Just like for the one-dimensional random hierarchical graphs in \sec{1d}, these assignments induces a partition of $V$ into $V_{\mathrm{even}}$ and $V_{\mathrm{odd}}$ (visualised by blue vertices in \fig{pathgraph2}). For each vertex $u \in V_{\mathrm{even}}$, all adjacent edges have the same weight and direction, allowing us to easily generate the star state $\ket{\psi_u}$. For each $u \in V_{\mathrm{odd}} \backslash \{s,t\}$, we have $\ket{\psi_u} \in \Psi_{\star}(u) = \hat{\Psi}_\star(u)$. Like in \sec{1d}, we can assume without loss of generality that we know for any $u \in V$ whether it belongs to $V_{\mathrm{even}}$ or $V_{\mathrm{odd}}$ by keeping track of the parity of the distance from $s$ that is initially 0, and flips every time the algorithm takes a step.

Since the welded tree graph sends through all flow coming into one root to the other, any \st alternative unit flow on $G_2$ is equivalent to a \st alternative unit flow on $G_1$, with the addition that we also have flow running through each welded tree graph. By \fig{pathgraph1} and \fig{pathgraph2} we therefore see that the energy of a \st alternative unit flow $\falt$ can be decomposed by the energy in $G_1$ in addition with the energy on these welded tree graphs and is hence given by
$${\cal E} (\falt)= 5y^2+4x^2+3 + {\cal R}_1y^2 + {\cal R}_2\left(\frac{x}{2}\right)^2+ {\cal R}_3\left(\frac{y}{2}\right)^2 = (5 + 2{\cal R})y^2+(4 + {\cal R})x^2+3.$$
This is minimised by taking $x = (5+ 2{\cal R})/(9 + 3{\cal R})$, meaning $y = 1-x = (4 + {\cal R})/(9 + 3{\cal R})$. For readability, we actually keep $x$ in the resulting alternative effective resistance, but simplify it slightly by making use of that for these values of $x$ and $y$ we have $(5+ 2{\cal R})y = (4 + {\cal R})x$:
$${\cal R}_{s,t}^{\alt} = (5 + 2{\cal R})y^2+(4 + {\cal R})x^2+3 = (4 + {\cal R})(x^2 + xy) + 3 = (4 + {\cal R})x + 3.$$

We now explicitly construct the alternative potential $\palt$ corresponding to this \st alternative electrical flow, that satisfies $\palt_{s,w_1} = \palt_{s,v_2} = {\cal R}_{s,t}^{\alt} = (4 + {\cal R})x + 3$, $\palt_{t,v_5} = 0$ and Alternative Ohm's Law. We do this by constructing the states $\ket{\palt_{|u}} \in \mathrm{span}\{\Psi_\star(u)\}$ from \eq{pot-reduced}. We slightly abuse notation however and only show the edges visible in \fig{pathgraph2}, meaning we will not explicitly write down the amplitudes and basis states for edges inside the welded tree graphs:

\begin{small}
\begin{align*}
  &\ket{\palt_{|s}} = (3+5y+2{\cal R}y) \ket{s,w_1} + (3+4x+{\cal R}x) \ket{s,v_2}, &&\ket{\palt_{|w_1}} = - (3+4y+2{\cal R}y) \ket{w_1,s},\\
  &\ket{\palt_{|w_2}} = (3+4y+{\cal R}y) \ket{w_2,v_1}, &&\ket{\palt_{|w_3}} = - (2+2x+2{\cal R}x)\sqrt{\frac{1}{4}} \ket{w_3,v_2},\\
  &\ket{\palt_{|w_4}} = (2+2x)\sqrt{\frac{1}{4}} \ket{w_4,v_3}, &&\ket{\palt_{|w_5}} = -(4+2y+2{\cal R}y) \ket{w_5,v_1},\\
  &\ket{\palt_{|w_6}} = (4+2y) \ket{w_6,v_4}, &&\ket{\palt_{|t}} = 0 \ket{t,v_5},
\end{align*}
\begin{align*}
  &\ket{\palt_{|v_1}} = -(3+3y+{\cal R}y)\ket{v_1,w_2} + (4+4y+2{\cal R}y) \sqrt{\frac{1}{4}}\ket{v_1,w_5} + (2+2y)\sqrt{\frac{1}{4}}\ket{v_1, v_3},\\
  &\ket{\palt_{|v_2}} = (2+4x+2{\cal R}x)\sqrt{\frac{1}{4}} \ket{v_2,w_3}+(4+2x)\sqrt{\frac{1}{4}}\ket{v_2,v_4}-(3+3x+{\cal R}x)\ket{v_2,s}, \\
  &\ket{\palt_{|v_3}} = - 2\sqrt{\frac{1}{4}}\ket{v_3,v_5} - 2\sqrt{\frac{1}{4}} \ket{v_3,w_4}- 2 \sqrt{\frac{1}{4}} \ket{v_3,v_1}, \\
  &\ket{\palt_{|v_4}} = - 4\sqrt{\frac{1}{4}}\ket{v_4,v_5} - 4 \sqrt{\frac{1}{4}}\ket{v_4,v_2}- 4 \sqrt{\frac{1}{4}}\ket{v_4,w_6},\\
  &\ket{\palt_{|v_5}} = - \ket{v_5,t} + 0 \sqrt{\frac{1}{4}} \ket{v_5,v_3})+ 2\sqrt{\frac{1}{4}} \ket{v_5,v_4}. \\  
\end{align*}
\end{small}

It is straightforward to verify that these states indeed satisfy Alternative Ohm's Law for all edges outside the welded tree graphs as well as the equations $\palt_{s,w_1} = \palt_{s,v_2} = {\cal R}_{s,t}^{\alt}$, since $(2 + 5{\cal R})y = (4 + {\cal R})x$ and that $\palt_{t,v_5} = 0$. It is also clear that $\ket{\palt_{|u}} \in \mathrm{span}\{\Psi_\star(u)\}$ for every $u \in \{s,v_3,v_4,t\}$, since all edge potentials are the same. For $u \in \{v_1,v_2,v_5\}$, we can confirm that $\ket{\palt_{|u}} \in \mathrm{span}\{\Psi_\star(u)\}$ by calculating that all the amplitudes of $\ket{\palt_{|u}}$ sum to $0$. For the edges in the welded tree graphs, we have seen in \sec{weexample} that the \st alternative electrical flow through each welded tree graph satisfies Ohm's Law. This means there exist potential values for all vertices (and hence edges), that are smaller than the potential at the root where the flows enters, in at each welded tree graph that satisfy Alternative Ohm's Law. These are consistent with our potential $\palt$ since
$$(\palt_{w_1,s} - \palt_{w_2,v_1})\frac{1}{y} = (\palt_{w_3,v_2} - \palt_{w_4,v_3})\frac{1}{x} = \left(\palt_{w_5,v_1} - \palt_{w_6,v_4}\right)\frac{1}{y} = {\cal R}.$$

Recall from the proof of \thm{welded} that for a welded tree graph of depth $n$ we have ${\cal R} = \Theta(n)$, meaning that ${\cal R}_{s,t}^{\alt} = \Theta(n)$. For the alternative potential, since for each edge potential we have $\palt_{u,v} = O(n)$, we find by \eq{pot-reduced} that
$$||\ket{\palt}||^2 = \frac{2}{{\cal R}_{s,t}^{\alt}}\sum_{(u,v) \in E} (\palt_{u,v})^2\w_{u,v} = O(n) \sum_{(u,v) \in E} \w_{u,v} = O(n^2).$$

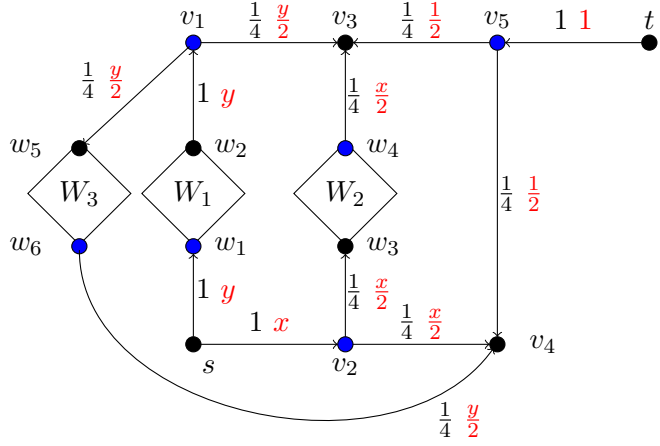
\begin{figure}
\centering
\begin{tikzpicture}

\filldraw (0,0) circle (.1);     \node at (0.2,-0.3) {$s$};
\filldraw[fill=blue] (2,0) circle (.1);          \node at (2,-0.3) {$v_{2}$};
\filldraw (4,0) circle (.1);     \node at (4.6,0) {$v_{4}$};
\filldraw[fill=blue] (0,4) circle (.1);          \node at (0,4.3) {$v_{1}$};
\filldraw (2,4) circle (.1);     \node at (2,4.3) {$v_{3}$};
\filldraw[fill=blue] (4,4) circle (.1);          \node at (4,4.3) {$v_{5}$};
\filldraw (6,4) circle (.1);  \node at (6,4.3) {$t$};

\node[diamond,
  draw = black,
  minimum width = 1cm,
  minimum height = 1.3cm] (d) at (0,2) {$W_1$};
\node[diamond,
  draw = black,
  minimum width = 1cm,
  minimum height = 1.3cm] (d) at (2,2) {$W_2$};
\node[diamond,
  draw = black,
  minimum width = 1cm,
  minimum height = 1.3cm] (d) at (-1.5,2) {$W_3$};

\filldraw[fill=blue](0,1.3) circle (.1);         \node at (0.5,1.3) {$w_{1}$};
\filldraw (2,1.3) circle (.1);    \node at (2.5,1.3) {$w_{3}$};
\filldraw (0,2.6) circle (.1);    \node at (0.5,2.6) {$w_{2}$};
\filldraw[fill=blue] (2,2.6) circle (.1);         \node at (2.5,2.6) {$w_{4}$};
\filldraw[fill=blue] (-1.5,1.3) circle (.1);        \node at (-2.2,1.3) {$w_{6}$};
\filldraw (-1.5,2.6) circle (.1);  \node at (-2.2,2.6) {$w_{5}$};

\draw[->] (0.1,0)--(1.9,0);       \node at (1,0.3) {$1$ {\color{red} $x$}};
\draw[->] (0,0.1)--(0,1.2);       \node at (0.3,0.7) {$1$ {\color{red} $y$}};
\draw[->] (2.1,0)--(3.9,0);       \node at (3,0.3) {$\frac{1}{4}$ {\color{red} $\frac{x}{2}$}};
\draw[->] (2,0.1)--(2,1.2);       \node at (2.3,0.7) {$\frac{1}{4}$ {\color{red} $\frac{x}{2}$}};
\draw[->] (4,3.9)--(4,0.1);       \node at (4.3,2) {$\frac{1}{4}$ {\color{red} $\frac{1}{2}$}};
\draw[->] (0,2.7)--(0,3.9);       \node at (0.3,3.3) {$1$ {\color{red} $y$}};
\draw[->] (2,2.7)--(2,3.9);       \node at (2.3,3.3) {$\frac{1}{4}$ {\color{red} $\frac{x}{2}$}};
\draw[->] (-.05,3.95)--(-1.45,2.65);  \node at (-1.2,3.5) {$\frac{1}{4}$ {\color{red} $\frac{y}{2}$}};
\draw[->] (0.1,4)--(1.9,4);       \node at (1,4.3) {$\frac{1}{4}$ {\color{red} $\frac{y}{2}$}};
\draw[->] (3.9,4)--(2.1,4);       \node at (3,4.3) {$\frac{1}{4}$ {\color{red} $\frac{1}{2}$}};
\draw[->] (5.9,4)--(4.1,4);       \node at (5,4.3) {$1$ {\color{red} $1$}};
\draw[->] (-1.5,1.25) to[out=-90,in=-125] (3.95,-0.05); \node at (3.5,-1) {$\frac{1}{4}$ {\color{red} $\frac{y}{2}$}};

\end{tikzpicture}
\caption{The graph $G_2$ with corresponding edge directions where the blue vertices are the vertices in $V_{\mathrm{odd}}$ and have the alternative neighbourhoods $\Psi_{\star}(u) = \hat{\Psi}_\star(u)$ (see \defin{fourier}). Each diamond, indexed by $i \in [3]$ represents a welded tree graph of depth $n$. For each $(u,v) \in \E$, the weights $\w_{u,v}$ are denoted in black and the flow values $\falt_{u,v}$ in red for any valid unit \st alternative flow are parametrised by $x$ and $y = 1-x$. The black vertices are the vertices in $V_{\mathrm{even}}$, where the edge directions are swapped and where adjacent edges have the same weight and direction.}\label{fig:pathgraph2}
\end{figure}

We can now invoke \thm{qwflows} to approximate the state $\ket{\falt}$. Since the energy along the \st path $(s,v_2),(v_2,v_4),(v_4,v_5),(v_5,t)$ contains a constant fraction of the energy ${\cal R}_{s,t}^{\alt}$, we could then sample from this state to recover a \st path. However, since this path is of constant length, any classical algorithm can also recover this path by an exhaustive search of its neighbours in constant time.

\subsection{The welded tree circuit graph $G$}\label{sec:pathfinding}
In this section, we construct a welded tree circuit graph $G$ by connecting $n$ graphs isomorphic to $G_2$ from \sec{exG2} (see \fig{pathgraph2}) as a path as indicated in \fig{graph} and define a pathfinding problem for this type of graph.

Each layer contains three welded tree graphs $W_1,W_2,W_3$ and the following $7$ vertices 
\begin{align*}
  V_{p,i} := \{v_{p,i,j}: j \in [7]\}.
\end{align*}
These layers are connected through the fact that $v_{p,i,7} = v_{p,(i+1),2}$ for every $i \in [n-1]$. The welded tree graphs structure is shown in \fig{welded-trees} for $j = 1$ (the edge directions are simply reversed for $j \in \{2,3\}$) and the weight assignment and edge directions for these welded tree graphs, as well as for the remaining edges, are the same as for the graph $G_2$ in \fig{pathgraph2}. The complete graph $G$ is shown in \fig{graph}. 
Due to this construction, each vertex have degree $3$ except for the vertices $s = v_{p,1,1}$ and $t = v_{p,n,7}$. It therefore induces the same partition of $V$ into $V_{\mathrm{even}}$ and $V_{\mathrm{odd}}$ as in $G_2$ (visualised by blue vertices in \fig{graph}). For each vertex $u \in V_{\mathrm{even}}$, all adjacent edges have the same weight and direction, allowing us to easily generate the star state $\ket{\psi_u}$. For each $u \in V_{\mathrm{odd}} \backslash \{s,t\}$, we have $\ket{\psi_u} \in \Psi_{\star}(u) = \hat{\Psi}_\star(u)$. 

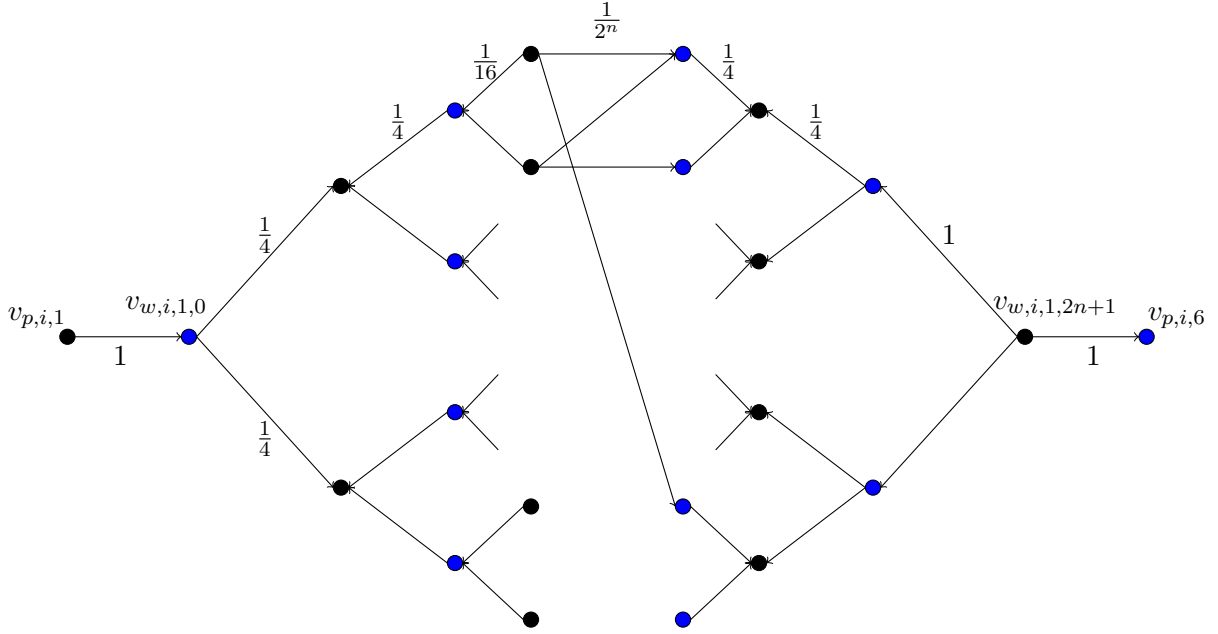
\begin{figure}
\centering
\begin{tikzpicture}
\filldraw (-7.1,0) circle (.1); 		\node at (-7.5,.25) {$v_{p,i,1}$};

	\draw[->] (-7,0)--(-5.6,0);	\node at (-6.4,-.25) {$1$};

\filldraw[fill=blue] (-5.5,0) circle (.1); 	\node at (-5.8,.4) {$v_{w,i,1,0}$};

	\draw[->] (-5.4,0) -- (-3.6,2); 	\node at (-4.5,1.35) {$\frac{1}{4}$};
	\draw[->] (-5.4,0) -- (-3.6,-2);	\node at (-4.5,-1.35) {$\frac{1}{4}$};

\filldraw (-3.5,2) circle (.1); 
\filldraw (-3.5,-2) circle (.1);

	\draw[<-] (-3.4,2) -- (-2.1,3);	\node at (-2.75,2.85) {$\frac{1}{4}$};
	\draw[<-] (-3.4,2) -- (-2.1,1);
	\draw[<-] (-3.4,-2) -- (-2.1,-3);
	\draw[<-] (-3.4,-2) -- (-2.1,-1);

\filldraw[fill=blue] (-2,3) circle (.1);
\filldraw[fill=blue] (-2,1) circle (.1);
\filldraw[fill=blue] (-2,-1) circle (.1);
\filldraw[fill=blue] (-2,-3) circle (.1);

	\draw[<-] (-1.9,3) -- (-1.1,3.75);		\node at (-1.6,3.65) {$\frac{1}{16}$};
	\draw[<-] (-1.9,3) -- (-1.1,2.25);
	 \draw[<-] (-1.9,1) -- (-1.43,1.5);
	 \draw[<-] (-1.9,1) -- (-1.43,.5);
	 \draw[<-] (-1.9,-1) -- (-1.43,-.5);
	 \draw[<-] (-1.9,-1) -- (-1.43,-1.5);
	\draw[<-] (-1.9,-3) -- (-1.1,-3.75);
	\draw[<-] (-1.9,-3) -- (-1.1,-2.25);

\filldraw (-1,3.75) circle (.1);
\filldraw (-1,2.25) circle (.1);
\filldraw (-1,-2.25) circle (.1);
\filldraw (-1,-3.75) circle (.1);

\draw[->] (-.9,3.75) -- (.9,3.75);		\node at (0,4.2) {$\frac{1}{2^{n}}$};
\draw[->] (-.9,3.75) -- (.9,-2.25);
\draw[->] (-.9,2.25) -- (.9,2.25);
\draw[->] (-.9,2.25) -- (.9,3.75);

\filldraw[fill=blue] (7.1,0) circle (.1); 		\node at (7.5,.25) {$v_{p,i,6}$};

	\draw[->] (5.6,0)--(7,0);	\node at (6.4,-.25) {$1$};

\filldraw (5.5,0) circle (.1); \node at (5.9,.4) {$v_{w,i,1,2n+1}$};

	\draw[->] (5.4,0) -- (3.6,2);	\node at (4.5,1.35) {$1$};
	\draw[->] (5.4,0) -- (3.6,-2);

\filldraw[fill=blue] (3.5,2) circle (.1); 
\filldraw[fill=blue] (3.5,-2) circle (.1);

	\draw[->] (3.4,2) -- (2.1,3);	\node at (2.75,2.85) {$\frac{1}{4}$};
	\draw[->] (3.4,2) -- (2.1,1);
	\draw[->] (3.4,-2) -- (2.1,-3);
	\draw[->] (3.4,-2) -- (2.1,-1);

\filldraw (2,3) circle (.1);
\filldraw (2,1) circle (.1);
\filldraw (2,-1) circle (.1);
\filldraw (2,-3) circle (.1);

	\draw[<-] (1.9,3) -- (1.1,3.75);	\node at (1.6,3.65) {$\frac{1}{4}$};
	\draw[<-] (1.9,3) -- (1.1,2.25);
	 \draw[<-] (1.9,1) -- (1.43,1.5);
	 \draw[<-] (1.9,1) -- (1.43,.5);
	 \draw[<-] (1.9,-1) -- (1.43,-.5);
	 \draw[<-] (1.9,-1) -- (1.43,-1.5);
	\draw[<-] (1.9,-3) -- (1.1,-3.75);
	\draw[<-] (1.9,-3) -- (1.1,-2.25);

\filldraw[fill=blue] (1,3.75) circle (.1);
\filldraw[fill=blue] (1,2.25) circle (.1);
\filldraw[fill=blue] (1,-2.25) circle (.1);
\filldraw[fill=blue] (1,-3.75) circle (.1);

\end{tikzpicture}
\caption{The $1$st welded tree graph in the $i$'th layer. For $j \in \{2,3\}$ the edge directions are simply reversed. The black vertices are the vertices in $V_{\mathrm{even}}$, where the edge directions are reversed and where adjacent edges have the same weight and direction.}\label{fig:welded-trees}
\end{figure}

\begin{figure}
\centering
\begin{tikzpicture}

\node at (2,5) {$i=1$};
\filldraw (0,0) circle (.1);     \node at (0.2,-0.3) {$s = v_{p,1,1}$};
\filldraw[fill=blue] (2,0) circle (.1);          \node at (2,-0.3) {$v_{p,1,2}$};
\filldraw (4,0) circle (.1);     \node at (4.6,0) {$v_{p,1,3}$};
\filldraw[fill=blue] (0,4) circle (.1);          \node at (0,4.3) {$v_{p,1,6}$};
\filldraw (2,4) circle (.1);     \node at (2,4.3) {$v_{p,1,5}$};
\filldraw[fill=blue] (4,4) circle (.1);          \node at (4,4.3) {$v_{p,1,4}$};

\node[diamond,
  draw = black,
  minimum width = 1cm,
  minimum height = 1.3cm] (d) at (0,2) {$W_1$};
\node[diamond,
  draw = black,
  minimum width = 1cm,
  minimum height = 1.3cm] (d) at (2,2) {$W_2$};
\node[diamond,
  draw = black,
  minimum width = 1cm,
  minimum height = 1.3cm] (d) at (-1.5,2) {$W_3$};

\filldraw[fill=blue] (0,1.3) circle (.1);         \node at (0.7,1.3) {$v_{w,1,1,0}$};
\filldraw (2,1.3) circle (.1);    \node at (3,1.3) {$v_{w,1,2,2n+1}$};
\filldraw (0,2.6) circle (.1);    \node at (1,2.6) {$v_{w,1,1,2n+1}$};
\filldraw[fill=blue] (2,2.6) circle (.1);         \node at (2.9,2.6) {$v_{w,1,2,0}$};
\filldraw[fill=blue] (-1.5,1.3) circle (.1);        \node at (-2.6,1.3) {$v_{w,1,3,0}$};
\filldraw (-1.5,2.6) circle (.1);  \node at (-2.6,2.6) {$v_{w,1,3,2n+1}$};

\draw[->] (0.1,0)--(1.9,0);       \node at (1,0.3) {$1$};
\draw[->] (0,0.1)--(0,1.2);       \node at (0.2,0.7) {$1$};
\draw[->] (2.1,0)--(3.9,0);       \node at (3,0.3) {$\frac{1}{4}$};
\draw[->] (2,0.1)--(2,1.2);       \node at (2.2,0.7) {$\frac{1}{4}$};
\draw[->] (4,3.9)--(4,0.1);       \node at (4.3,2) {$\frac{1}{4}$};
\draw[->] (0,2.7)--(0,3.9);       \node at (0.2,3.3) {$1$};
\draw[->] (2,2.7)--(2,3.9);       \node at (2.2,3.3) {$\frac{1}{4}$};
\draw[->] (-.05,3.95)--(-1.45,2.65);  \node at (-1,3.5) {$\frac{1}{4}$};
\draw[->] (0.1,4)--(1.9,4);       \node at (1,4.3) {$\frac{1}{4}$};
\draw[->] (3.9,4)--(2.1,4);       \node at (3,4.3) {$\frac{1}{4}$};
\draw[->] (5.9,4)--(4.1,4);       \node at (5,4.3) {$1$};
\draw[->] (-1.5,1.25) to[out=-90,in=-125] (3.95,-0.05); \node at (3.5,-1) {$\frac{1}{4}$};

\node at (8,5) {$i=2$};
\filldraw[fill=blue] (6,0) circle (.1);     \node at (5.7,-0.3) {$v_{p,2,6}$};
\filldraw (8,0) circle (.1);          \node at (8,-0.3) {$v_{p,2,5}$};
\filldraw[fill=blue] (10,0) circle (.1);     \node at (10.6,0.3) {$v_{p,2,4}$};
\filldraw (6,4) circle (.1);          \node at (6,4.3) {$v_{p,2,1}$};
\filldraw[fill=blue] (8,4) circle (.1);     \node at (8,4.3) {$v_{p,2,2}$};
\filldraw (10,4) circle (.1);          \node at (10,4.3) {$v_{p,2,3}$};

\node[diamond,
  draw = black,
  minimum width = 1cm,
  minimum height = 1.3cm] (d) at (6,2) {$W_1$};
\node[diamond,
  draw = black,
  minimum width = 1cm,
  minimum height = 1.3cm] (d) at (8,2) {$W_2$};
\node[diamond,
  draw = black,
  minimum width = 1cm,
  minimum height = 1.3cm] (d) at (11.5,2) {$W_3$};

\filldraw (6,1.3) circle (.1);         \node at (7,1.3) {$v_{w,2,1,2n+1}$};
\filldraw[fill=blue] (8,1.3) circle (.1);    \node at (8.8,1.3) {$v_{w,2,2,0}$};
\filldraw[fill=blue] (6,2.6) circle (.1);    \node at (6.7,2.6) {$v_{w,2,1,0}$};
\filldraw (8,2.6) circle (.1);         \node at (9,2.6) {$v_{w,2,2,2n+1}$};
\filldraw (11.5,1.3) circle (.1);        \node at (12.6,1.3) {$v_{w,2,3,2n+1}$};
\filldraw[fill=blue] (11.5,2.6) circle (.1);  \node at (12.3,2.6) {$v_{w,2,3,0}$};

\draw[->] (6.1,0)--(7.9,0);     \node at (7,0.3) {$\frac{1}{4}$};
\draw[->] (6,1.2)--(6,0.1);     \node at (6.2,0.7) {$1$};
\draw[->] (9.9,0)--(8.1,0);     \node at (9,0.3) {$\frac{1}{4}$};
\draw[->] (8,1.2)--(8,0.1);     \node at (8.3,0.7) {$\frac{1}{4}$};
\draw[->] (10,0.1)--(10,3.9);    \node at (10.3,2) {$\frac{1}{4}$};
\draw[->] (6,3.9)--(6,2.7);     \node at (6.2,3.3) {$1$};
\draw[->] (8,3.9)--(8,2.7);     \node at (8.2,3.3) {$\frac{1}{4}$};
\draw[->] (11.45,2.65)--(10.05,3.95);\node at (11,3.5) {$\frac{1}{4}$};
\draw[->] (6.1,4)--(7.9,4);     \node at (7,4.3) {$1$};
\draw[->] (8.1,4)--(9.9,4);     \node at (9,4.3) {$\frac{1}{4}$};
\draw[->] (11.9,0)--(10.1,0);    \node at (11.5,-0.3) {$1$};
\draw[->] (6,-0.1) to[out=-45,in=-90] (11.5,1.2); \node at (6.5,-1) {$\frac{1}{4}$};

\node at (12.5,0) {$\cdots$};

\node at (5,-2) {$i=n$};
\filldraw (3,-7) circle (.1);     \node at (3,-7.3) {$v_{p,n,1}$};
\filldraw[fill=blue] (5,-7) circle (.1);          \node at (5,-7.3) {$v_{p,n,2}$};
\filldraw (7,-7) circle (.1);     \node at (7.6,-7) {$v_{p,n,3}$};
\filldraw[fill=blue] (3,-3) circle (.1);          \node at (3,-2.7) {$v_{p,n,6}$};
\filldraw (5,-3) circle (.1);     \node at (5,-2.7) {$v_{p,n,5}$};
\filldraw[fill=blue] (7,-3) circle (.1);          \node at (7,-2.7) {$v_{p,n,4}$};
\filldraw (9,-3) circle (.1);          \node at (9.3,-2.7) {$t=v_{p,n,7}$};
\draw[->] (9,-3)--(7.1,-3);       \node at (8,-2.7) {$1$};

\node[diamond,
  draw = black,
  minimum width = 1cm,
  minimum height = 1.3cm] (d) at (3,-5) {$W_1$};
\node[diamond,
  draw = black,
  minimum width = 1cm,
  minimum height = 1.3cm] (d) at (5,-5) {$W_2$};
\node[diamond,
  draw = black,
  minimum width = 1cm,
  minimum height = 1.3cm] (d) at (1.5,-5) {$W_3$};

\filldraw[fill=blue] (3,-5.7) circle (.1);         \node at (3.7,-5.7) {$v_{w,n,1,0}$};
\filldraw (5,-5.7) circle (.1);    \node at (6,-5.7) {$v_{w,n,2,2n+1}$};
\filldraw (3,-4.4) circle (.1);    \node at (4,-4.4) {$v_{w,n,1,2n+1}$};
\filldraw[fill=blue] (5,-4.4) circle (.1);         \node at (5.9,-4.4) {$v_{w,n,2,0}$};
\filldraw[fill=blue] (1.5,-5.7) circle (.1);        \node at (0.4,-5.7) {$v_{w,n,3,0}$};
\filldraw (1.5,-4.4) circle (.1);  \node at (0.4,-4.4) {$v_{w,n,3,2n+1}$};

\draw[->] (3.1,-7)--(4.9,-7);       \node at (4,-6.7) {$1$};
\draw[->] (3,-6.9)--(3,-5.8);       \node at (3.2,-6.3) {$1$};
\draw[->] (5.1,-7)--(6.9,-7);       \node at (6,-6.7) {$\frac{1}{4}$};
\draw[->] (5,-6.9)--(5,-5.8);       \node at (5.2,-6.3) {$\frac{1}{4}$};
\draw[->] (7,-3.1)--(7,-6.9);       \node at (7.3,-5) {$\frac{1}{4}$};
\draw[->] (3,-4.3)--(3,-3.1);       \node at (3.2,-3.7) {$1$};
\draw[->] (5,-4.3)--(5,-3.1);       \node at (5.2,-3.7) {$\frac{1}{4}$};
\draw[->] (2.95,-3.05)--(1.55,-4.35);  \node at (2,-3.5) {$\frac{1}{4}$};
\draw[->] (3.1,-3)--(4.9,-3);       \node at (4,-2.7) {$\frac{1}{4}$};
\draw[->] (6.9,-3)--(5.1,-3);       \node at (6,-2.7) {$\frac{1}{4}$};
\draw[->] (2.9,-7)--(0.9,-7);       \node at (1.5,-7.3) {$1$};
\draw[->] (1.5,-5.75) to[out=-90,in=-125] (6.95,-7.05); \node at (6.5,-8) {$\frac{1}{4}$};

\node at (0.3,-7) {$\cdots$};

\end{tikzpicture}
\caption{The welded tree circuit graph $G$ showing all edge directions and edge weights. The blue vertices are the vertices in $V_{\mathrm{odd}}$ and have the alternative neighbourhoods $\Psi_{\star}(u) = \hat{\Psi}_\star(u)$ (see \defin{fourier}). The black vertices are the vertices in $V_{\mathrm{even}}$, where the edge directions are swapped and where adjacent edges have the same weight and direction. Each diamond, indexed by $j \in [3]$ represents the $j'$-th welded tree graph in that layer. See \fig{welded-trees} for a detailed overview of the welded tree graph's structure.}\label{fig:graph}
\end{figure}
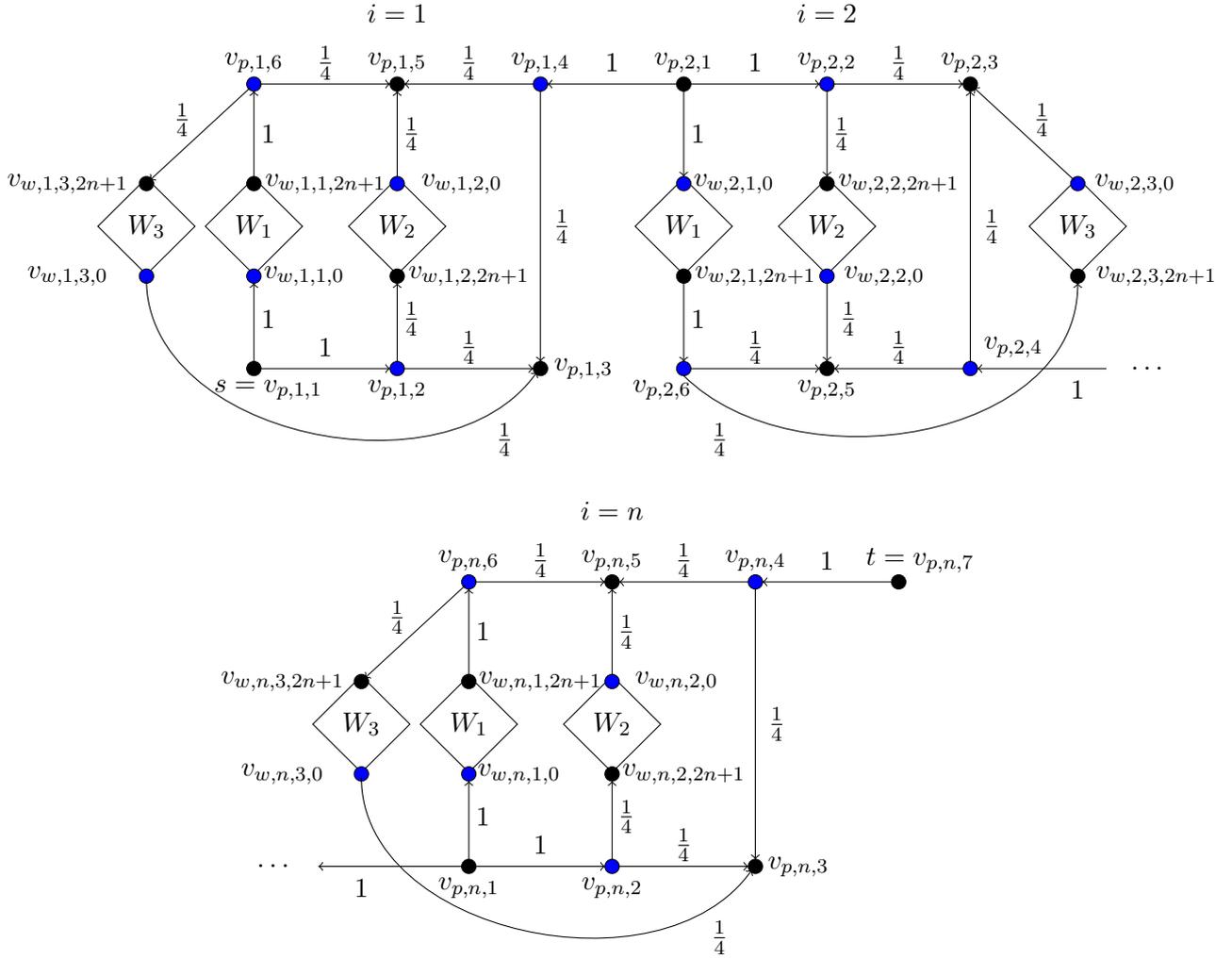

All these names $v_{a,b,c}$ to refer to vertices are simply for notation purposes to properly define the graph. Similar to the setting in \sec{1d}, we assign a random name from the set $\{0,1\}^{3n}$ to each vertex $u \in V$. To access the neighbours of a particular vertex, we are given quantum access to an adjacency list oracle $O_G$ for the graph $G$. Given an $3n$-bit string $\sigma \in \{0,1\}^{3n}$ corresponding to a vertex $u \in V$, the adjacency list oracle $O_G$ provides the bit strings of the neighbouring vertices in $\N(u)$. If $\sigma$ does not correspond to any vertex, which will most often be the case than not since $2^{3n} \gg |V|$, the oracle instead returns $\perp$. 

As the graph $G$ consists of $n$ identical subgraphs isomorphic to $G_2$, the flow and potential vector analysis almost directly follows from \sec{exG2}. Starting with the \st alternative electrical flow $\falt$, we can obtain this flow by simply connecting $n$ \st alternative electrical flows on each copy of $G_2$. This results in an alternative effective resistance ${\cal R}_{s,t}^{\alt} = \Theta(n^2)$. The alternative potential $\palt$ can also be obtained directly from combining all the alternative potentials from each copy of $G_2$, where we add $((4+{\cal R})x + 3)(n-i)$ to each edge potential obtained from the copy of $G_2$ in the $i$'th layer. This way we ensure that for every $i \in [n-1]$
\begin{align*}
  &\ket{\palt_{v_{p,i+1,1}}} = ((4+{\cal R})x + 3)(n-i)\ket{\psi_{v_{p,i+1,1}}},
\end{align*}
meaning $||\ket{\palt}|| = O(n^2).$ We now consider the following problem on the graph $G$, for which we exhibit a quantum algorithm that can solve the given problem exponentially faster than any classical algorithm can. 
\begin{problemo}[The pathfinding problem on a welded tree circuit graph] \label{probl:weldedGpath} Given an adjacency list oracle $O_G$ to the welded tree circuit graph $G$ (as defined in \sec{pathfinding}) and the names of the starting vertex $s=0^{3n}$, the goal is to output the names of vertices of an \st path.
\end{problemo}

\subsection{The algorithm}
In this section, we provide a quantum algorithm that can find the \st shortest path in the welded tree circuit graph $G$ and hence solves \probl{weldedGpath} in polynomial time.

\begin{algorithm}[h!]
\caption{Quantum algorithm for solving \probl{weldedGpath}}
\begin{algorithmic}\label{alg:path}
\REQUIRE Graph $G$ as defined in \sec{pathfinding}, the starting vertex $s = 0^{3n}$, a success probability parameter $\delta >0$.
\ENSURE The labels of an \st path on $G$.
\begin{enumerate}
  \STATE Set $i = 1$, $S = \emptyset$, $T_1 = \Theta(n^2)$ and $T_2 = \Theta(n^2\log(n/\delta))$. 
  \STATE For $j = 0$ to $T_1$, run phase estimation on the multidimensional quantum walk operator $U_{\A\Balt}$ and state $\ket{\psi_s^+}$ to precision $O(\epsilon^2/n^2)$, where $\epsilon=O(1/n^2)$, and measure the phase register. If the output is ``$0$'', return the resulting state $\ket{\f'}$ and immediately continue to Step $3$. 
  \STATE Measure $\ket{\f'}$ to obtain an outcome $\ket{u,v}$, representing the edge $(u,v) \in E$, and add it to $S$. If $i < T_2$, increment $i$ by $1$ and return to Step $2$.
  \STATE Search through $S$ using Breadth First Search for an \st path and output the path if it is found.
\end{enumerate}
\end{algorithmic}
\end{algorithm} 

\begin{theorem}\label{thm:weldedG}
 Given an adjacency list oracle $O_G$ to the welded tree circuit graph, there exists a quantum algorithm that solves \cref{probl:weldedGpath} with success probability $1 - O(\delta)$ and cost
  \begin{align*}
    &O\left(n^{11}\log(n/\delta)\right) \text{ queries}, &O\left(n^{12}\log(n/\delta)\right) \text{ time}.
    \end{align*}
\end{theorem}
\begin{proof}
  The proof consists of a cost and success probability analysis of \cref{alg:path}, where we focus on the success probability that the algorithm outputs the path
  \begin{equation*}
    {\cal P} = ((s,v_{p,1,2}),(v_{p,1,2},v_{p,1,3}),(v_{p,1,3},v_{p,1,4}),\dots,(v_{p,n,4},t)).
  \end{equation*}
  
  We apply \cref{thm:qwflows}, which states that each run of phase estimation in Step $2$ succeeds with a probability of at least $\Theta\left(\frac{1}{{\cal R}_{s,t}^{\alt}}\right) = \Theta\left(\frac{1}{n^2}\right)$. Hence, the probability that at least a single out of the $T_1= \Theta(n^2)$ runs succeed is constant.

  Suppose that we had a perfect copy of $\ket{\falt}$, then after measuring it we would obtain an edge $(u,v) \in {\cal P}$
  with probability at least
  $$\min\limits_{(u,v) \in {\cal P}}\frac{1}{{\cal R}_{s,t}^{\alt}}\frac{(\falt_{u,v})^2}{\w_{u,v}} = \Omega\left(\frac{1}{n^2}\right).$$
  
  \noindent Instead, we have access to a state $\ket{\f'}$, which by \cref{thm:qwflows} satisfies 
  $$ \frac{1}{2}\norm{\proj{\f'} - \proj{\f}}_1 \leq \epsilon = O\left(\frac{1}{n^2}\right).$$
  Hence by measuring $\ket{\f'}$, we obtain an edge $(u,v) \in E$ that contains the vertex $t$ with probability at least $\Omega\left(\frac{1}{n^2}\right)$. The probability that all edges in ${\cal P}$ are present in $S$ after reaching Step $4$ is due to the union bound therefore at least
  $$1 - \abs{{\cal P}}\left(1 - O\left(\frac{1}{n^2}\right)\right)^{T_2} \geq 1 - O(\delta).$$

  \noindent For the cost of Step $2$, each iteration of the phase estimation requires 
  $$O\left(\frac{\|\ket{\palt}\|\sqrt{{\cal R}_{s,t}^{\alt}}}{\epsilon^2}\right) = O(n^7)$$
  calls to $U_{\A\Balt}$. By \cref{lem:generate-fourier}, each such call has a cost of $O(1)$ queries and $O(n)$ elementary operations. Since we can set up the initial state $\ket{\psi_s}$ in the same cost and we run at most $T_1\cdot T_2$ iterations of phase estimation, we find that the total contribution of Step $2$ to the cost is 
  \begin{align*}
    &O\left(n^{11}\log(n/\delta)\right) \text{ queries}, &O\left(n^{12}\log(n/\delta)\right) \text{ time}.
    \end{align*}

  For the cost of Step $4$, we must only do a Breadth First Search to search for any \st path in the subgraph defined by the edges in $S$. Since identifying the vertex $s$ and $t$ can both be done using a single operation due to them having a distinct degrees, the total cost of this step is $O(T_2) = O(n^2\log(n/\delta))$ queries and other basic operations. So the cost of Step $2$ dominates the total cost of the algorithm.
\end{proof}

\subsection{Classical lower bound}
In this section we show that our \alg{path} actually provides an exponential speedup compared to any classical algorithm under the assumption that the following welded tree pathfinding problem is classically hard. To simplify the proof of our lower bound for the pathfinding problem \probl{weldedGpath}, we use the following assumption and the known classical lower bound of the welded tree problem.

\begin{problemo}[The welded tree pathfinding problem]\label{probl:weldedpath} Given an adjacency list oracle $O_G$ to the welded tree graph $G$ and the names of the starting vertex $s$ and the ending vertex $t$, the goal is to output the names of the vertices of an \st path.
\end{problemo}

It is folklore that the welded tree pathfinding problem is classically difficult, however, there is no formal statement as far as we are aware. 
\begin{assumpt}\label{assump:weldedpath} There exist constants $c_1 > 0$ and $c_2 \in (0,2)$ such that any classical algorithm that makes at most $2^{n/6}$ number of queries to $O_G$ to the welded tree graph $G$ solves \probl{weldedpath} with probability at most $c_1 \cdot 2^{-c_2n}$.
\end{assumpt}

\begin{lemma}[Theorem 9 in \cite{childs2003ExpSpeedupQW}] \label{lem:weldedlowerbound} 
For the welded tree problem \probl{welded}, any classical algorithm that makes at most $2^{n/6}$ queries to the oracle $O_G$ finds the ending vertex or a cycle with probability at most $4\cdot 2^{-n/6}$.
\end{lemma}

We follow the proof of the lower bound proof in \cite{li2023exponential}, which in turn is based on the lower bound proof in \cite{childs2003ExpSpeedupQW}. To prove the lower bound, we analyse the difficulty of any classical algorithm ${\cal A}$ winning a simpler game:

  \paragraph{Game A} Let $n$ be odd and let $G$ be the graph as defined in \sec{pathfinding}. Let Game A be the game where any classical algorithm $\mathcal{A}$ wins if it outputs the name of one of the vertex $v_{p,(n+1)/2,1}$, or if the vertices visited by $\mathcal{A}$ contain a cycle. Following \cite{childs2003ExpSpeedupQW}, the additional cycle condition that allows $\mathcal{A}$ to win in Game A allows us to analyse the success probability of ${\cal A}$ winning. This analysis involves determining whether a random embedding of a random rooted binary tree into the random graph $G$ contains a cycle or the vertex $v_{p,(n+1)/2,1}$. 
  
  Given the starting vertex $s$, the random embedding of a rooted binary tree $T$ into the graph $G$ is defined as a function $\pi$ from the vertices of $T$ to the vertices of $G$ such that $\pi(\textsf{ROOT}) = s$ and such that for any $(u,v) \in E$, we also have that $\pi(u),\pi(v)$ are neighbours in $T$. We say that an embedding $\pi$ is proper if $\pi(u)\neq \pi(v)$ for $u\neq v$. We say that $T$ exits under $\pi$ if $\pi(v) = v_{p,(n+1)/2,1}$. The random embedding can be obtained as follows: 
  \begin{enumerate}
    \item Set $\pi(\textsf{ROOT}) = s$.
    \item Let $i$ and $j$ be the two neighbours of $\textsf{ROOT}$ in $T$ and let $u$ and $v$ be the neighbours of $s$ in $G$. With probability $1/2$ set $\pi(i)= u$ and $\pi(j)=v$, and with probability $1/2$ set $\pi(i)= v$ and $\pi(j)=u$. 
    \item For any vertex $i$ in $T$, if $i$ is not a leaf and $\pi(v) \notin \{s,v_{p,(n+1)/2,1}\}$, let $j$ and $k$ denote the children of vertex $i$, and let $\ell$ denote its parent. Let $u$ and $v$ be the two neighbours of $\pi(i)$ in $G$ other than $\pi(\ell)$. With probability $1/2$ set $\pi(j)= u$ and $\pi(k)=v$, and with probability $1/2$ set $\pi(j)= v$ and $\pi(k)=u$.
  \end{enumerate}

\begin{theorem}\label{thm:lower-bound}
  Let $G$ be the graph defined in \sec{pathfinding}. Let $c_1,c_2$ be the constants from \assump{weldedpath} and assume that this assumption is true. Then any classical algorithm that makes at most $2^{n/6}$ queries to $O_G$ solves \probl{weldedGpath} with probability at most $(5+c_1)\cdot 2^{-\min\{c_2,1/6\}n}$.
\end{theorem}
 
  \begin{proof}
  Let $T$ be a random rooted binary tree with $2^{n/6}$ vertices and $\pi(T)$ be the image in the graph $G$ under the random embedding $\pi$. Given the name of the starting vertex $s$, similar to \cite{childs2003ExpSpeedupQW}, the probability of $\mathcal{A}$ winning Game $A$ can be expressed as the probability that $\pi(T)$ contains a cycle or the vertex $v_{p,(n+1)/2,1}$.

First, $\mathcal{A}$ has to enter a welded tree subgraph to find a cycle, as seen in \fig{graph}. There are two possibilities to get a cycle in a welded tree subgraph. One is to find a cycle that contains only one root in one of the welded tree subgraphs. In this case, \lem{weldedlowerbound} states that, in one of the welded tree subgraphs, starting from one root, any classical algorithm that makes at most $2^{n/6}$ queries to the oracle and finds the other root or a cycle with probability at most $ 4\cdot 2^{-n/6}$. The other is to find a cycle that contains two roots of a welded tree subgraph. By \assump{weldedpath}, any classical algorithm that makes at most $2^{n/6}$ queries to the oracle and finds such a cycle with probability at most $ c_1\cdot 2^{-c_2n}$.
  
  We can now assume that $\mathcal{A}$ will not encounter any cycle. Conditioned on this fact, the probability that $\mathcal{A}$ finds the name of the vertex $v_{p,(n+1)/2,1}$ can be expressed as the probability that $\pi(T)$ contains the vertex $v_{p,(n+1)/2,1}$, for which $\pi$ must follow the corresponding path $2n$ times, which has probability $2^{-2n}$. Since there are at most $2^{n/6}$ tries on each path of $T$ and there are at most $2^{n/6}$ paths, the probability of finding the name of the vertex $v_{p,(n+1)/2,1}$ is by the union bound at most $2^{n/3}2^{-2n} \leq 2^{-5n/3}$. We have the same result if the given name is $t$. Therefore, given the name of the starting vertex $s$ and $t$, the probability of $\mathcal{A}$ finding the vertex $v_{p,(n+1)/2,1}$ is $2\cdot 2^{n/3}2^{-2n} \leq 2^{-5n/3}$. 

  By combining the two cases with the union bound, we find that the probability of $\mathcal{A}$ winning Game $A$ is at most $2^{-5n/3} + (4+c_1)\cdot 2^{-\min\{c_2,1/6\}n} \leq (5+c_1)\cdot 2^{-\min\{c_2,1/6\}n}$. Since solving \probl{weldedGpath} automatically wins Game $A$, the theorem follows.
\end{proof}


\subsection*{Acknowledgements}
We thank Andrew Childs for referring the multidimensional quantum walk paper to JL as a flow approach, which led to this wonderful collaboration. We also thank anonymous reviewers for their very helpful and detailed comments. Part of the work was done while the authors were visiting QuICS, and JL was visiting Simons Institute, and we thank them for their hospitality. J.L. acknowledges funding from the National Science Foundation awards CCF-1618287, CNS-1617802, and CCF-1617710, and a Vannevar Bush Faculty Fellowship from the US Department of Defense. SZ is supported by an NWO Veni Innovational Research Grant under project number 639.021.752. JL was supported by NSF
Award FET-2243659 and NSF Career Award FET-2339116, Welch Foundation Grant no. A22-0307, a Microsoft Research Award, an Amazon Research Award, and a Ken Kennedy Research Cluster award (Rice K2I), in part through seed funding from the Ken Kennedy Institute at Rice University, the Rice CS Department, and the George R. Brown School of Engineering and Computing, by Rice University and the Department of Computer Science at Rice University, by the Ken Kennedy Institute and Rice Quantum Initiative, which is part of the Smalley-Curl Institute.

\bibliographystyle{quantum}
\bibliography{references}
 
\newpage
\begin{appendix}

\tocless\section{Proof of \lem{qwflows}}\label{app:qwflows}
Our analysis of the phase estimation algorithm, as in \cite{kitaev1996PhaseEst}, will use elements of the analyses in \cite{jeffery2023multidimensional} and \cite{piddock2019electricfind}, as well as the following lemma:
\begin{lemma}[Effective Spectral Gap Lemma \cite{lee2011QQueryCompStateConv}]\label{lem:effective-spectral-gap}
Fix $\epsilon\in [0,\pi)$, and let $\Lambda_\epsilon$ be the orthogonal projector onto the $e^{i\theta}$-eigenspaces of $U_{\cal AB}$ with $|\theta|\leq \epsilon$.
If $\ket{\phi} \in {\cal B}$, then 
$$\norm{\Lambda_\epsilon(I - \Pi_{\cal A})\ket{\phi}} \leq \frac{\epsilon}{2}\norm{\ket{\phi}}.$$
\end{lemma}

\begin{lemma}
  Define the unitary $U_{\cal AB} = (2\Pi_{\cal A} - 1)(2\Pi_{\cal B} - 1)$ acting on a Hilbert space ${\cal H}$ for projectors $\Pi_{\cal A},\Pi_{\cal B}$ onto some subspaces ${\cal A}$ and ${\cal B}$ of ${\cal H}$ respectively. Let $\ket{\psi} = \sqrt{p}\ket{\varphi} + (I - \Pi_{\cal A})\ket{\phi}$ be a normalised quantum state such that $U_{\cal AB}\ket{\varphi} = \ket{\varphi}$ and $\ket{\phi}$ is a (unnormalised) vector satisfying $\Pi_{\cal B}\ket{\phi} = \ket{\phi}$. Then performing phase estimation on the state $\ket{\psi}$ with operator $U_{\cal AB}$ and precision $\delta$ outputs ``$0$'' with probability $p' \in [p,p + \frac{17\pi^2\delta\norm{\ket{\phi}}}{16}]$, leaving a state $\ket{\psi'}$ satisfying
  $$ \frac{1}{2}\norm{\proj{\psi'} - \proj{\varphi}}_1 \leq \sqrt{\frac{17\pi^2\delta\norm{\ket{\phi}}}{16p}}. $$
  Consequently, when the precision is $O\left(\frac{p\epsilon^2}{\norm{\ket{\phi}}}\right)$, the resulting state $\ket{\psi'}$ satisfies
  $$ \frac{1}{2}\norm{\proj{\psi'} - \proj{\varphi}}_1 \leq \epsilon. $$
\end{lemma}
\begin{proof}

By the promise that $\ket{\psi} = \sqrt{p}\ket{\varphi} + (I - \Pi_{\cal A})\ket{\phi}$ with $\Pi_{\cal B}\ket{\phi} = \ket{\phi}$, we can apply \lem{effective-spectral-gap} to obtain
\begin{equation}\label{eq:proj-diff}
  \norm{\Lambda_\epsilon\left(\ket{\psi} - \sqrt{p}\ket{\varphi}\right)} = \norm{\Lambda_\epsilon(I - \Pi_{\cal A})\ket{\phi}} \leq \frac{\epsilon}{2}\norm{\ket{\phi}}.
\end{equation}

Let $\{\theta_j\}_{j\in J}\subset (-\pi,\pi]$ be the set of phases of $U_{\cal AB}$, and let $\Pi_j$ be the orthogonal projector onto the $e^{i\theta_j}$-eigenspace of $U_{\cal AB}$, so we can write
\begin{equation*}
  U_{\cal AB} = \sum_{j\in J} e^{i\theta_j}\Pi_j.
\end{equation*}

Phase estimation starts by making a superposition over $t$ from $0$ to $T-1$ in the phase register, where $T = 1/\delta$, and conditioned on this register we apply $U_{\cal AB}^t$ to $\ket{\psi}$, creating
\begin{equation*}
\begin{split}
  \sum_{t=0}^{T-1}\frac{1}{\sqrt{T}}\ket{t}U_{\cal AB}^t\ket{\psi} &= \sum_{j\in J}\sum_{t=0}^{T-1}\frac{1}{\sqrt{T}}\ket{t}e^{it\theta_j}\Pi_j\ket{\psi}.
\end{split}
\end{equation*}

The phase estimation algorithm then proceeds by applying an inverse Fourier transform, $F_T^\dagger$, to the first register, and then measuring the result. The probability $p'$ of measuring $0$ is
\begin{equation}
\begin{split}
p' &:= \norm{\bra{0}F_T^\dagger \otimes I\left( \sum_{j\in J} \sum_{t=0}^{T-1}\frac{1}{\sqrt{T}}\ket{t}e^{it\theta_j}\Pi_j\ket{\psi} \right)}^2
= \norm{\sum_{t=0}^{T-1}\frac{1}{\sqrt{T}}\bra{t}\otimes I \left( \sum_{j\in J} \sum_{t=0}^{T-1}\frac{1}{\sqrt{T}}\ket{t}e^{it\theta_j}\Pi_j\ket{\psi} \right)}^2\\
&= \frac{1}{T^2}\norm{\sum_{j\in J}\sum_{t=0}^{T-1}e^{it\theta_j}\Pi_j\ket{\psi}}^2
= \frac{1}{T^2}\sum_{j\in J:\theta_j\neq 0}\left| \frac{1-e^{i\theta_j T}}{1-e^{i\theta_j}} \right|^2 \norm{\Pi_j\ket{\psi}}^2 + \norm{\Lambda_0\ket{\psi}}^2\\
&=\frac{1}{T^2}\sum_{j\in J:\theta_j\neq 0}\frac{\sin^2(T\theta_j/2)}{\sin^2(\theta_j/2)}\norm{\Pi_j\ket{\psi}}^2+\norm{\Lambda_0\ket{\psi}}^2,
\end{split}\label{eq:phase-prob}
\end{equation}
since $\abs{\sum_{t=0}^{T-1}e^{i t\theta}} = \abs{\frac{1-e^{i\theta T}}{1-e^{i\theta}}}$,
and $|1-e^{i\theta}|^2=4\sin^2\frac{\theta}{2}$ for any $\theta\in\mathbb{R}$.

\noindent For the lower bound on $p'$, we know by \eq{phase-prob} that
\begin{equation*}
    \begin{split}
        p' &\geq \norm{\Lambda_0\ket{\psi}}^2 \geq |\brakett{\psi}{\varphi}|^2= p.
    \end{split}
\end{equation*}
We note that the second inequality in the above equation is in fact an equality, due to \eq{proj-diff}:
\begin{equation}
    \begin{split}
        \norm{\Lambda_{0}\ket{\psi}}^2 \leq \norm{\Lambda_{0}\left(\ket{\psi} - \sqrt{p}\ket{\varphi}\right)}^2 + \norm{\Lambda_{0}\sqrt{p}\ket{\varphi}}^2 =  \norm{\Lambda_{0}\sqrt{p}\ket{\varphi}}^2 \leq p.
    \end{split}\label{eq:bound-0space}
\end{equation}
For the upper bound, we continue where we left off in \eq{phase-prob}:
\begin{align}\label{eq:upper-p-total}
    p' &= \frac{1}{T^2}\sum_{j\in J:\theta_j \neq 0}\frac{\sin^2(T\theta_j/2)}{\sin^2(\theta_j/2)}\norm{\Pi_j\ket{\psi}}^2 
    + \norm{\Lambda_0\ket{\psi}}^2.
\end{align}
We split this sum into two parts, depending on whether the value of $\theta_j$ is smaller or larger than $\sqrt{1/(T\norm{\ket{\phi}})}$. In both parts, we make use of the identity $\sin^2\theta\leq \min\{1,\theta^2\}$ for all $\theta$. By inserting the same resolution of the identity as in \eq{bound-0space}, we upper bound the first sum as
\begin{equation}\label{eq:upper-p-1}
\begin{split}
    &\frac{1}{T^2}\sum_{j\in J:0 < \theta_j < \sqrt{1/(T\norm{\ket{\phi}})}}\frac{\sin^2(T\theta_j/2)}{\sin^2(\theta_j/2)}\norm{\Pi_j\ket{\psi}}^2\\
    &\leq \frac{1}{T^2}\sum_{j\in J:0 < \theta_j < \sqrt{1/(T\norm{\ket{\phi}})}}\frac{\sin^2(T\theta_j/2)}{\sin^2(\theta_j/2)}\left(\norm{\Pi_j\left(\ket{\psi} - \sqrt{p}\ket{\varphi}\right)}^2 
    + \norm{\Pi_j\sqrt{p}\ket{\varphi}}^2\right)\\
    &= \frac{1}{T^2}\sum_{j\in J:0 < \abs{\theta_j} < \sqrt{1/(T\norm{\ket{\phi}})}}\frac{\sin^2(T\theta_j/2)}{\sin^2(\theta_j/2)}\norm{\Pi_j\left(\ket{\psi} - \sqrt{p}\ket{\varphi}\right)}^2\\
    &\leq \frac{1}{T^2}\sum_{j\in J:0 < \abs{\theta_j} < \sqrt{1/(T\norm{\ket{\phi}})}}\frac{\pi^2T^2}{4}\frac{\norm{\ket{\phi}}}{4T} \hspace{5cm}\mbox{by \lem{effective-spectral-gap}} \\
    &\leq \frac{\pi^2\delta\norm{\ket{\phi}}}{16}.
\end{split}
\end{equation}
For the second sum, we additionally use the bound $\sin^2(\theta/2)\geq \theta^2/\pi^2$ whenever $|\theta|\leq \pi$:
\begin{equation}\label{eq:upper-p-2}
    \begin{split}
    &\frac{1}{T^2}\sum_{j\in J:\abs{\theta_j} \geq \sqrt{1/(T\norm{\ket{\phi}})}}\frac{\sin^2(T\theta_j/2)}{\sin^2(\theta_j/2)}\norm{\Pi_j\ket{\psi}}^2
    \leq \frac{1}{T^2}\sum_{j\in J:\abs{\theta_j} \geq \sqrt{1/(T\norm{\ket{\phi}})}}\frac{1}{(\theta_j/\pi)^2}\norm{\Pi_j\ket{\psi}}^2 \\
    &\leq \frac{\pi^2\norm{\ket{\phi}}}{T}\sum_{j\in J:\abs{\theta_j} \geq \sqrt{1/(T\norm{\ket{\phi}})}}\norm{\Pi_j\ket{\psi}}^2 \leq \pi^2\delta\norm{\ket{\phi}}.
\end{split}
\end{equation}

\noindent By substituting the upper bounds from \eq{bound-0space},\eq{upper-p-1} and \eq{upper-p-2} into \eq{upper-p-total}, we conclude that
\begin{align*}
    p' \leq \frac{\pi^2\delta\norm{\ket{\phi}}}{16} + \pi^2\delta\norm{\ket{\phi}} + p = \frac{17\pi^2\delta\norm{\ket{\phi}}}{16} + p.
\end{align*}

Finally, let $\ket{\psi'}$ be the (normalised) post measurement state after measuring $0$. We abbreviate ${\sf PE}$ for the phase estimation algorithm followed by the projection onto measuring $0$, as described in \eq{phase-prob}, such that $\ket{\psi'} = \frac{1}{\sqrt{p'}}{\sf PE}\ket{\psi}$. Note that since $\ket{\varphi}$ is an $1$-eigenvector of $U$, we have $\ket{\varphi} = {\sf PE}\ket{\varphi}$, meaning we can conclude the lemma via the inequality
\begin{equation*}
    \begin{split}
        \frac{1}{2}\norm{\proj{\hat{w}} - \proj{\varphi}}_1 &\leq \sqrt{1 -\abs{\brakett{\hat{w}}{\varphi}}^2} = \sqrt{1 -\frac{\abs{\bra{\psi}{\sf PE}\ket{\varphi}}^2}{p'}}\\
        & = \sqrt{1 -\frac{p}{p'}}\leq \sqrt{\frac{17\pi^2\delta\norm{\ket{\phi}}}{16p}}.
    \end{split}
\end{equation*}
\end{proof}

\end{appendix}

\end{document}

zursebastian@gmail.com